\documentclass[lettersize,journal]{IEEEtran}
\usepackage{amssymb}
\usepackage{lineno}
\usepackage{natbib}
\usepackage{subcaption}
\usepackage{multibib}
\usepackage{subfiles} 
\usepackage{marginnote,xcolor}
\usepackage[ruled,vlined]{algorithm2e}
\usepackage{graphicx}
\usepackage{amsthm}
\usepackage{picinpar}
\usepackage{amsmath}
\usepackage{url}
\usepackage{flushend}
\usepackage{float}
\usepackage[latin1]{inputenc}
\usepackage{colortbl}
\usepackage{amsfonts}
\usepackage{soul}
\usepackage{multirow}
\usepackage{pifont}
\usepackage{color}
\usepackage{alltt}
\usepackage[hidelinks]{hyperref}
\usepackage{enumitem}
\usepackage{siunitx}
\usepackage{breakurl}
\usepackage{algorithmic}
\usepackage{caption}
\usepackage{pbox}
\usepackage{comment}
\usepackage{multicol}
\newtheorem{theorem}{Theorem}

\newtheorem{lemma}{Lemma}
\newtheorem{definition}{Definition}
\newtheorem{assumption}{Assumption}
\newtheorem{remark}{Remark}

\newcommand{\bs}[1]{\boldsymbol{#1}}
\newcommand{\mc}[1]{\mathcal{#1}}
\newcommand{\mb}[1]{\mathbb{#1}}
\newcommand{\black}[1]{\textcolor{black}{#1}}

\begin{document}

\title{Classifying Transient Regimes in Dynamic Systems through Properties of Spatial Curves and Stochastic Processes: A Data-Driven Approach}
\author{%
    \IEEEauthorblockN{Cristian Puerto-Santana \IEEEauthorrefmark{1} \IEEEauthorrefmark{2}, Javier Diaz-Rozo\IEEEauthorrefmark{1}, Carlos Puerto-Santana\IEEEauthorrefmark{3}, and Carlos Ocampo-Martinez\IEEEauthorrefmark{2}}%

    \IEEEauthorblockA{\IEEEauthorrefmark{1} Aingura IIoT, 20009 San Sebastian, Spain }\\
    \IEEEauthorblockA{\IEEEauthorrefmark{2}Automatic Control Department, Universitat Polit\`ecnica de Catalunya-BarcelonaTECH, 08028 Barcelona, Spain}\\%
    \IEEEauthorblockA{\IEEEauthorrefmark{3} Independent researcher, 20003 San Sebastian, Spain}%
}


\markboth{}%
{}

\maketitle

\begin{abstract}
           \black{ This article proposes a novel methodology for the classification of transient and stationary regimes in dynamic systems. Several sensor-based solutions for regime classification in the literature require the setting of several parameters, or are not suitable for scenarios involving multivariate systems that may contain periodic signals. 
            The proposed method introduces a spatial curve representation of the considered system based on its sample mathematical moments. Then, by connecting concepts of stability theory, geometrical properties of spatial curves and stationary stochastic processes, two regime classifiers are designed using the arc length and the curvatures of the proposed curve. Both classifiers are capable of describing and detecting transient regimes, considering behaviors such as: multivariate asymptotically, marginally stability, and cyclostationarity. Furthermore, a quantitative comparison in performance and computation resources of the proposed classifiers against existing classifiers in the literature illustrates that the proposed regime classifier based on the arc length outperforms other techniques in classifying transient regimes for simulated linear, non-linear, and discontinuous multivariate systems under the specified studied conditions.}
    \end{abstract}




\section{Introduction}
\label{sec:Introduction}

\black{During decades, the effective characterization of  dynamical systems has relied on data gathered from physical sensors. Generally, such characterization consists of analyzing their response during transient and stationary regimes. On the one hand, the stationary time domains show the static relationships between the systems inputs and outputs, whereas the transient time periods reveal the dynamic ones. The proper classification and characterization of both regimes is the foundation of system identification techniques such as those reported in \cite{ljung2010perspectives, hagg2016transient, lataire2016transfer}, which are crucial for applications in fields such as control design and fault detection and diagnosis \cite{qin2000overview, estevez2011model}}.

An application of transient regime classifiers involves the planning and validation of prescribed controllers or event-based controllers. For prescribed control applications, control strategies were designed for dynamic systems to reach the steady state within predefined times \cite{yao2024novel, ye2022prescribed, hua2021adaptive}. Conversely, event-based controllers were designed to control systems in \cite{song2023prescribed, yan2018event, ji2024event}. These strategies aim to balance controller performance and resource efficiency. Event-driven scheduling mechanisms for these controllers rely on the residuals and information of the system under study.


\black{In \cite{liu2020dynamic}, an overview of the latest advances in the field of prescribed controllers was presented. Such a study  included definitions of stability and controllers for single and multiple input/output systems. The objective of these techniques was to achieve stability within a finite time interval. This work proposes various settling times for different types of dynamic systems based on a Lyapunov function. The design of these controllers depends on the previously estimated settling time. In \cite{zhang2021adaptive, chen5,  ge2020dynamic}, the latest advances in dynamic event triggered distributed coordination control were studied. The event-triggered scheduling mechanisms for these controllers depend on the residuals and information of the system under study.}
In data-driven system analysis, there is no knowledge of the model of the dynamic system under study, so new alternatives must be proposed to estimate such settling times. Similarly, in time-varying systems with a non-deterministic model, using data-driven techniques to estimate possible transient deviations of the system from a set point can enhance the general scheme of these controllers.

\black{Concerning industrial processes, \cite{puerto2022mechanical} proposed a methodology for unbalance monitoring in electric spindles that characterizes the gyroscopic effect of rotors during transient time periods at various rotational speeds. Furthermore, in the study of roto-dynamical systems, transient vibration during stationary rotational speeds are commonly associated with faulty operation of components such as bearings, gears, shafts, among others, as reported in \cite{zhang2016time, antoni2007fast, chen2013detecting}. Further applications regarding transient and stationary classification for the fields of communications, power systems and audio coding can be found in \cite{ureten2007wireless, markalous2008detection, zhang2011transient}.}

In summary, the characterization of dynamical systems is crucial for various applications. Stationary regimes aid in identifying abnormalities, optimizing, and controlling operations, while transient regimes enable understanding and controlling dynamics. This paper introduces a novel approach using statistical tools, signal processing techniques, and geometrical properties of spatial curves for classifying transient and stationary regimes in dynamical systems, with potential applications across science and engineering fields.

\subsection{Related work}
\label{subsec:related_work}

Recent algorithms in transient and stationary regime classification use Machine Learning, statistical analysis, and signal processing tools to analyze sensor data from dynamical systems. However, applying these methods to multivariable systems with oscillatory temporal responses remains challenging. Moreover, there is a need for transient and stationary regime classification algorithms that require minimal parameterization for online applications.

Considering the aforementioned challenges, Table \ref{tb:comparison_rw} presents relevant literature categorized by: \textit{(a)} classifying regimes for both periodic and non-oscillatory signals; \textit{(b)} working with multivariable dynamical systems; \textit{(c)} not requiring extensive parameter tuning for classification; and \textit{(d)} performing effectively with noisy measurements.

\begin{table}[!t]
	\caption{Classification of some literature methodologies for transient and stationary regime classification.}
	\label{tb:comparison_rw}
	\centering
	\small
    \scalebox{0.75}{
	\begin{tabular}{c p{6cm} c c c c}
		\hline
		Reference& Description & (a) & (b) & (c)& (d)\\ [0.5ex]
		\hline\hline
        \cite{shaw1997multifractal}& Multi-fractal dimension analysis. &\checkmark&-&-&\checkmark\\
        \cite{vasyutynskyyPASSIVEMONITORINGCONTROL2005}& Sample moving moments study.&-&-&\checkmark&-\\
        \cite{schladtSoftSensorsBased2007}& Sparsity measure of signals. &-&\checkmark&-&\checkmark\\
        \cite{markalousDetectionLocationPartial2008}& Wavelet transformation and spectrum energy. &-&\checkmark&\checkmark&-\\
        \cite{yaoBatchtoBatchSteadyState2009}& Principal component analysis. &-&\checkmark&\checkmark&\checkmark\\
		\cite{rhinehartAutomatedSteadyTransient2013}& Hypothesis test and R-test estimation. &-&\checkmark&-&\checkmark\\
        \cite{liu2020fast} & Median absolute variable measure. &-&-&\checkmark&-\\
        \cite{yu2022identification} & Differential signal and automatic noise band amplitude selection. &-&-&\checkmark&\checkmark\\[1ex]
		\hline
	\end{tabular}}
\end{table}

Traditional transient and stationary classifiers use features from data of sensors that monitor the studied system. For instance, in \cite{shaw1997multifractal}, a transient and stationary regimen classification algorithm was implemented using multi-fractal segmentation and neural network models. In \cite{vasyutynskyyPASSIVEMONITORINGCONTROL2005}, a data-driven method used moving statistical measures to identify regime changes in monitored system without signals with oscillatory responses. 

Recent studies, such as \cite{schladtSoftSensorsBased2007}, used sparsity measurements for transient and stationary regimen classification in industrial processes. In terms of signal processing techniques, \cite{markalousDetectionLocationPartial2008} developed a transient and stationary regimen classification method using signals processing tools such as the Wavelet transform. Regarding fault detection in power systems, \cite{liu2020fast} proposed a single-variable transient and stationary regimen classification algorithm using a proposed statistical measure.

Another transient and stationary regimen classification algorithm for multivariate industrial processes is proposed in \cite{yaoBatchtoBatchSteadyState2009}, the authors used principal component analysis and process information correlation for transient detection in batch processes. Besides, \cite{rhinehartAutomatedSteadyTransient2013} reported an automated transient and stationary regimen classification technique for processes, using statistical tools such as variance, and hypothesis testing. Finally, \cite{yu2022identification} introduced an algorithm for classifying transient regimes using differential signals and a dynamic noise band selection technique. The proposed methodology was compared with other classifiers in the literature with experimental data.

\subsection{Contribution}
\label{subsec:contribution}


\black{Many of the methods reviewed in Section~\ref{subsec:related_work} and summarized in Table~\ref{tb:comparison_rw} depend on heuristic thresholds and require the tuning of multiple parameters to operate effectively. Additionally, there is still a lack of algorithms well-suited for multivariate dynamical systems that exhibit cyclo-stationary behavior. Hence, this paper introduces a data-driven approach for developing two classifiers of transient and stationary regimes in monitored dynamical systems. The proposed method supports regime classification in multivariate systems whose signals may include periodic and noisy components, by combining signal processing techniques, statistical analysis, and curve geometry tools. Importantly, the implementation of this approach involves calibrating only two parameters, which are determined by the response time of the system and the characteristics of noise and periodicity observed during its stationary regime operation.}

\subsection{Outline}
The remainder of this document is organized as follows. The  definitions and fundamental concepts that are essential for the discussion of the proposed approach are provided in Section~\ref{sec:conceptual_framework}. The general scheme of the proposed approach is presented in Section~\ref{sec:Proposed_Approach}. Additionally, a detailed explanation of the methodology is shown in this section. The validation of the proposed approach for three simulated dynamical systems is shown in Section~\ref{sec:Validation}. Finally, Section~\ref{sec:Concluding_Remarks} shows a discussion of the proposed methodology based on a performed validation assessment. Moreover, the concluding remarks and following steps to enhance the proposed approach are drawn in this section. 

\subsection*{Notation}
\label{subec:notation}
Throughout this paper, parameters and scalar variables are represented in lower case; matrices and sets in upper case, and vectors in bold. Moreover, $t\in \mb{R}_{\geq 0}$ represents the continuous time domain, $n\in \mb{Z}_{\geq 0}$ represents the discrete time domain and $f\in \mb{R}_{\geq 0}$ represents the continuous frequency domain, where $\mb{R}_{\geq 0}\triangleq\{x \in \mb{R} \mid x \geqslant 0\}$ and $\mb{Z}_{\geq 0}\triangleq\{n \in \mb{Z} \mid n \geqslant 0\}$. Finally, $\mc{C}^q$ will be used to describe the set of functions which are $q$-differentiable, with $q \in \mb{N}$.

\section{Conceptual framework}
\label{sec:conceptual_framework}

In order to establish a solid theoretical framework for the proposed approach, this section presents an overview of the concepts and principles that underlie the methodology.

\subsection{Mathematical moments}
\label{subsec:mm}
\black{The $p$-th mathematical moment of a function $g(x): \mb{R}\rightarrow\mb{R}$ with the properties of  a density function shown in \cite{williams1991probability} around $x_0~\in \mb{R}$ is defined as}
\begin{equation}
    \label{eq:moments}
\beta_{(p,x_0)}=\int_{-\infty}^{\infty}(x-x_0)^pg(x)dx,   
\end{equation}
\noindent where $x_0$ is the value for which the moment is calculated. When $x_0=0$ the moment is commonly called as \emph{raw moment} and when $x_0=x_m$, where $x_m$ is the functions mean coordinates, then, the moment is known as \emph{central moment} \cite{papoulis2002probability}. The moments described by \eqref{eq:moments} can be interpreted as quantities that represent the geometrical shape of functions.

If $g(\cdot)$ represents a probabilistic density function in  $\mb{R}$, then the notation for raw moments is defined as  $\beta_{(p,0)} \triangleq E[(x)^p]$, while central moments are denoted as $\beta_{(p,x_m)}\triangleq E[(x-x_m)^p]$. \black{Additionally, standardized moments can be defined as $\tilde{\beta}_{(p,x_m)}\triangleq \frac{\mathbb{E}\left[x-x_m\right]^p}{\left(\mathbb{E}\left[x-x_m\right]^2\right)^{\frac{p}{2}}}$ for $p>2$. Such standardized moments are used to compare shape measures of a distribution independent of translation and scaling. Hence, if $g(x)$ in \eqref{eq:moments} is a probability density function, raw, central, and standardized moments can be employed to estimate expected value, variance, skewness, and kurtosis}.

Consider a stochastic process $x(n) \in \mathbb{R}$, where the observed values at discrete time steps $n$ are random variables. Assuming that the process is measured using a finite number of samples, the \emph{raw moments} are commonly known as \emph{sample moments} and can be estimated by the next expression:
\begin{equation}
    \label{eq:sample_moments}
\hat{\beta}_{(p,0)}(x,l)=\frac{1}{l}\sum_{n=1}^{l}(x(n))^p.
\end{equation}

\black{Moreover, if $x(n)$ is an ergodic stochastic process and the $p$-th moment exist, then, based on the Birkhoff Ergodic theorem $\hat{\beta}_{(p,0)}(x,l)=\beta_{(p,0)}$ as $l\rightarrow \infty$ \cite{papoulis2002probability, walters2000introduction}. In other words, as the number of data points increases, the sample moments of the stochastic process converge towards the true moments.}

\subsection{Stationary process}
\label{subsec:sp}

The definition of a strictly stationary stochastic process specifies that the probability distribution of the process remains constant over time intervals, without any changes in the parameters that shape the distribution, as described in \cite{park_2018}. However, a less stringent form of strict stationarity, known as weak-sense stationarity is shown in Definition~\ref{def:wss}. This definition of stationary stochastic process is often utilized for signal processing applications, as illustrated in \cite{mohammadi2022new}.
\begin{definition}
    \label{def:wss}
A stochastic process $\bs{w}(t)$ is said to be weak-sense stationary if the following conditions are met:
\begin{subequations}
    \label{eq:wss}
\begin{align}
    E[\bs{w}(t)]&=E[\bs{w}(t+\tau)], \tau \in \mb{R},\\
    K_{\bs{w} \bs{w}}(t_1,t_2)&=K_{\bs{w} \bs{w}}(t_2-t_1,0), t_1,t_2 \in \mb{R},\\
    E[\bs{w}^2(t)]&<\infty,
\end{align}
\end{subequations}

\noindent where $E[\bs{w}(t)]$ is the first raw moment, $K_{\bs{w} \bs{w}}(t_1,t_2)$ is the autocovariance function and $E[\bs{w}^2(t)]$ is the second raw moment. It is worth mentioning that  $K_{\bs{w} \bs{w}}(0,0)=E[(\bs{w}(t)-\bs{w}_m(t))^2]=E[(\bs{w}(t+\tau)-\bs{w}_m(t+\tau))^2], \bs{w}_m=E[\bs{w}(t+\tau)]~\forall~\tau$, therefore, \eqref{eq:wss} indicates that a process is considered  weak-sense stationary if the first and second centered-moments of the distribution are time-invariant and the autocovariance function only depends on lags $t_2-t_1$ \cite{gagniuc2017markov}.
\end{definition}

\begin{remark}    
\label{re:wsc}
For periodic processes, the concept of stationarity is known as wide-sense cyclo-stationary. A stochastic process $\bs{w}(t)$ is said to be wide-sense cyclo-stationary with period $t_r$ if
\begin{subequations}
    \label{eq:wcs}
    \begin{align}
        E[\bs{w}(t)]&=E[\bs{w}(t+t_r)],~\forall~t \in \mb{R},\\
        K_{\bs{w} \bs{w}}(t,t_2)&=K_{\bs{w} \bs{w}}(t+t_r,t_2),~\forall~t_2 \in \mb{R},
\end{align}
\end{subequations}
\noindent where $K_{\bs{w}\bs{w}}(t_r,0)=E[\bs{w}(t)]$. Therefore, the $1^{st}$ raw moment and $2^{nd}$ centered moment of the stochastic process are constant under time shifts given by $t_r$ \cite{gardner2006cyclostationarity}.
\end{remark}

\subsection{Geometrical properties of spatial curves}
Spatial curve representations are commonly employed in engineering to illustrate the trajectory or temporal evolution of dynamic systems. The geometrical properties of these curves can represent the behavior of the analyzed system. For instance, the arc-length, is a geometrical property that measures the distance between two points on a curve along its segment. This geometrical property is given in Definition~\ref{def:al}
\begin{definition}
    \label{def:al}
Let $\bs{g}:[t_0,t_f] \rightarrow \mb{R}^q$ be a one to one and differentiable function. The arc-length $l_r$ of $\bs{g}$ can be written as ${l_r}(t_0,t_f,\bs{g})=\int_{t_0}^{t_f}\lVert\bs{g}^{(1)}(t)\rVert dt$, where $\bs{g}^{(1)}(t)$, represents the first time-derivative of $\bs{g}(t)$ \cite{carmo}.
\end{definition}
Along with the arc length, the behavior of a spatial curve can be inferred by its curvatures, which are described through a mathematical model known as the Frenet-Serret formulas \cite{wheelerFrenetSerretFormulaeHigher}. These formulas were originally formulated to represent a system of differential equations that describe the motion and kinematic properties of a particle that moves along a $\mc{C}^3$ spatial curve $\bs{r}(t) \in \mb{R}^3$. The curvature\footnote{The curvature measures the divergence of a curve to be a straight line.} $k_1(t)$ and torsion\footnote{The torsion measures the likelihood of a curve to be fitted into a plane.} $k_2(t)$ of $\bs{r}(t)$ are the characteristic parameters of the Frenet-Serret formulas for $\mb{R}^3$ curves and can be analytically calculated, as illustrated in \cite{carmo}, using the following expressions:

\begin{equation}
    \label{eq:kur}
    \resizebox{0.99\linewidth}{!}{$
    k_1(t) = \frac{\left\lVert\bs{r}^{(1)}(t) \times \bs{r}^{(2)}(t)\right\rVert}{\left\lVert \bs{r}^{(1)}(t) \right\rVert^3},~
    k_2(t) = \frac{\langle \bs{r}^{(1)}(t) \times \bs{r}^{(2)}(t),\bs{r}^{(3)}(t)\rangle}{\left\lVert\bs{r}^{(1)}(t) \times \bs{r}^{(2)}(t)\right\rVert^2}.$}
\end{equation}

The subsequent sections of this document illustrate the proposed methodology for detecting transient and stationary regimes in dynamic systems, building upon the aforementioned concepts of mathematical moments, stationary processes, and geometrical properties of spatial curves.

\section{Proposed approach}
\label{sec:Proposed_Approach}

The proposed methodology of this article for stationary  and transient classification is summarized in Figure~\ref{fig:SMTR-GPD}.
\begin{figure*}[!t]
	\centering
	\includegraphics[width=0.8\hsize]{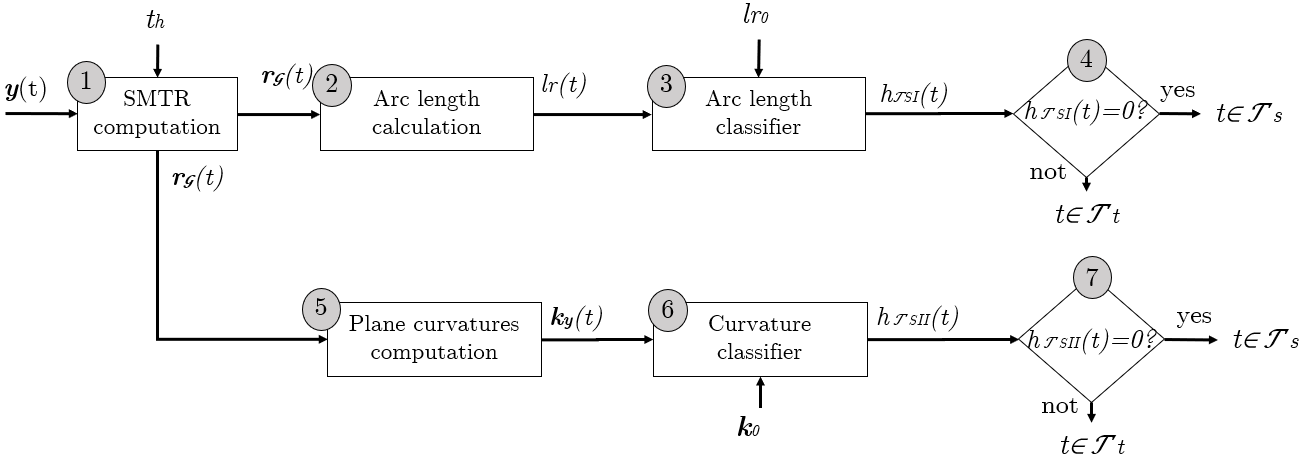}
	\caption{Scheme of the proposed methodology for transient regimes and stationary regime classification.}
	\label{fig:SMTR-GPD}
\end{figure*}

The general idea behind each stage of the proposed methodology can be summarized as follows:

\begin{enumerate}
	
	\item The sample moment-time representation curve $\boldsymbol{r}_\mathcal{G}(t)$ of the dynamical system is computed using the available sensors measuring the systems output $\boldsymbol{y}(t) \in \mathbb{R}^{l_y}$ and a previously selected time moving window $t_h$ (see Sections  \ref{subsec:srds} and \ref{subsec:crds}).

	\item After constructing $\boldsymbol{r}_\mathcal{G}(t)$, its arc length $l_r(t)$ is calculated (see Section \ref{subsec:atrc}).

	\item A threshold $l_{r0}$ and the arc length $l_r(t)$ is used to generate the piece-wise classifier $h_{\mathcal{T}_{sI}}(t)$ (see Section \ref{subsec:atrc}).

	\item The arc-length-based classifier $h_{\mathcal{T}sI}(t)$ is applied to detect time periods where the system is in a transitory or stationary regime (see Section \ref{subsec:atrc}).

	\item \black{The curve representing the dynamic system is segmented into plane curves $\bs{r}_{y_j}(t), j=1,2,\dots,l_y$, and calculate the vector of curvatures $\boldsymbol{k}_{\boldsymbol{y}}(t)$ (see Section \ref{subsec:ctrc})}.
	
	\item \black{The threshold $\boldsymbol{k}_{0}$ and the magnitude of the curvatures $\boldsymbol{k}_{\boldsymbol{y}}(t)$ is applied to generate the piece-wise classifier $h_{\mathcal{T}_{sII}}(t)$ (see Section \ref{subsec:ctrc})}.
	
	\item  The curvature-based classifier $h_{\mathcal{T}sII}(t)$ is used to detect transient or stationary regimes in the analyzed system (see Section \ref{subsec:ctrc}).
\end{enumerate}


For the application of the proposed approach, it is assumed that the analyzed dynamical system $\mc{G}$ represented by the tuple $\{\mc{X}, \mc{T}, \mc{R}\}$ (where $\mc{X}$ represents the state-space of $\mc{G}$, $\mc{T}$ denotes the set of times when the state-space is mapped, and $\mc{R}$ is the function that specifies the evolution of $\mc{X}$, as reported in \cite{giuntiDYNAMICALSYSTEMSMONOIDS2012}) is asymptotically stable or marginal stable following the definitions given in \cite{berglundGeometricalTheoryDynamical, hespanha2018linear}. Furthermore, the following assumptions over the analyzed dynamical system are stated. 

\begin{assumption}
The analyzed dynamical system is monitored by a group of discrete sampled sensors with a constant sampling time $t_s$ such that $\bs{y}(t)\approx\bs{y}(nt_s)$. Also, $t_s$ is considered to properly capture the evolution of the dynamical system.  
\end{assumption} 
\begin{assumption}
The analyzed dynamical system is assumed to be observable, such that the behavior of all the states $\bs{x}(t) \in \mc{X}$ of the analyzed system can be inferred from the data of $\bs{y}(t)$. 
\end{assumption} 
\begin{assumption}
The dynamical system is assumed to be asymptotically stable or marginal stable, such that the states either converge to an equilibrium point or oscillate permanently like a cyclo-stationary process. 
\end{assumption} 
\begin{assumption}
The dynamical system can be represented as an ergodic stochastic process during its stationary regimes.  
\end{assumption} 
            

\subsection{Stationary regime of dynamical systems}
\label{subsec:srds}

Let $\boldsymbol{y}(t)$ represent the information captured by sensors located at a dynamical system $\mc{G}$. These observations can be interpreted as a stochastic process where the parameters of the joint probability function will vary depending on the dynamics and the noise of $\boldsymbol{y}(t)$. Then, based on Definition~\ref{def:wss}, a weak-stationary time interval $\mc{T}_s$ can be defined as

\begin{equation}
    \label{eq:tswg}
    \begin{aligned}
    \mc{T}_{s} \triangleq \{t \in \mc{T}| E[\bs{y}(t)]=&E[\bs{y}(t+\tau)], E[\bs{y}^2(t)]<\infty,\\ K_{\bs{y} \bs{y}}(t_1,t_2)=&K_{\bs{y} \bs{y}}(t_2-t_1,0), t_1, t_2, \tau \in \Gamma \subset \mathbb{R}_{\geq0}\}.
\end{aligned}
\end{equation}

Therefore, the time intervals $\mc{T}_s$ defined in \eqref{eq:tswg} correspond to time intervals of stationary regime. Therefore, \black{the dynamical system is at an equilibrium or asymptotic stable operational point following Assumption 3}. Furthermore, the noise associated with the sensors is invariant during this time interval. It should be noted that the transient regime periods will correspond to the complement of $\mc{T}_s$, i.e., $\mc{T}_t \triangleq \left\{ t\in \mc{T}  \vline~t\notin \mc{T}_s \right\}$.

Accordingly, if $t\in \Gamma \subset \mathbb{R}_{\geq0}$ is a stationary regime time interval, then, the moments $\beta_{(1,\bs{0})}=\bs{y}_{m}=\text{const}$, $\beta_{(2,\bs{y}_m)}=\text{const}$, $\beta_{(2,\bs{0})}<\infty$ and the autocorrelation function  $K_{\bs{y} \bs{y}}(t_1,t_2)$ just depends on lags inside $\Gamma$.

It is essential to emphasize that for dynamic systems with periodic signals, the definition of $\mc{T}_{s}$ will correspond to the time intervals where the stochastic process $\bs{y}(t)$ is cyclo-stationary, as established in Remark~\ref{re:wsc}. Consequently, the periods of time where $\beta_{(1,\bs{0})}=\bs{y}_{m}=const$, and $\beta_{(2,\bs{y}_m)}=const$ will occur over time intervals of $t_r$, where $t_r$ is the time period of the periodic stationary response of the system.

\subsection{Sample moment-time representation of dynamical systems}
\label{subsec:crds}
Based on the defined set in \eqref{eq:tswg}, the sample statistical moments of each component of $\bs{y}(t) \in \mb{R}^{l_y}$ defined in \eqref{eq:sample_moments} can be used to identify the  periods $\mc{T}_{s}$ where the measurements over $\mc{G}$ are considered weak-stationary processes. Therefore, for a stochastic process the spatial curve, defined as
\begin{equation}
\label{eq:smtrr}
\bs{r}_{\mc{G}}(t) \triangleq [t,\hat{\beta}_{(2,\bs{0})}(y_1,t),\hat{\beta}_{(2,\bs{0})}(y_2,t),...,\hat{\beta}_{(2,\bs{0})}(y_{l_y},t)], 
\end{equation}

\noindent  is constructed. Here, $\bs{r}_{\mc{G}}(t)$ will be referred as the sample moment-time representation (SMTR) of $\mc{G}$. $\bs{r}_{\mc{G}}(t)$ consists of the time signal and the moving raw sample moments of each sensor computed from  \eqref{eq:moments} as
\begin{equation}
    \label{eq:mscm}
    \resizebox{0.99\linewidth}{!}{$\hat{\beta}_{(2,\bs{0})}(y_j,t)= E[(y_j(\tau))^2]= \frac{1}{t_{h}}\int_{t-t_h}^{t}(y_j(\tau))^2 d\tau,~j=1,2,\dots, l_y,$}
\end{equation}

\noindent where the time moving window for computing sample moments is denoted as $t_h$, such that the moment is computed from the time $t-t_h$ until the current measurement at $t$. It should be noted that for wide-sense cyclo-stationary processes or marginal stable dynamical systems, $t_h$ needs to be greater than $t_r$, where $t_r$ is the period of the stochastic cyclo-stationary process. \black{Moreover, according to Assumption 1, maintaining a constant sample time $t_s$ enables a deterministic identification of the number of samples within $t_h$ necessary to accurately cover the period of the cyclo-stationary process.} This is in accordance with Remark~\ref{re:wsc} and the definition of a marginal stable system response. By utilizing the SMTR of $\mc{G}$, changes in the sample moments of the sensors can be monitored, which in turn, enables the identification of time intervals where the system is weak-stationary.

\black{In accordance with Assumption 3 and Assumption 4, a probabilistic measure of the monitored signals converges to a constant value during the stationary regimes of the system. Consequently, to ensure the weak-sense stationarity defined in \eqref{eq:tswg}, it is suggested to monitor the second raw sample moment in~\eqref{eq:mscm} for a significant probabilistic sample size. It is worth to mention that based on \cite{amentaPowerCrustUnions2001}, the second raw moment is related to the first raw moment and the second-centered moment by the expression $ E[(y_j(t))^2]=E[(y_j(t))]^2+E[(y_j(t)-y_{j_m})^2],~E[(y_j(t))]=y_{j_m},~j=1,2,..., l_y$.}


The geometrical properties of the SMTR curve were employed to classify time intervals as either stationary or transient regime. The theoretical development of the proposed classifiers is explained along the following subsections.

\subsection{Arclength transient regime classifier}
\label{subsec:atrc}
The constructed SMTR curve $\bs{r}_\mc{G}(t)$ that represents a dynamical system has the next properties regarding its arc length:
\begin{lemma}
    \label{lm:rp}
    If the curve $\bs{r}_\mc{G}(t)$ is the SMTR of a dynamical system $\mc{G}$, hence, $\left\lVert\bs{r}^{(1)}_{\mc{G}}(t)\right\rVert\geq 1~\forall~t \in \mc{T}_t$ and  $\left\lVert\bs{r}^{(1)}_{\mc{G}}(t)\right\rVert= 1$ if $t \in \mc{T}_s$. 
    
    \end{lemma}
    
    \begin{proof}

        Here, $\bs{r}_\mc{G}(t) = [t,\hat{\beta}_{(2,\bs{0})}(y_1,t),\hat{\beta}_{(2,\bs{0})}(y_2,t), \dots,\linebreak
        \hat{\beta}_{(2,\bs{0})}(y_{l_y},t)(t)]$. Therefore,  the magnitude of its time derivative can be computed as $\left\lVert\bs{r}^{(1)}_{\mc{G}}(t)\right\rVert = \left(1+\sum_{j=1}^{l_y}(\hat{\beta}^{(1)}_{(2,\bs{0})}(y_j,t))^2\right)^{\frac{1}{2}}$. Since $t \in \mc{T}_s$, then, $\hat{\beta}^{(1)}_{(2,\bs{0})}(y_j,t)(t)=\emptyset, ~\forall j =1,2,..., l_y$. Therefore, $\left\lVert \bs{r}^{(1)}_{\mc{G}}(t)\right\rVert = 1$. Otherwise, if $t \notin \mc{T}_s$, it is straightforward that $\left\lVert \bs{r}^{(1)}_{\mc{G}}(t)\right\rVert > 1$.

    \end{proof}
    \begin{theorem}
        \label{th:lr}
        If the curve $\bs{r}_\mc{G}(t)$ is the SMTR of a dynamical system $\mc{G}$ and the open interval $(t-t_h,t) \in \mc{T}_s $, hence,\quad \quad $l_r(t-t_h,t,\bs{r}_\mc{G})=t_h$.
        \end{theorem}
        \begin{proof}
        The arc length $l_r(\cdot)$ of $\bs{r}_{\mc{G}}(t)$ between $t-t_h$ and $t$ is defined as 
        $\int_{t-t_h}^{t}\left\lVert \bs{r}^{(1)}_{\mc{G}}(t) \right\rVert dt$. Following Lemma \ref{lm:rp}, $\left\lVert \bs{r}^{(1)}_{\mc{G}}(t) \right\rVert=1$ for $t \in \mc{T}_{s}$. Therefore, $\int_{t-t_h}^{t}1dt = t_h$.
        \end{proof}

        From Theorem \ref{th:lr}, the arc length of the SMTR curve is constant when the analyzed dynamical system is in a weak-stationary time interval. By defining \linebreak
        $l_{r\mathcal{T}}(t)=l_r(t-t_h,t,\bs{r}_\mc{G})-t_h$, \\
        then it becomes possible to utilize Theorem \ref{th:lr} and the arc length of $\bs{r}_\mc{G}(t)$ to identify transient regimes through the piece-wise function written as

\begin{equation}
    \label{eq:hrl}
    h_{\mc{T}_sI}(t)  \triangleq 
    \left\{
        \begin{array}{lr}
            0, & \text{if }\quad l_{r\mathcal{T}}(t-t_h,t,\bs{r}_{\mc{G}})= l_{r_0},\\
            1, & \text{if }\quad l_{r\mathcal{T}}(t-t_h,t,\bs{r}_{\mc{G}})>l_{r_0},
        \end{array}
    \right.
    \end{equation}

\noindent where $l_r(t_h)$ is the moving arc length of  $\bs{r}_\mc{G}(t)$ defined in Definition~\ref{def:al} and $t_h$  is the same time window used for the computation of $\bs{r}_\mc{G}(t)$, whereas $l_{r_0}$ is a threshold that states the minimum value of $l_r$ to consider a transient regime in the dynamical system. \black{In theory, $l_{r_0}=0$ following Theorem~\ref{th:lr}. Nevertheless, for noisy signals this threshold must be calibrated by following suitable procedures\footnote{A proposed approach for the calibration of this threshold is reported in the Section E of supplementary material, available in \cite{linkannexes}.}.} The function~\eqref{eq:hrl} allows to identify the changes of the statistical parameters of the observations of $\mc{G}$. Following Theorem~\ref{th:lr}, if $t \in \mc{T}_s$, then, $h_{\mc{T}_sI}(t)=0$, otherwise $h_{\mc{T}_sI}(t)=0$.

    \color{black}
    \subsection{Curvature transient regime classifier}
    \label{subsec:ctrc}
    In addition to the classifier based on arc length, the proposed method employs the curvatures of the SMTR curve as indicators of transient regimes in dynamic systems as follows. 
    
    Consider the set of plane curves defined as
    \begin{equation}
        \label{eq:prset}
        \mc{D}:= \left\{\bs{r}_{y_j}(t) \in \mathbb{R}^2 |~\bs{r}_{y_j}(t) \triangleq [t,\hat{\beta}_{y_j}(t)],~j=1,2,...,l_y. \right\}.
    \end{equation}

    The elements of $\mc{D}$ represents the marginal evolution of the systems outputs sample moment over time. Based on \eqref{eq:kur}, the curvature of each element of the set \eqref{eq:prset} can be computed as
    \begin{equation}
        \label{eq:ky}
        k_{y_j}(t) =\frac{|\hat{\beta}_{y_j}^{(2)}(t)|}{\left( 1+\left(\hat{\beta}_{y_j}^{(1)}(t)\right)^2 \right)^{\frac{3}{2}}},~j=1,2,...,l_y.
    \end{equation}
    Theorem \ref{th:kur} describes the properties of each curvature in \eqref{eq:ky} regarding transient and stationary regime classification.

    \begin{theorem}
        \label{th:kur}
        If the curve $\bs{r}_\mc{G}(t)$ is the SMTR of a dynamical system $\mc{G}$ and $t \in \mc{T}_{s}$, hence, $\lVert\bs{k}_{\bs{y}}(t)\rVert=\emptyset$, with \linebreak
        $\bs{k}_{\bs{y}}(t)\triangleq [k_{y_1}(t),k_{y_2}(t),\dots,k_{y_{l_y}}(t)]$. 
    \end{theorem}

    \begin{proof}
        The SMTR of a dynamical system is represented as \newline$\bs{r}_{\mc{G}}(t)=[t,\hat{\beta}_{(2,\bs{0})}(y_1,t),\hat{\beta}_{(2,\bs{0})}(y_2,t),...,\hat{\beta}_{(2,\bs{0})}(y_{l_y},t)]$. If the open interval $(t-t_h,t) \in \mc{T}_s $, thus, $\hat{\beta}^{(2)}_{(2,\bs{0})}(y_j,t)=0 ~\forall j=1,2,..., l_y,$ $v = 1,2$. Therefore, the norm of the vector of plane curvatures $\bs{k}_{\bs{y}}(t)$ is computed as
        \begin{align*} 
        \lVert \bs{k}_{\bs{y}}(t) \rVert = \sqrt{\sum_{j=1}^{l_y}\left(\frac{\left(|\hat{\beta}_{y_j}^{(2)}(t)|\right)^2}{\left( 1+\left(\hat{\beta}_{y_j}^{(1)}(t)\right)^2\right)^3}\right)}=\emptyset.
        \end{align*} 
    \end{proof}
    \begin{remark}
    \label{re:conditions}

    Theorem \ref{th:kur} states that the curvatures of all elements within \eqref{eq:prset} are zero during stationary time intervals. In the context of stationary processes, Theorem \ref{th:kur} implies that there are no changes in the marginal sampled statistical parameters of the dynamic system states within the evaluated interval. It is important to emphasize that this result holds true solely within the open interval $(t-t_h,t)$, ensuring that the second derivative of the sampling moments is null. This nullity arises from the absence of changes in the trajectory of the curve, rather than the presence of inflection points in the trajectory of the curve.

    \end{remark}

    Based on Theorem \ref{th:kur}, it is proposed to classify the time intervals of transient regimes and stationary regimes using the piece-wise function defined as
    \begin{equation}
        \label{eq:hkkk}
        h_{\mc{T}_sII}(t)  \triangleq 
        \left\{
            \begin{array}{lr}
                0, & \text{if } \quad \lVert\bs{k}_{\bs{y}}(t)\rVert= \bs{k}_{0},\\
                1, & \text{if } \quad \lVert\bs{k}_{\bs{y}}(t)\rVert> \bs{k}_{0}.
            \end{array}
        \right.
        \end{equation}
    \color{black}

    \black{As shown in~\eqref{eq:hkkk}, $\bs{k}_{0}=0$ for ideal deterministic dynamic systems. Nevertheless, for noisy signals this threshold must be calibrated by following suitable procedures \footnotemark[3].} It is worthwhile to mention that the numerical derivatives of the proposed methodology were computed using a Savitzky-Golay filter proposed in \cite{gorry1990general}. This filter allows computing high-order derivatives of noisy signals by means of a smoothing filter. Also, the numerical integrals of the proposed method were computed using the trapezoidal integration rule reported in \cite{yeh2002using}.

    \black{For transient classification in online applications using acquisition systems with sampling times $t_s$, where sensor samples are taken at discrete times $t \approx t_s n$, it becomes feasible to approximate $\lVert\bs{k}_{\bs{y}}(t)\rVert$ and $l_{r\mathcal{T}}(t)$ for the SMTR curve of the $p$-th raw sampling moment using finite differences\footnote{The mathematical derivation of these formulas using finite differences is provided in Section~F of the supplementary material, available in \cite{linkannexes}.}. These approximations can be represented in discrete time as} 
    \begin{subequations}
        \label{eq:lrRiemann}
        \begin{align}
        \black{l_{r}(n)} \approx& \black{\frac{1}{l_h}\sum_{k=n-l_h}^{n}\sqrt[]{t_h^2+\lVert (\boldsymbol{y}(k))^p-(\boldsymbol{y}(k-\Delta_h))^p\rVert^2},}\\
        \black{\lVert\bs{k}_{\bs{y}}(n)\rVert} \approx& \black{t_hl_h\sqrt{\sum_{j=1}^{l_y}\left(\frac{\left|\Delta y^p_j(n,l_h) -\Delta y^p_j(n-1,l_h)\right|}{\left( (t_h^2+\Delta y^p_j(n,l_h)^2)\right)^{3/2}}\right)^2},}\\
        \black{\Delta y^p_j(n,l_h)} =&  \black{(y_j(n))^p-(y_j(n-l_h-1))^p,}
        \end{align}
    \end{subequations} 

    \noindent \black{where $l_h=\frac{t_h}{t_s}$ and $\Delta_h = l_h + 1$. Once the arc length and curvature norm are computed with~\eqref{eq:lrRiemann}, the classification can be performed using~\eqref{eq:hrl} and \eqref{eq:hkkk}. It is important to mention that the proposed method uses $p=2$, which implies that the classifiers are based on the second sample raw moments. Such moments are widely used for the characterization of periodic signals as exposed in \cite{petrovic2012root, van1991periodic, poomjan2013accurate}. Furthermore, as mention in Section~\ref{subsec:srds}, monitoring the second raw moment give enough conditions to ensure wide-sense stationary. Therefore, increasing $p$ makes the smaller transitions negligible and highlights the more abrupt changes in the arc length and curvature of the SMTR curve of the analyzed system.}

\color{black}

\subsection{Limitations}
\label{subsec:ctrc}

The proposed methodology relies on certain assumptions to facilitate the detection of transient and stationary regimes. These assumptions constitute the primary limitations of the method. A discussion of such limitations is next shown.

    \textit{Observability:} The proposed methodology relies on the assumption of full observability as stated in Assumption 2. However, if the dynamic system under study is not completely observable, inferring all transient and stationary regimes of the states solely from monitored outputs becomes unfeasible. Consequently, analyzing the systems outputs alone may not suffice for accurately classifying the regimes of the studied system\footnote{A study assessing the influence of observability on the proposed classifiers is detailed in Section G of the supplementary material, available in \cite{linkannexes}.}.

    \textit{Sampling:} Assumption 1 specifies that the sampling time of the sensors remains constant and adequately captures the dynamics of the studied system. Introducing a variable sampling time could result in inadequate acquisition of the dynamics of the system. Consequently, essential features such as response times, periods, and transitions may not be discernible in the signals monitored by the sensors \footnote{ An analysis of the influence of variable sampling time on the performance of the proposed method is detailed in Section I of the supplementary material, available in \cite{linkannexes}}.

    \textit{Noise:} The effectiveness of the presented classifiers relies on the dominance of the output signals over background noise. The signal-to-noise ratio is a common measure of the presence of noise in signals. Definitions and characteristics of the SNR and its variants are exposed in \cite{welvaert2013definition}. In general, for signal processing a SNR of 100 means a neglected signals noise, a SNR of 0 means the energy of the noise is comparable with the signal, and a negative SNR means a predominant noise in the signal\footnote{An analysis of the influence of signals noise on the performance of the proposed method is detailed in Section H of the supplementary material, available in \cite{linkannexes}.}.

\subsection{\black{Pseudocode}}

In order to ensure precision in the evaluation of the classifiers within an online computation environment, the integrals and derivatives of the signals described in the proposed methodology must be computed using numerical approximations, such as the Stencil method or Savitzky-Golay filters, as presented in \cite{gorry1990general, kamakoti2010high}. By applying finite difference approximations, both the arc length and the plane curvature can be directly obtained from the raw measurements acquired by the sensors monitoring the system. The pseudocode for calculating the proposed transition classifiers using finite differences is provided below.

\begin{algorithm}[!htb]
	\SetAlgoLined
	\KwIn{$\boldsymbol{y}(t)$, $t_s$, $t_h$, $\boldsymbol{k}_{0}$, $l_{r0}$.}
	\KwResult{$h_{\mathcal{T}_{sI}}(t)$, $h_{\mathcal{T}_{sII}}(t)$. }

    Estimate the magnitude of the plane curvatures and arc length of the SMTR curve inside the analyzed time window $t_h $using  $\boldsymbol{y}(t)$ as: \;
      
    $l_r(n) \approx \frac{1}{l_h}\sum_{k=n-l_h}^{n}\sqrt[]{t_h^2+\lVert (\boldsymbol{y}(k))^2-(\boldsymbol{y}(k-\Delta_h))^2\rVert^2},$\;

    $\lVert\boldsymbol{k}_{\boldsymbol{y}}(n)\rVert \approx t_hl_h\sqrt{\sum_{j=1}^{l_y}\left(\frac{\left|\Delta y^2_j(n,l_h) -\Delta y^2_j(n-1,l_h)\right|}{\left( (t_h^2+\Delta y^2_j(n,l_h)^2)\right)^{3/2}}\right)^2},$\;
      
    with, \;

    $ \Delta y^2_j(n,l_h) =  (y_j(n))^2-(y_j(n-l_h-1))^2$, $l_h=\frac{t_h}{t_s}$ and $\Delta_h = l_h + 1$.\;

    Finally, build the classifiers $h_{\mathcal{T}_{sI}}(t)$ and $h_{\mathcal{T}_{sII}}(t),~t \in (0,t_h)$ inside the analyzed the moving time window $t_h$ with $\boldsymbol{k}_{0}$ and $l_{r0}$ as:\; 

    \If{$l_{r\mathcal{T}}(t) > l_{r0}$}{
    $h_{\mathcal{T}_{sI}}(t)=1$.\;
    }
    \Else{
        $h_{\mathcal{T}_{sI}}(t)=0$.\;
    }

    \If{$\lVert \boldsymbol{k}_{\boldsymbol{y}}(t) \rVert > \boldsymbol{k}_{0}$}{
    $h_{\mathcal{T}_{sII}}(t)=1$.\;
    }
    \Else{
        $h_{\mathcal{T}_{sII}}(t)=0$.\;
    }
\caption{Computation of transient regime classifiers using simple finite differences}
\label{alg:3}
\end{algorithm}

Finally, it is worth mentioning that, based on the constructed classifiers, it is possible to detect transient regime events. The value of the classifiers will be equal to 1 in the presence of transient behavior, and equal to 0 in the case of stationary or cyclo-stationary behavior, according to the assumptions established for the proposed methodology.

The pseudocode of the computation of the general form of the classifiers using the SMTR curve is presented in Section N. of the supplementary material presented in \cite{linkannexes}.

The implementation of the classifiers was carried out primarily in C++ and Python, utilizing the  \texttt{lapack}, \texttt{blas}, and \texttt{numpy} libraries.  
\color{black}

\section{Validation}
\label{sec:Validation}
The proposed methodology was validated with three dynamical systems. These systems were subjected to external inputs that varied over time to induce regime transitions. The simulations included a linear system, a non-linear system, and a discontinuous forced system.  The observations of the dynamical systems were sampled at $t_s=0.002~[\si[]{\second}]$, and the external inputs were generated accordingly. \black{The parameters for the proposed method were calibrated for each of the three dynamical systems \footnote{ A proposed approach for the calibration of the methods parameters is reported in the Section E of supplementary material, available in \cite{linkannexes}.}}. The specific parameters utilized for each simulation are detailed in Table~\ref{tb:smtr_gpd_params}.

\begin{figure*}[!t]
    \centering
    \begin{subfigure}{0.32\linewidth}
        \includegraphics[width=\linewidth]{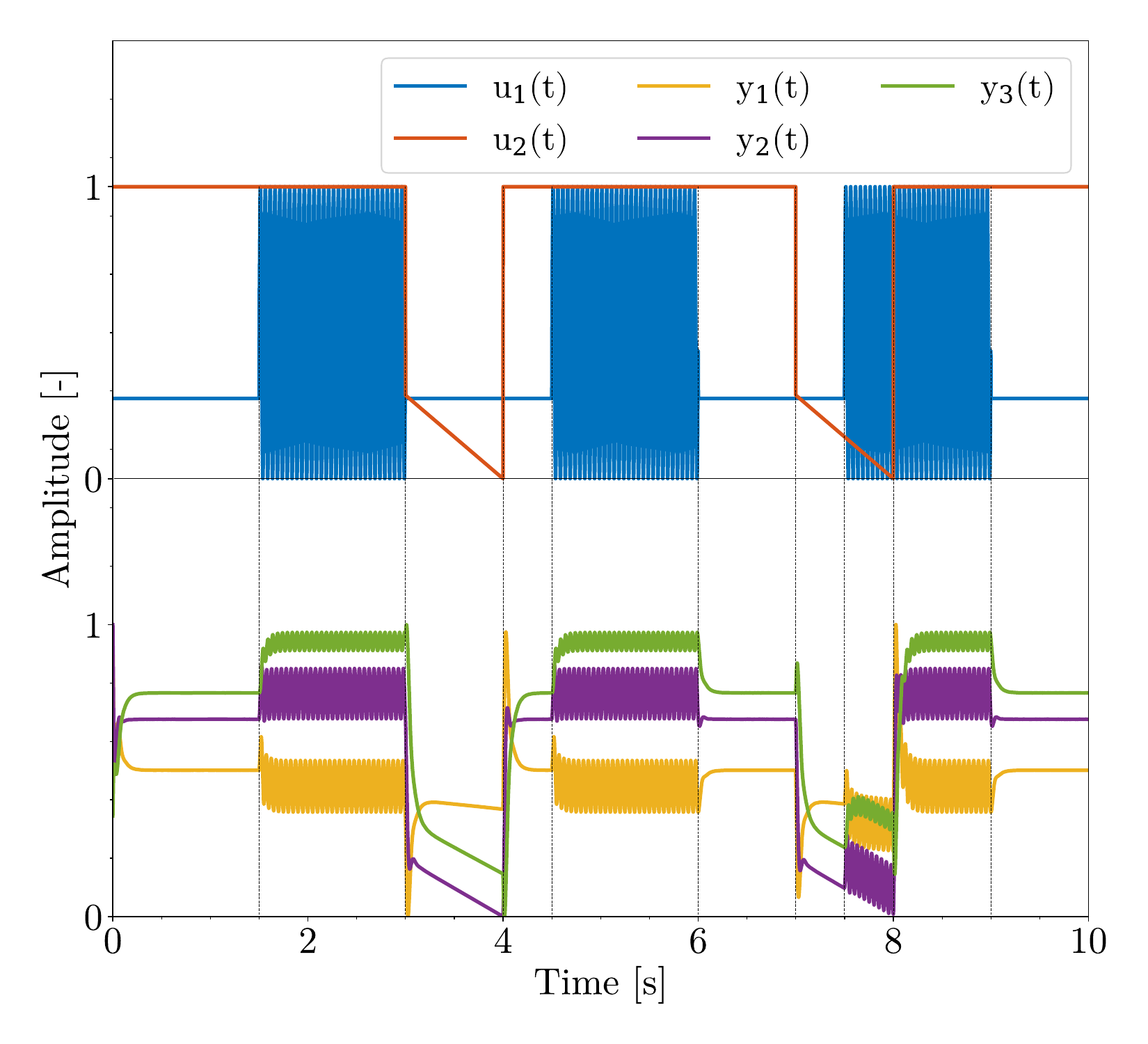}
        \caption{}
        \label{fg:linear_ds_inputs and_output}
    \end{subfigure}
    \centering
    \begin{subfigure}{0.32\linewidth}
        \includegraphics[width=\linewidth]{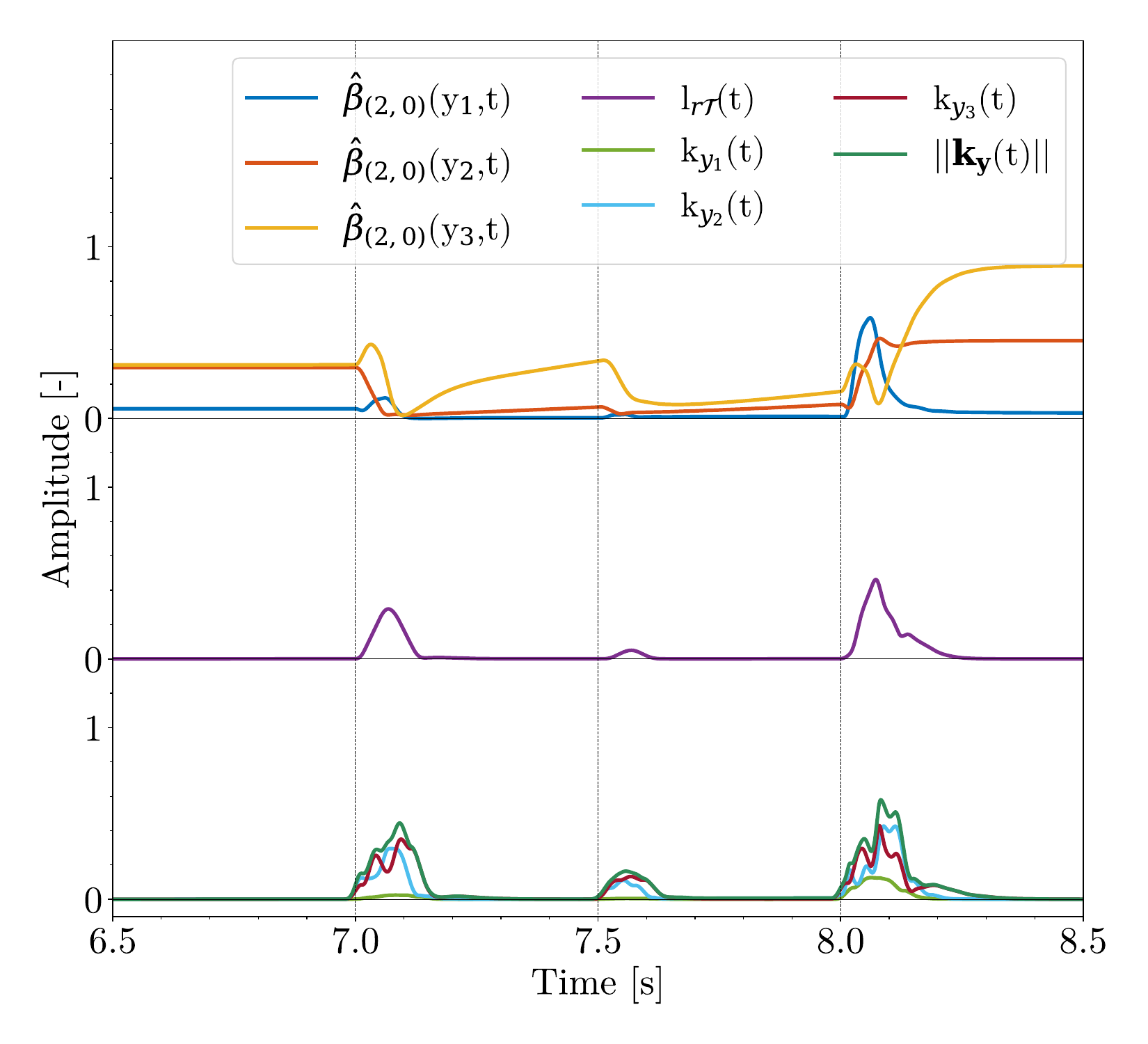}
        \caption{}
        \label{fg:SMTR_linear_DS}
    \end{subfigure}
    \centering
    \begin{subfigure}{0.32\linewidth}
        \includegraphics[width=\linewidth]{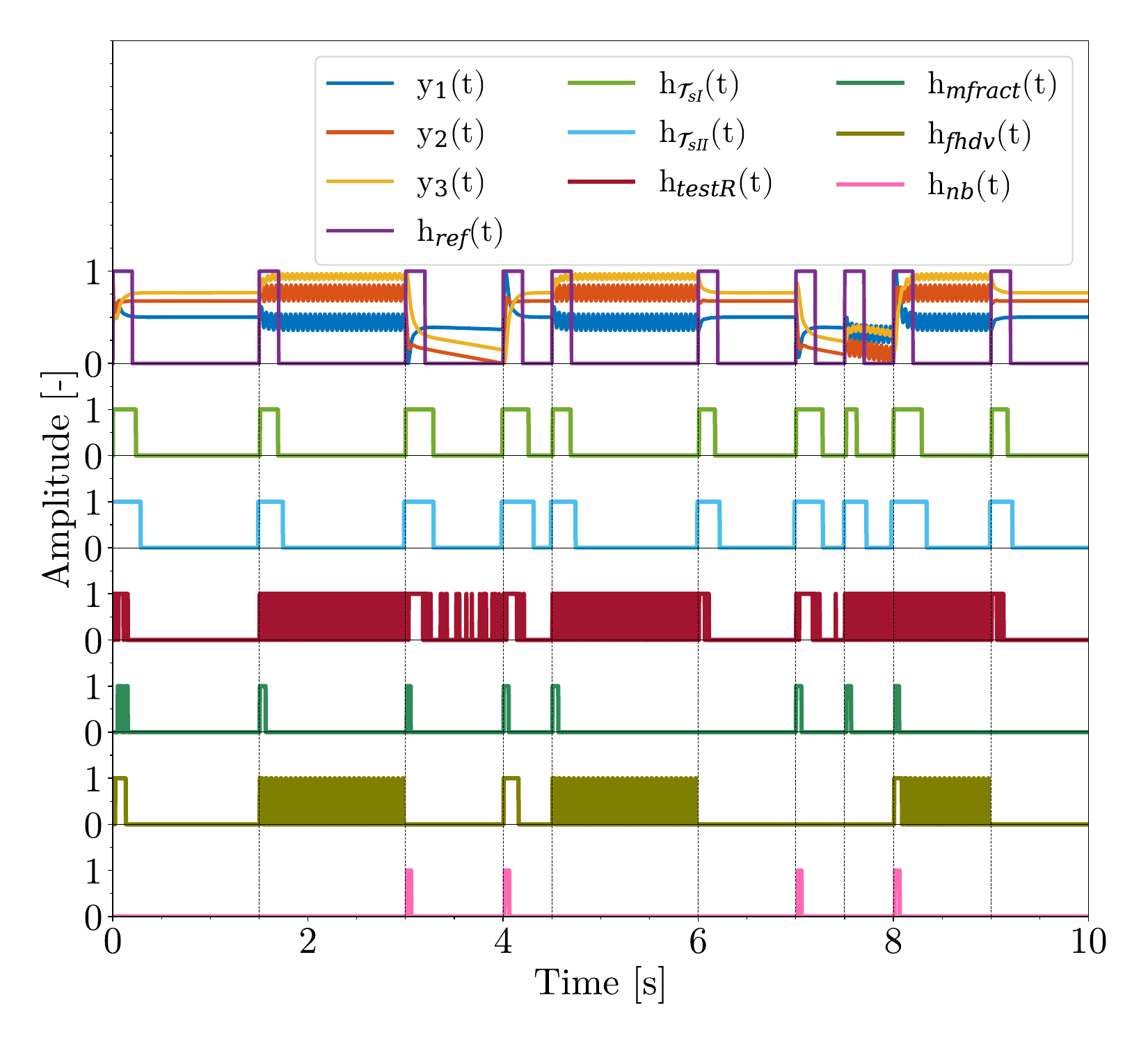}
        \caption{}
        \label{fg:linear_ds_inputs and_output2}
    \end{subfigure}

    \caption{Results of the simulated linear dynamic system. (a) Input and output signals of the simulated linear dynamical system. (b) SMTR curve and its geometrical properties for the considered linear-dynamical system. (c) Comparison between the proposed classifiers and the evaluated classifiers from the literature.}
    \label{fg:linear_ds_inputs and_outputs}
\end{figure*}

\begin{table}[!t]
    \centering
    \small
    \caption{Parameters used for the validation assessment of the proposed methodology for each analyzed dynamical system.}
    \label{tb:smtr_gpd_params}
    \scalebox{0.84}{
    \begin{tabular}{  l  c  c  c }
        dynamical system  &Time window& Arc length & Curvature norm \\ 
        type&$(t_h~[\si[]{\second}])$&threshold $(l_{r_0}~[-])$&threshold $(\bs{k}_{0}~[-])$\\ 
        \hline
        Linear &$0.05$ & $3.2\times10^{-4}$ &  $4.89$   \\
        Non-linear & $0.15$ & $3.50\times10^{-5}$ & $0.02$ \\
        Discontinuous  & $0.32$ & $0.02$ &$2.00$ \\
        \hline
    \end{tabular}}
\end{table}

Subsequently, the computed classifiers $h_{\mc{T}_s I}(t)$ and $h_{\mc{T}_s II}(t)$ were compared with existing methods found in the literature. Table~\ref{tb:h_to_compare} outlines the classifiers utilized for this comparative analysis. Additionally, such a table also provides a comprehensive overview of the underlying principles and concepts upon which each classifier is based, as well as their corresponding reference. The parameters used for each classifier were tuned to ensure optimal performance in the classification assessment for each analyzed system\footnote{The parameters employed for computing each studied classifier in the literature are shown in Section D of the supplementary material, available in \cite{linkannexes}.}.

\begin{table}[!t]
    \centering
    \small
    \caption{Transient regimes classifiers from the literature.}
    \scalebox{0.84}{
    \begin{tabular}{  l  c  c  c }
        Name  & transient regimes classifier description                            & Reference                                         \\ [0.3ex]
        \hline
        $h_{testR}$ & Hypothesis test and R-test estimation & \cite{rhinehartAutomatedSteadyTransient2013}      \\
        $h_{mfrac}$ & Multi-fractal dimension analysis & \cite{shaw1997multifractal}       \\
        $h_{fhdv}$  & HVDC-median statistical measure & \cite{liu2020fast}                    \\
        $h_{nb}$  & Difference signal and dynamic noise band selection & \cite{yu2022identification}                    \\
        \hline
    \end{tabular}}
    \label{tb:h_to_compare}
\end{table}

\subsection{Linear dynamical system simulation}
\label{subsec:sdsv}

As a first validation of the proposed method, a stable linear continuos-time state-space model was simulated\footnote{A complete description of this dynamical system is specified in Section~A of the supplementary material, available in \cite{linkannexes}.}. This dynamical system is commonly written as 
\begin{align*}
    \dot{\bs{x}}(t)&=A\bs{x}(t)+B\bs{u}(t) + \bs{\upsilon}(t),\\
    \bs{y}(t)&=C\bs{x}(t)+D\bs{u}(t) + \bs{\eta}(t).
\end{align*}

 By means of the modulus operator (\%), the inputs of the linear dynamical system were defined using the piece wise functions  
\begin{align*}
u_1(t) &= 
     \begin{cases}
        \text{$-1.0$}&\text{$t\%3.0<1.5$, } \\
        \text{$2.0\sin(40\pi t)$}&\text{$t\%3.0>1.5$,}\\
     \end{cases}
     \\
     u_2(t) &= 
     \begin{cases}
       \text{$2.5$}&\text{$t\%4.0<3.0$,} \\
        \text{$-t$}&\text{$t\%4.0>3.0$.}\\
     \end{cases}
\end{align*}

It is worth to mention that the selected inputs allow to evaluate the classifiers under scenarios of stationarity and cyclo-stationarity. The system was simulated for 10~\si[]{\second} and the classifiers were evaluated using the parameters specified in Table~\ref{tb:smtr_gpd_params}. The inputs and outputs of the simulated linear system are shown in Figure~\ref{fg:linear_ds_inputs and_outputs}(\subref*{fg:linear_ds_inputs and_output}).  For illustrative purposes, it is worth mentioning that the inputs and outputs displayed in Figure~\ref{fg:linear_ds_inputs and_outputs}(\subref*{fg:linear_ds_inputs and_output}) were shifted and normalized.

\subsection{Linear dynamical system results}

A detail view of the normalized SMTR curve, arc length and curvatures of the linear dynamical system are shown in Figure~\ref{fg:linear_ds_inputs and_outputs}(\subref*{fg:SMTR_linear_DS}). It should be noted that the geometric parameters of the curve increase their values when the systems input underwent changes attributable to the presence of a transient regime. Moreover, during the stationary regimes, the parameters tend towards zero following Theorems~\ref{th:lr} and~\ref{th:kur}.

Note that the selection of the time window $t_h$ allows to identify the time intervals were the outputs of the dynamical system are wide-sense cyclo-stationary, for instance, during the time interval between 7.5 and 8.0 ~\si[]{\second}. 



 The comparison between the proposed classifier and the one specified in Table~\ref{tb:h_to_compare} is presented in Figure~\ref{fg:linear_ds_inputs and_outputs}(\subref*{fg:linear_ds_inputs and_output2}). The proposed classifiers $h_{\mc{T}_s I}$ and $h_{\mc{T}_s II}$ show effectiveness in identifying all the changes in the steady regimes of the linear dynamical system. Regarding the classifiers in the literature, the classifiers $h_{mfract}$ and $h_{fhdv}$ showed the best performance for this scenario.

\subsection{Non-linear dynamical system simulation}
\label{subsec:edsv1}

Nonlinear systems are a distinct category of dynamic systems that exhibit a temporary evolution that cannot be described simply by a proportional relationship between their inputs and outputs. Rather, their response is also influenced by their initial conditions and characteristic parameters, resulting in unpredictable and often highly complex behavior. These systems are expressed in state-space form as
\begin{align*}
    \dot{\bs{x}}(t)&=\bs{\alpha}(\bs{x},\bs{u},t)+ \bs{\upsilon}(t),\\
    \bs{y}(t)&=\bs{\gamma}(\bs{x},\bs{u},t) + \bs{\eta}(t),
\end{align*}

\noindent where the functions $\bs{\alpha}(\bs{x},\bs{u},t)$ and $\bs{\gamma}(\bs{x},\bs{u},t)$ represent the nonlinear equations that govern the time evolution of the outputs and states of the system. The non-linear Lorenz attractor was implemented using $a_1=560, a_2=200$, $a_3=53.3$ and $\bs{x}_0=[102,~102,~199]$ to applying the proposed method\footnote{A full description of the Lorenz Attractor is shown in Section~B of the supplementary material, available in \cite{linkannexes}.}. The parameters were selected in order to have a non-chaotically behavior, such that the transient regimes time intervals can be attributed to the external inputs and not to its self excited nature.

The inputs for the validation of the non-linear dynamical systems were defined by the piece wise functions written as
\[   
\resizebox{0.99\linewidth}{!}{$
u_1(t) = 
     \begin{cases}
        \text{$6.0$}&\text{$t\%4.0<2.0$,}\\
        \text{$-30.0$}&\text{$t\%4.0>2.0$,} \\
     \end{cases} \quad u_2(t) = 
     \begin{cases}
        \text{$9.0$}&\text{$t\%5.0<2.5$,}\\
        \text{$21.0$}&\text{$t\%5.0>2.5$.}\\
     \end{cases}
$}
\]
The system was simulated for 10~\si[]{\second} after reaching the initial steady-regime conditions. Figure~\ref{fg:nonlinear_ts_detection2}(\subref*{fg:nonlinear_ts_detection}) illustrates the inputs and outputs of the simulated system, indicating that the systems outputs demonstrate a cyclo-stationary and marginally stable response, even in the presence of constant inputs.




\subsection{Non-linear dynamical system simulation results}
\begin{figure*}[!t]
    \centering
    \begin{subfigure}{0.32\linewidth}
        \includegraphics[width=\linewidth]{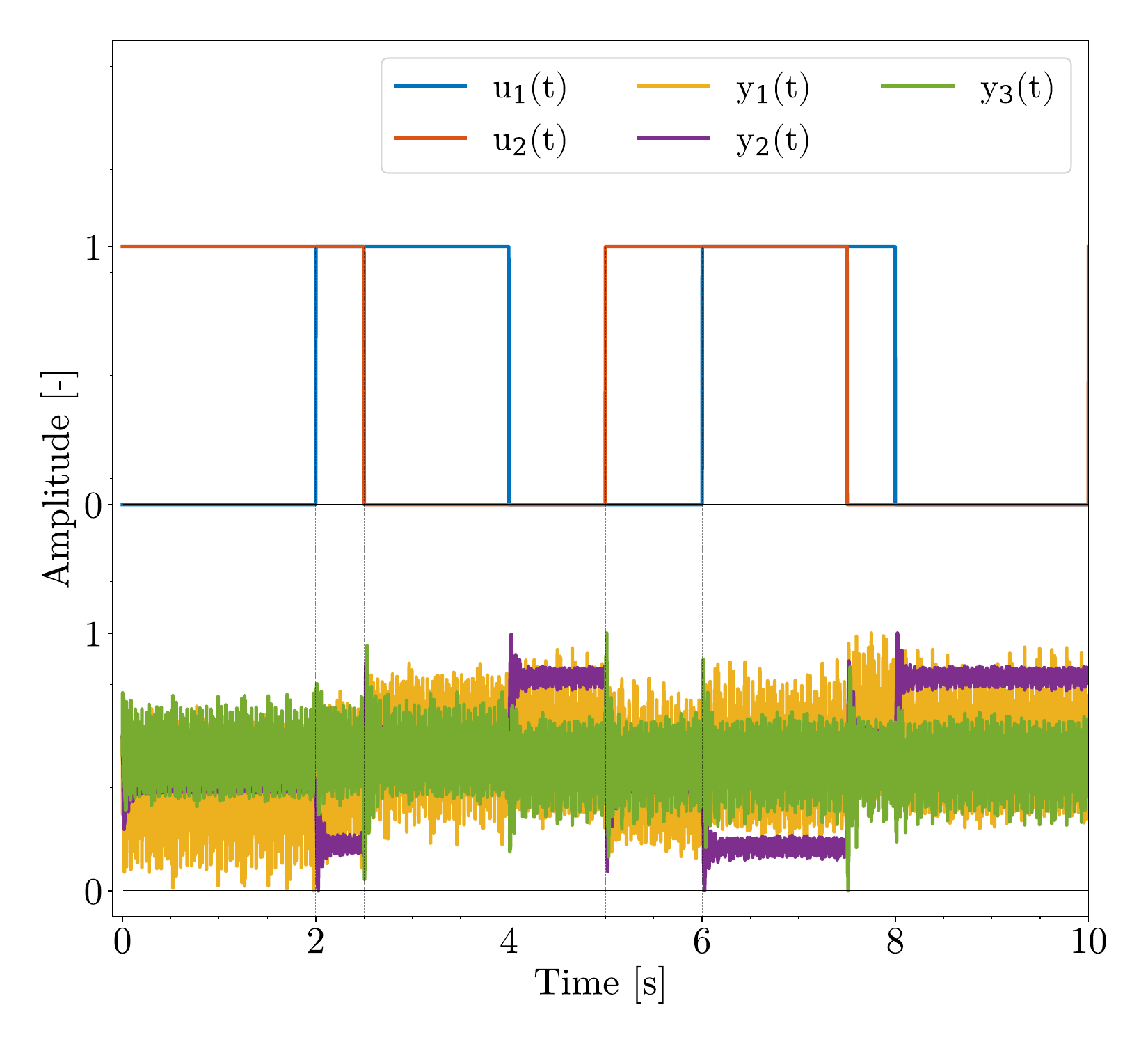}
        \caption{}
        \label{fg:nonlinear_ts_detection}
    \end{subfigure}
    \centering
    \begin{subfigure}{0.32\linewidth}
        \includegraphics[width=\linewidth]{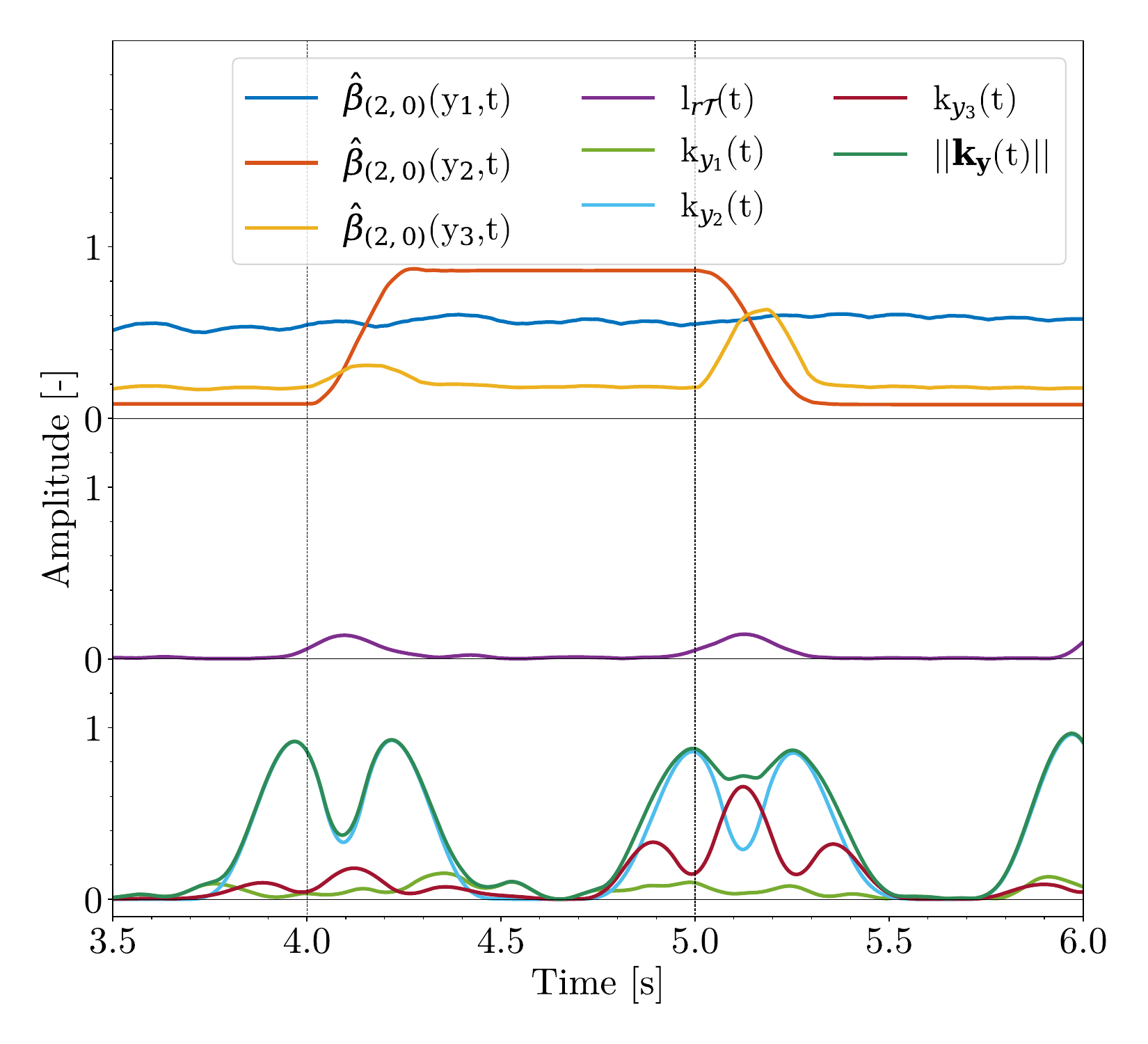}
        \caption{}
        \label{fg:SMTR_nonlinear_DS}
    \end{subfigure}
    \centering
    \begin{subfigure}{0.32\linewidth}
        \includegraphics[width=\linewidth]{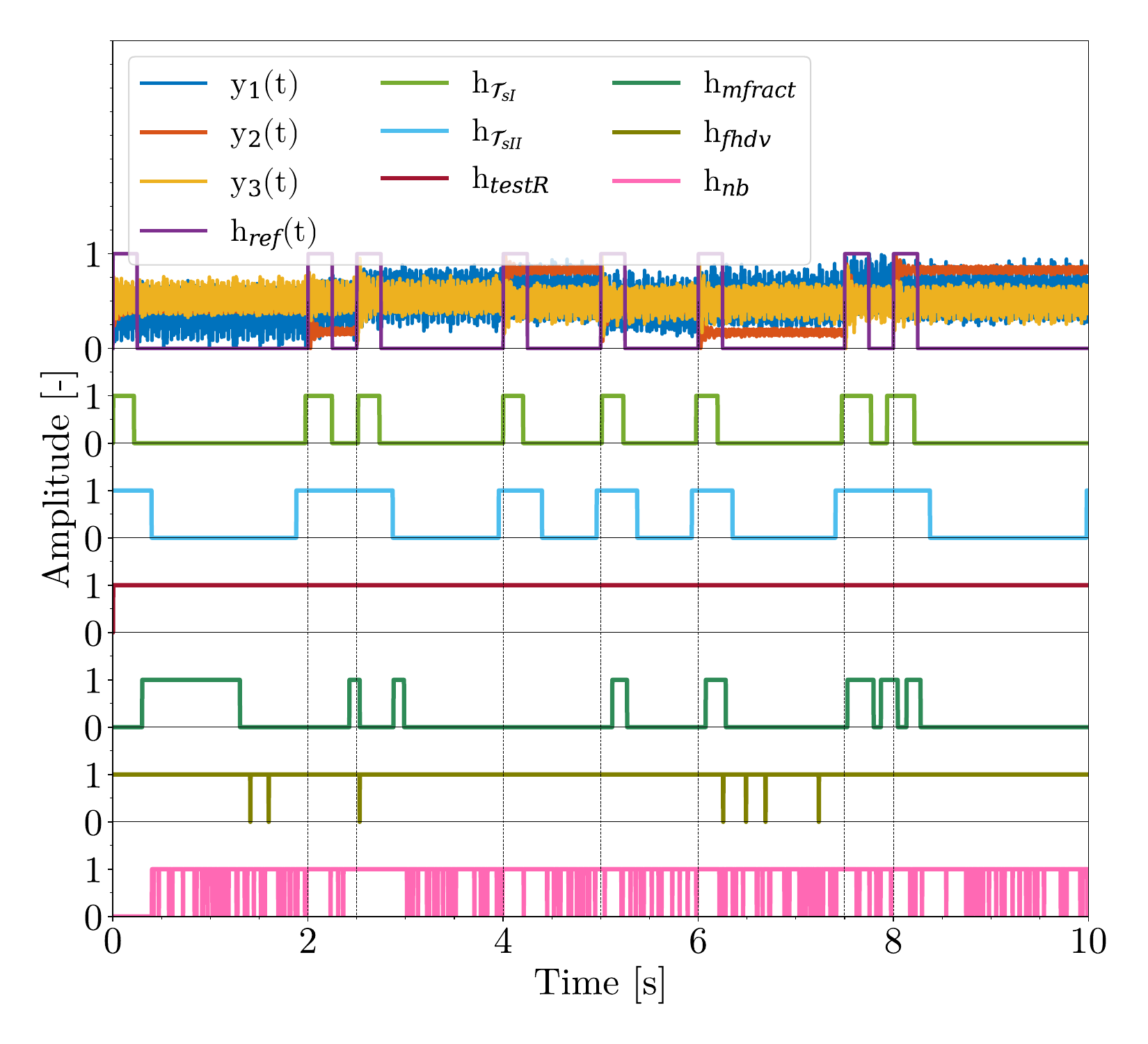}
        \caption{}
        \label{fg:nonlinear_ts_ltcompartison}
    \end{subfigure}

    \caption{Results of the simulated non-linear dynamic system. (a) Input and output signals of the simulated non-linear dynamical system. (b) SMTR curve and its geometrical properties for the Lorenz attractor. (c) Comparison between the non-linear-dynamical system outputs and the  computed classifiers. }
    \label{fg:nonlinear_ts_detection2}
\end{figure*}



Figure~\ref{fg:nonlinear_ts_detection2}(\subref*{fg:SMTR_nonlinear_DS}) provides a detailed view of the geometric properties of the SMTR curve of the Lorenz attractor. Analogous to the linear dynamic system, the findings reveal a direct correlation between input change and an increment in the arc length and curvatures of the SMTR curve. Furthermore, the analysis shows that as the system approaches its steady cyclo-stationary response, these values gradually tends to zero.


The results shown in Figure~\ref{fg:nonlinear_ts_detection2}(\subref*{fg:nonlinear_ts_ltcompartison}) illustrate the outputs of the simulated Lorenz attractor and the evaluated classifiers. In this instance, the function $h_{\mc{T}s I}(t)$ accurately identifies the time intervals of the transient regime in the system. Conversely, $h_{\mc{T}s II}(t)$ is unable to classify all the transient regimes during the simulation. Among the classifiers outlined in the literature, the  $h_{mfract}$ successfully classifies some transient regimes, but misidentifies certain stationary regimes as transient regimes. Meanwhile, the other classifiers are unable to accurately classify any regime.



\subsection{Discontinuous dynamical system simulation}
\label{subsec:edsv2}


Discontinuous dynamic systems are a class of nonlinear dynamical systems that exhibit abrupt changes and discontinuities in their dynamics, which occur at specific temporal or spatial points. A form of linear-discontinuous dynamical systems can be mathematically characterized by the conditions imposed on its states as follows:
\begin{align*}
    \dot{\bs{x}}(t)&=A\bs{x}(t)+B\bs{u}(t) + \bs{\upsilon}(t),\\
    \bs{y}(t)&=C\bs{x}(t)+D\bs{u}(t) + \bs{\eta}(t),
\end{align*}
\noindent where the parameters of the system can change under certain conditions on its states. For instance, consider a dynamical system in which the matrix $A$ undergoes a transition between two distinct modes of operation, triggered by the presence of certain conditions denoted as $\Phi_1(\boldsymbol{x})$ and $\Phi_2(\boldsymbol{x})$ in the states of the system. Accordingly, the matrix of the system can be expressed as a function that takes into account the prevailing conditions as
\[   
 A = 
     \begin{cases}
       \text{$A_1$} & \text{if $\Phi_1(\bs{x})$}\\
       \text{$A_2$} &\text{if $\Phi_2(\bs{x})$}\\
     \end{cases}.
\]
The double discontinuous oscillator is an example of a discontinuous dynamical system with the aforementioned description. This system comprises a double oscillator obstructed by an obstacle that restricts its movement\footnote{A full description of the studied double discontinuous oscillator is depicted in Section~C of the supplementary material, available in \cite{linkannexes}.}. For the simulated double discontinuous oscillator dynamical system, the inputs were defined as follows:
\[   
 u_2(t) = 0.0, \quad
 u_1(t) = 
     \begin{cases}
       \text{$t\%20.0<10.0$ } & \text{$200.0$,}\\
       \text{$t\%20.0>10.0$} &\text{$-100.0$.}\\
     \end{cases}
\]
After the initial steady state of the system, a simulation was conducted for a duration of 50~\si[]{\second}, during which the proposed methodology was applied to compute the regime classifiers. Figure~\ref{fg:disc_ts_detection2}(\subref*{fg:disc_ts_detection}) displays the inputs and outputs of the simulated system, which exhibit a marginal cyclo-stationary response to the constant excitations imposed, owing to the absence of damping parameters in the simulated oscillators.

\subsection{Discontinuous dynamical system results}



The SMTR curve of the discontinuous dynamical system and its geometrical properties are shown in Figure~\ref{fg:disc_ts_detection2}(\subref*{fg:SMTR_disc_DS}). The arc length demonstrates a pattern consistent with the behavior observed in other dynamic systems, wherein transitions and stationary regimes are present. However, in contrast, the curvatures remain non-zero during the stationary regime between 10 and 20~\si[]{\second}, as shown in Figure~\ref{fg:disc_ts_detection2}(\subref*{fg:disc_ts_detection}).

\begin{figure*}[!t]
    \centering
    \begin{subfigure}{0.32\linewidth}
        \includegraphics[width=\linewidth]{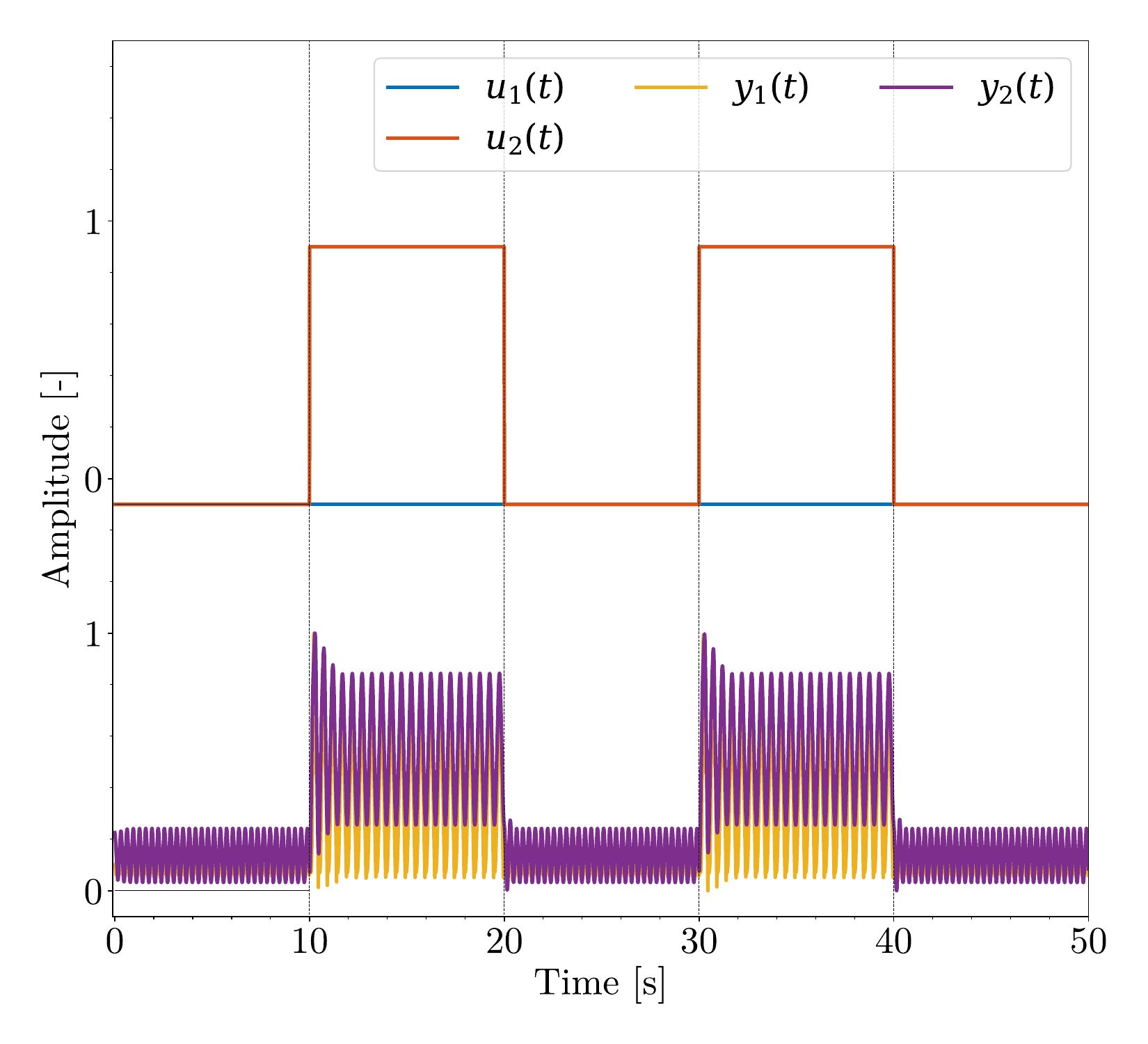}
        \caption{}
        \label{fg:disc_ts_detection}
    \end{subfigure}
    \centering
    \begin{subfigure}{0.32\linewidth}
        \includegraphics[width=\linewidth]{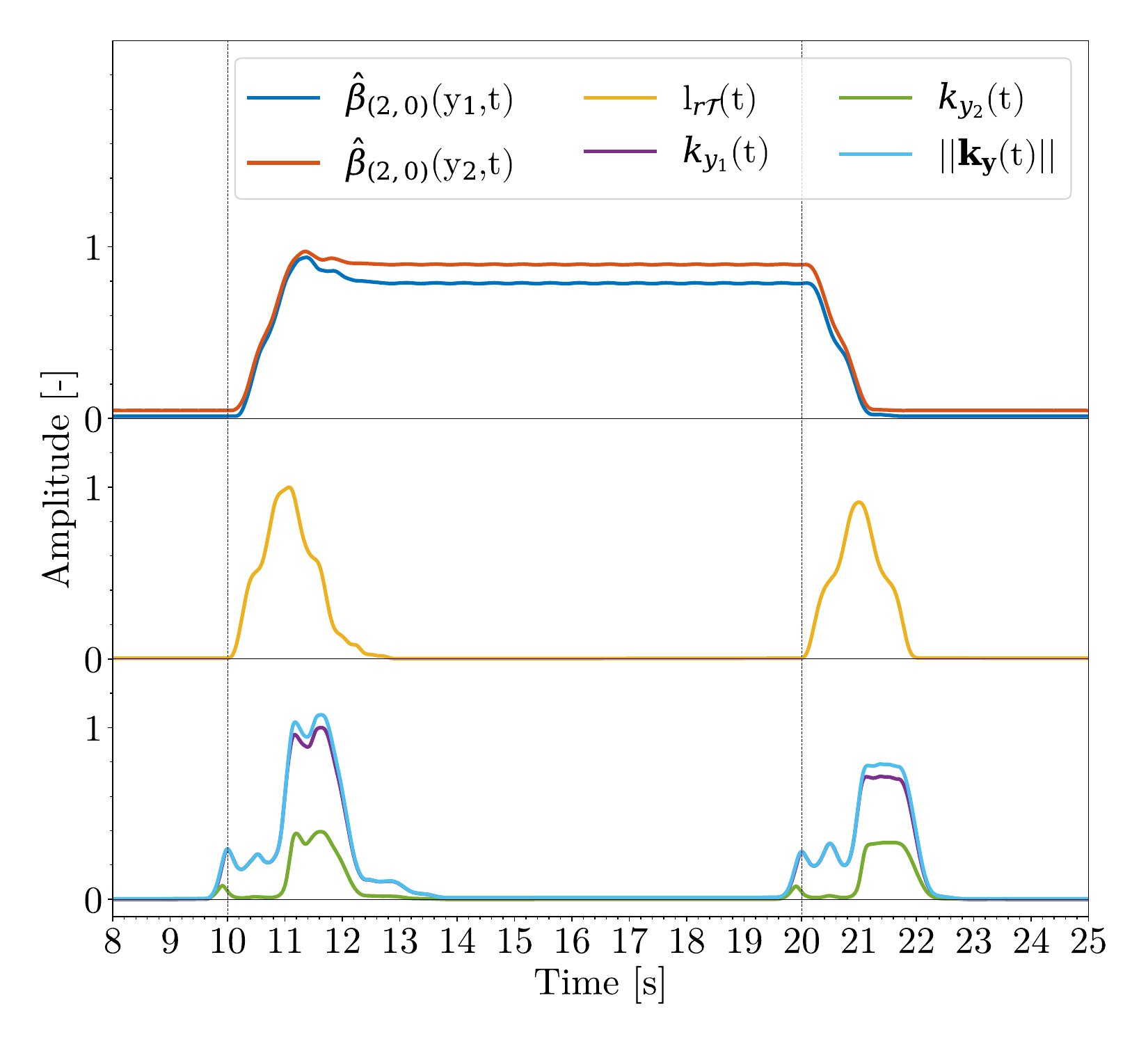}
        \caption{}
        \label{fg:SMTR_disc_DS}
    \end{subfigure}
    \centering
    \begin{subfigure}{0.32\linewidth}
        \includegraphics[width=\linewidth]{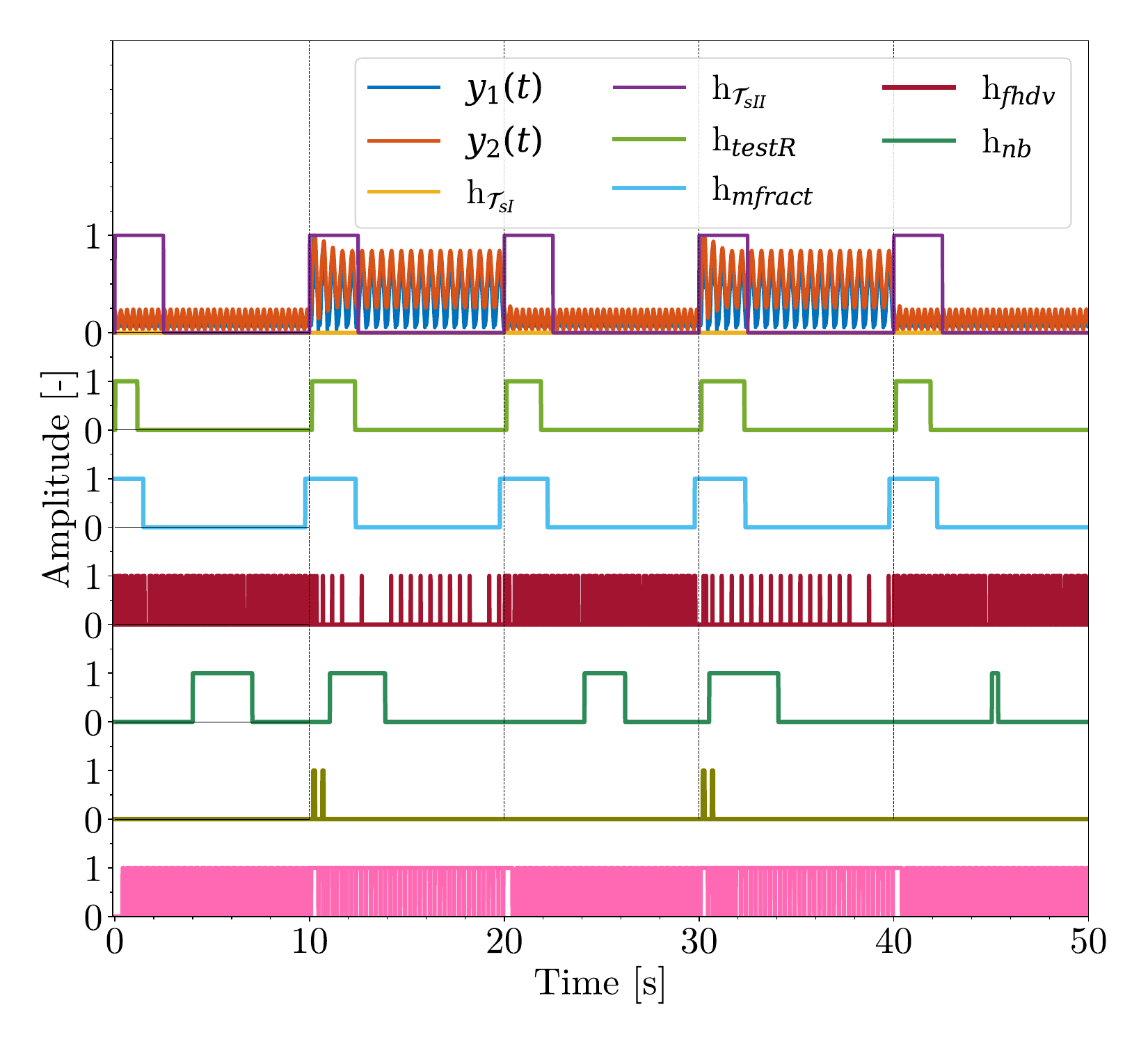}
        \caption{}
        \label{fg:disc_ts_ltcompartison}
    \end{subfigure}

    \caption{Results of the simulated discontinuous dynamic system. (a) Input signals of the simulated discontinuous dynamical system. (b) SMTR curve of the studied discontinuous dynamical system and its geometrical properties. (c) Comparison between the discontinuous dynamical system outputs and the classifiers computed.}
    \label{fg:disc_ts_detection2}
\end{figure*}

The validation assessment of the proposed classifiers as well as the classifiers from the literature for the discontinuous dynamical system are illustrated in Figure~\ref{fg:disc_ts_detection2}(\subref*{fg:disc_ts_ltcompartison}). For the discontinuous dynamical system, $h_{\mc{T}_s I}(t)$ and  $h_{\mc{T}_s II}(t)$ manage to classify all transient regimes time intervals. Among the other classifiers from the literature, $h{mfract}$ successfully classifies all transient regimes, whereas the others fail to classify any regime.

\color{black}
Finally, a reference classifier was computed based on the theoretical dynamic time responses of the considered dynamic systems. Thus, the analyzed classifiers were quantitatively compared with the reference classifier by means of the probability of error of Type I and II\footnote{The procedures for computing the reference classifier $h_{ref}(t)$ and computing probability of Type I and Type II errors  have been detailed in Section~K of the supplementary material, available in \cite{linkannexes}.}. Furthermore, the computational times and memory resources of each classifier were also computed. The probabilities of wrong classification and the computational resources used by each studied classifier are shown in Tables~\ref{tab:simple_tablec} and \ref{tab:resources}, respectively.

\begin{table}[!t]
    \color{black}
    \centering
    \caption{\black{Probability of Type I and type II errors for the tested transient regime classifiers based on the reference classifier $h_{ref}(t)$.}}
    \scalebox{0.78}{
    \begin{tabular}{ l l c c c c c c}
        \hline
        \textbf{Probability} &\textbf{Dynamic system}&$h_{\mathcal{T}_{sI}}$&$h_{\mathcal{T}_{sII}}$&$h_{testR}$&$h_{mfrac}$&$h_{fhdv}$&$h_{nb}$\\
        \hline
        \multirow{3}{0.06\textwidth}{Type I error [\%]}&Linear&\textbf{9.6} &    9.8 &   39.6 &    0.1 &    5.2 &    0.0\\
        &Non-linear&\textbf{2.1} &   20.9 &   95.1 &   17.6 &   99.9 &   91.1\\
        &Discontinuous&0.0 &   \textbf{2.4} &    0.2 &   22.4 &    0.0 &   74.1\\
        \hline
        \multirow{3}{0.06\textwidth}{Type II error [\%]}&Linear&\textbf{7.1} &    0.0 &   21.5 &   79.7 &   74.7 &   91.3\\
        &Non-linear&\textbf{11.4} &    0.0 &    0.0 &   58.9 &    0.2 &    1.5\\
        &Discontinuous&27.1 &   \textbf{14.7}  &   99.4 &   72.5 &   97.4 &   27.8\\
        \hline
    \end{tabular}
    }
    \label{tab:simple_tablec}
\end{table}

\begin{table}[!t]
    \color{black}
    \centering
    \caption{\black{Computational time and memory usage required for the computation of the tested classifiers.}}
    \scalebox{0.72}{
    \begin{tabular}{ l l c c c c c c}
        \hline
        \textbf{Resource}&\textbf{Dynamic system}&$h_{\mathcal{T}_{sI}}$&$h_{\mathcal{T}_{sII}}$&$h_{testR}$&$h_{mfrac}$&$h_{fhdv}$&$h_{nb}$\\
        \hline
        \multirow{3}{0.08\textwidth}{Comp. Time [\si[]{ m\second}]}&Linear&420.00 &  157.00 &  119.00 &  9310.00 & \textbf{12.00} &   15.00\\
        &Non-linear&171.00 &   83.00 &   57.00 &  2575.00 &    \textbf{1.00} &    8.00\\
        &Discontinuous&2347.00 &  1263.00 &  277.00 &  43385.00 &    \textbf{2.00} &   31.00\\
        \hline
        \multirow{3}{0.08\textwidth}{Mem. Usage [\si{MB}]}&Linear&1.03 &    1.76 &    1.68 &    2.16 &    \textbf{0.93} &    1.05\\
        &Non-linear&0.55 &    1.04 &    0.84 &    1.07 &    \textbf{0.4}7 &    0.52\\
        &Discontinuous&1.87 &    3.60 &    2.45 &    3.15 &    \textbf{1.70} &    2.02\\
        \hline
    \end{tabular}
    }
    \label{tab:resources}
\end{table}

The analysis of Table~\ref{tab:resources} reveals that while the classifier $h_{fhdv}$ excels in terms of computational resource utilization, it lags behind other classifiers in transient regime identification for the analyzed dynamic systems. Conversely, the proposed classifier $h_{\mathcal{T}_{sII}}$ demonstrates a faster computational time compared to $h_{\mathcal{T}_{sI}}$ at the expense of increased memory usage.

Furthermore, the findings from Table~\ref{tab:resources} indicate that the proposed classifiers $h_{\mathcal{T}_{sI}}$ and $h_{\mathcal{T}_{sII}}$ exhibit greater resilience to errors in regimen classification compared to $h_{fhdv}$, which is more prone to erroneous classifications under the specified test conditions. Based on the evaluation assessment, the classifier leveraging the arc length of the SMTR curve, denoted as $h_{\mc{T}_{s I}}$, consistently outperformed other classifiers across all studied dynamical systems. 
\color{black}
\section{Concluding Remarks}
\label{sec:Concluding_Remarks}
This article proposed a novel methodology for transient regime classification in dynamical systems, considering the challenges of transient and stationary regime classification algorithms in recent years. The proposed methodology aimed to classify regimes for systems that may contain cyclo-stationary responses, noisy signals, and multiple outputs. Therefore, two transient-regime classifiers were proposed using tools from signal processing, dynamical systems, stochastic processes, and geometry of spatial curves. The proposed classifiers were validated with simulated linear, non-linear, and discontinuous dynamical systems, and they were compared quantitatively with other existing transient regime classifiers in the literature. The results showed that the transient-regime classifier based on the arc-length of the proposed sample moment time representation curve outperformed other classifiers in all scenarios, while the classification of transient regimes based on curvature showed an appropriate performance for the linear and discontinuous dynamical systems.
The proposed methodology has the potential to enhance several applications in diverse fields, such as structural health monitoring, fault detection and diagnosis, and control design for dynamical systems \footnote{\black{An example of an application of this method regarding control design can be found in Section~J of the supplementary material, available in \cite{linkannexes}.}}. It is worth mentioning that the proposed methodology established solely the spatial representation of dynamical systems, utilizing the second sampling moment to ensure weakly stationary conditions. An experimental and theoretical study of the effect of the order of the sampling moment on the transient regime classification will be carried out in future works of this research.
\section{Acknowledgment}
This work has been supported by the Spanish project SEAMLESS: Sustainable learning-based Management of Multi-resource Large-scale Systems (ref. PID2023-148840OB-I00), funded by MCIN/AEI/10.13039/501100011033/FEDER, UE. Also, it was partially supported by the INDUSTR-IA project and the Spanish Ministry of Science and Innovation with the PID2023-148840OB-I00 (SEAMLESS).

\bibliographystyle{unsrt}
\bibliography{./bibliography/transient_Detection.bib}
\vskip 0pt plus -1fil
\begin{IEEEbiography}[{\includegraphics[width=1in,height=1.25in,clip,keepaspectratio]{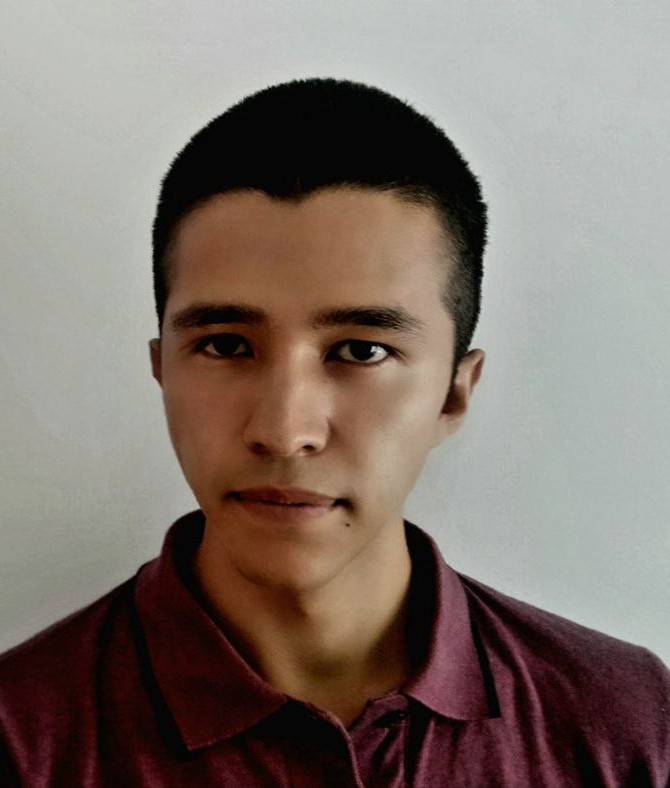}}]{Cristian Puerto-Santana}
	received his bachellor's degree in Mechanical and Electrical Engineer from Universidad de los Andes, Bogot\'{a}, Colombia, in 2016 and 2017, respectively. He obtained his master's degree in  Automation, Electronics and Industrial Control from Universidad  de Deusto, Bilbao, Spain, in 2019. He is currently a doctoral student at Universitat Polit\`{e}cnica de Catalunya, Barcelona, Spain.
\end{IEEEbiography}	
\vskip 0pt plus -1fil
\begin{IEEEbiography}[{\includegraphics[width=1in,height=1.25in,clip,keepaspectratio]{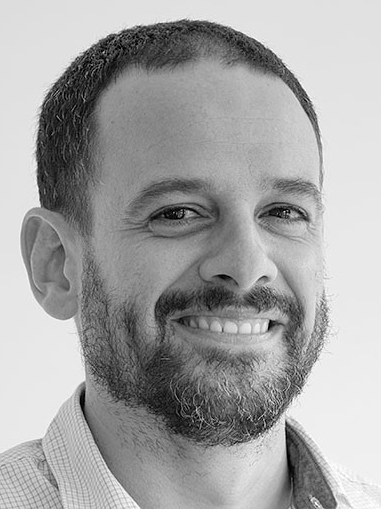}}]{Javier Diaz-Rozo}
Javier Diaz-Rozo holds an M. Eng. in Mechanical Engineering from the University of Los Andes, Colombia (2001) and an M.Sc. in Advanced Manufacturing Technologies and Productive Systems from the University of Manchester, UK (2003). In 2019, he earned a PhD in Artificial Intelligence from Universidad Polit\'ecnica de Madrid, Spain. Prior to becoming the Chief Technology Officer at Aingura IIoT, leading R\&D and industrial data analytics product development, he accrued nearly 20 years of industrial experience. This includes roles such as Project Manager at Ikergune, Senior Consultant in an R\&D consulting firm, R\&D Director in a wind energy business group, and Director for the Advanced Manufacturing Area at ASCAMM Technology Centre.
\end{IEEEbiography}	
	\vskip 0pt plus -1fil
\begin{IEEEbiography}[{\includegraphics[width=1in,height=1.25in,clip,keepaspectratio]{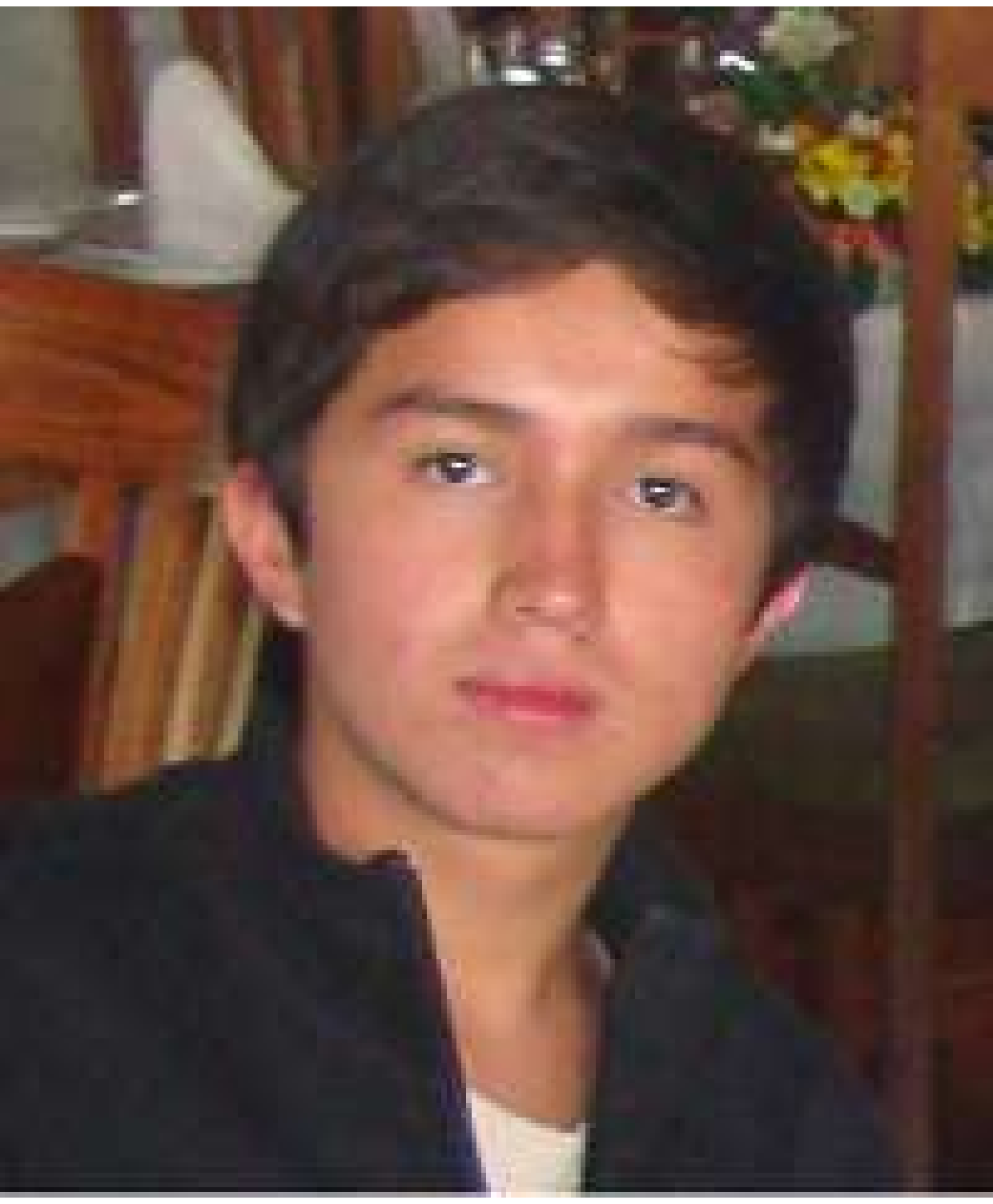}}]{Carlos Puerto-Santana}
received his bachellor's degree in Mathematics from Universidad de los Andes, Bogot\'{a}, Colombia, in 2016 and his master's degree in artificial intelligence from Universidad Polit\'{e}cnica de Madrid, Madrid, Spain, in 2018. He obtained his PhD. degree in 2023 at Universidad Polit\'{e}cnica de Madrid, Madrid, Spain.
\end{IEEEbiography}	
	\vskip 0pt plus -1fil
\begin{IEEEbiography}[{\includegraphics[width=1in,height=1.25in,clip,keepaspectratio]{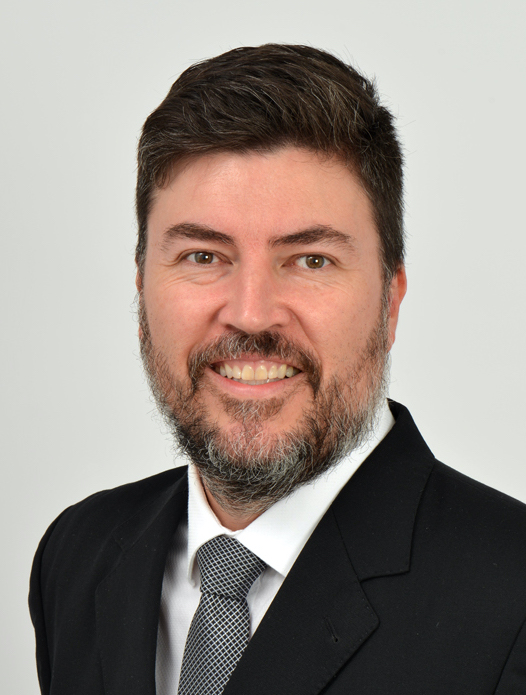}}]{Carlos Ocampo-Martinez}
Carlos Ocampo-Martinez received his Ph.D. degree in Control Engineering from the Universitat Polit\`ecnica de Catalunya - BarcelonaTECH (UPC), Spain. From 2007 to 2010, he held postdoctoral positions at the University of Newcastle (Australia) and at the Institut de Rob\`otica i Inform\`atica Industrial, CSIC-UPC (IRI). Since 2011, he is with UPC, Automatic Control Department (ESAII) as an Associate Professor. His main research interests include constrained model predictive control, large-scale systems management (partitioning and non-centralized control), process control and industrial applications (mainly related to the key scopes of water and energy, and smart manufacturing under the IoT framework).
\end{IEEEbiography}	
\vfill

\newpage
\section*{Appendices}
\setcounter{figure}{0}
\setcounter{table}{0}
\setcounter{equation}{0}
\renewcommand{\thetable}{A.\arabic{table}}
\renewcommand{\thefigure}{A.\arabic{figure}}
\renewcommand{\theequation}{A.\arabic{equation}}

\appendices
\subsection*{A. Linear dynamical system model}
\label{an:a}
A linear continuous-time dynamical system with $l_a \in \mb{N}$ states, $l_u \in \mb{N}$ inputs and $l_y \in \mb{N}$ measured outputs can be written in state-space form as
\begin{align*}
    \dot{\bs{x}}(t)&=A\bs{x}(t)+B\bs{u}(t) + \bs{\upsilon}(t),\\
    \bs{y}(t)&=C\bs{x}(t)+D\bs{u}(t) + \bs{\eta}(t),
\end{align*}
where $A \in \mb{R}^{l_a \times l_a}$ represents the dynamic matrix of the autonomous response of the dynamical system. Whereas $B \in \mb{R}^{l_a \times l_u}$ is the matrix that illustrate how the inputs affects the states of the system, and $C \in \mb{R}^{l_y \times l_a}$ is the matrix that maps the states in the measured outputs of the system. $D \in {R}^{l_y \times l_u}$ is the feedforward matrix which is usually a zero-matrix when the inputs do not affect the output measurements of the system. 

The vector field of the state-space and the outputs measurements are accompanied by unmeasurable zero-mean stationary Gaussian noise $\bs{\upsilon}(t)\sim \mc{N}(0,\Sigma_{\eta}),~\bs{\upsilon}(t) \in \mb{R}^{l_a}$ and $\bs{\eta}(t)\sim \mc{N}(0,\Sigma_{\upsilon}),~\bs{\eta}(t) \in \mb{R}^{l_y}$, respectively.   

Consider a dynamical system written in state-space form as  

\color{black}
\begin{align*}
    \begin{bmatrix}\dot{x}_1(t)\\\dot{x}_2(t)\\\dot{x}_3(t)
\end{bmatrix}&= \begin{bmatrix}  30.0& -102.0&  -60.0\\
                                90.0&  -72.0&  -30.0\\
                                30.0&  -90.0&  -73.2\end{bmatrix}\begin{bmatrix}x_1(t)\\x_2(t)\\x_3(t)\end{bmatrix}\\
&+
\begin{bmatrix}
    1.0& 2.0\\
    0.0&  -1.0\\
    0.0&  1.0
\end{bmatrix}
\begin{bmatrix}u_1(t)\\u_2(t)\end{bmatrix}+\bs{\upsilon}(t),\\
    \begin{bmatrix}y_1(t)\\y_2(t)\\y_3(t)
\end{bmatrix}&= \begin{bmatrix}  0.5& 0.1&  0.6\\
                                1.0& 2.0& 1.0\\
                                0.0& 1.0&-2.0\end{bmatrix}\begin{bmatrix}x_1(t)\\x_2(t)\\x_3(t)\end{bmatrix} + \bs{\eta}(t),\\
                            \end{align*}
\noindent with,

\resizebox{.98\linewidth}{!}{
  \begin{minipage}{\linewidth}
\begin{align*}
    \Sigma_{\eta}= 10^{-4}\begin{bmatrix}  0.13& 0.0&  0.0\\ 0.0& 0.16& 0.0\\0.0& 0.0&0.15\end{bmatrix}, \Sigma_{\upsilon}=10^{-4}\begin{bmatrix}  0.10& 0.0&  0.0\\0.0& 0.33& 0.0\\0.0& 0.0&0.45\end{bmatrix}.
\end{align*}
\end{minipage}
}

The Signal-to-noise ratio of the states and the outputs under the simulated noise is presented in Table~\ref{tab:example2}.
\begin{table}[H]
    \color{black}
    \caption{\black{Values of SNR of $\boldsymbol{x}(t)$ and $\boldsymbol{y}(t)$ for the simulated system.}}
    \centering
    \begin{tabular}{cc}
        \hline
         SNR($\boldsymbol{x}(t)$) [dB] & SNR($\boldsymbol{y}(t)$) [dB] \\
        \hline
         (68, 76, 73) & (76, 75, 70) \\
        \hline
    \end{tabular}
    \label{tab:example2}
\end{table}
\color{black}
The mapping between the states and the outputs assures complete observability of the systems states. Furthermore, for linear systems the observability of the states is equal to the rank of the observability matrix written as 
\begin{equation*}
    \mc{O} = [C~CA~CA^2~\cdots~CA^{l_a-1}]^\intercal , C \in \mb{R}^{l_y \times l_a}, A \in \mb{R}^{l_a\times l_a}.
\end{equation*}
For the analyzed state-space model, rank($\mc{O}$)=3. Then, all the states of the systems can be inferred from the measured outputs. On the other hand, the studied dynamical system has a stable and appropriate dynamic response. The natural frequencies $\bs{f}_n$, damping factors $\bs{\xi}_n$ and time constants $\bs{\tau}_n$ of the system are shown in Table~\ref{tb:linear_param}. The autonomous attractor of the systems is shown in Figure~\ref{fg:linear_attractor}.
\begin{table}[H]
    \caption{Eigenvalue parameters of the linear dynamical system.}
    \centering
    \begin{tabular}{c c} 
     \hline
     Parameter & Value \\ [0.5ex] 
     \hline
      $\bs{f}_n$~[Hz]& [13.57] \\ 
      $\bs{\xi}_n$~[\%] & [55.70] \\
      $\bs{\tau}_n$~[s] & [0.04] \\
     \hline
    \end{tabular}
    \label{tb:linear_param}
    \end{table}

\begin{figure}[H]
	\centering
	\includegraphics[width=0.85\linewidth]{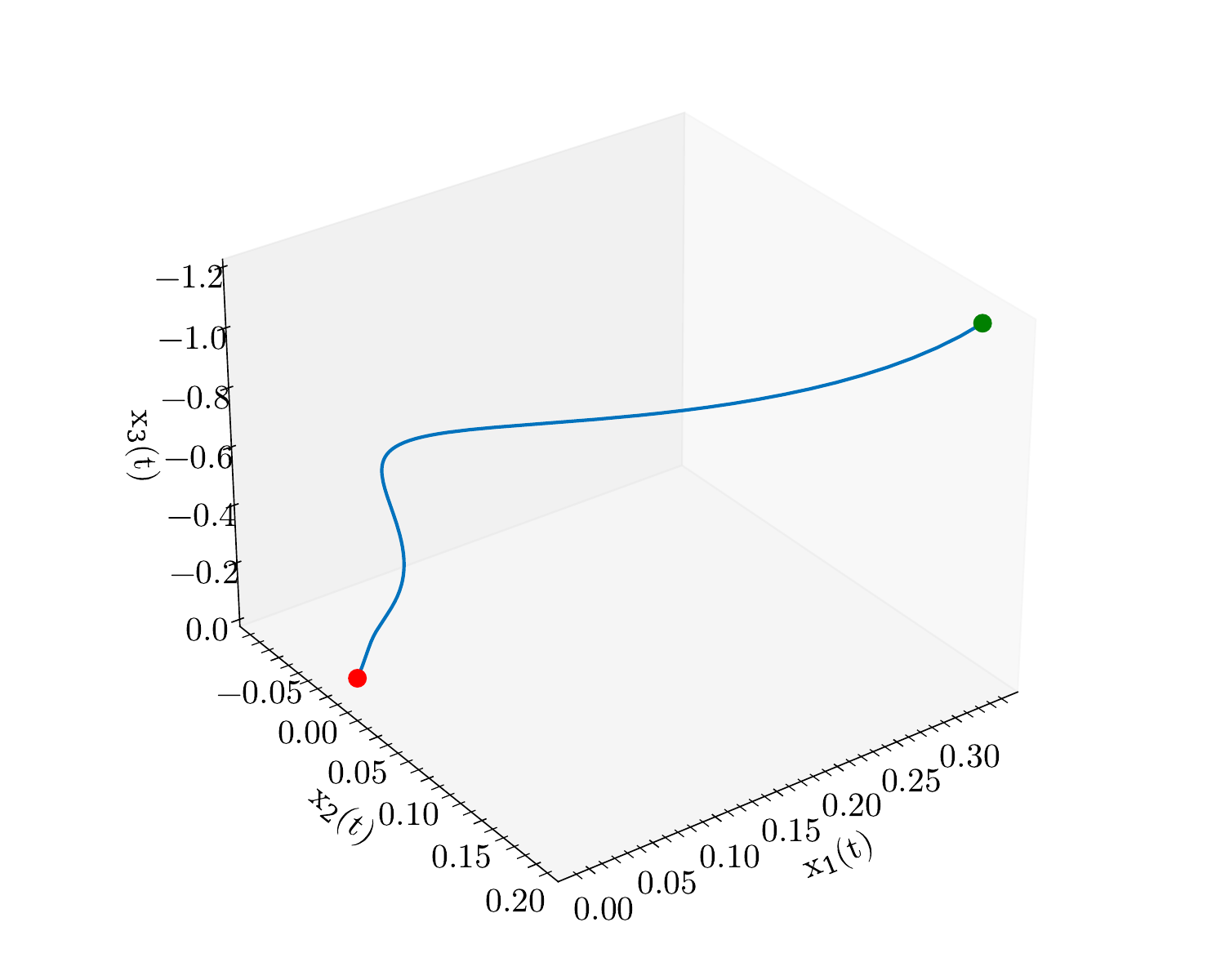}
	\caption{Free response attractor of the studied dynamical system. The green and red dot represent the initial and final state of the simulated dynamical system, respectively.}
	\label{fg:linear_attractor}
\end{figure}
In terms of implementation, the models and simulations were implemented using  Python3.10 with the help of the \texttt{scipy} and \texttt{numpy} standard libraries. The \texttt{scipy} library was used to simulate dynamical systems in their state space representation with the functions \texttt{lsim} for linear systems and \texttt{solveivp} for dynamical systems in general. The noise generation was done with \texttt{numpy} through the \texttt{random.normal} function that allows generating random samples from the mean and variance of a Gaussian distribution.

\setcounter{figure}{0}
\setcounter{table}{0}
\setcounter{equation}{0}
\renewcommand{\thetable}{B.\arabic{table}}
\renewcommand{\thefigure}{B.\arabic{figure}}
\renewcommand{\theequation}{B.\arabic{equation}}
\subsection*{B. Non-linear dynamical system model}
\label{an:b}
A general non-linear and forced dynamical system can be written in state-space form as
\begin{align*}
    \dot{\bs{x}}(t)&=\bs{\alpha}(\bs{x},\bs{u},t)+ \bs{\upsilon}(t),\\
    \bs{y}(t)&=\bs{\gamma}(\bs{x},\bs{u},t) + \bs{\eta}(t),
\end{align*}
where $\bs{\alpha}: \mb{R}^{l_a} \rightarrow \mb{R}^{l_a}$ is the vector field that derives the state transitions function of the dynamical system. Whereas $\bs{\gamma} : \mb{R}^{l_a} \rightarrow \mb{R}^{l_y}$ is the vector field that maps the states of the system into the measured outputs. $\bs{\upsilon}(t)$ and $\bs{\eta}(t)$ are unmeasurable zero-mean Gaussian noise as explained in the linear dynamical system of Section~A. The Lorenz attractor is a dynamical system an example of a non-linear dynamical system, the Lorenz attractor  written as

\color{black}
\begin{align*}
    \begin{bmatrix}\dot{x}_1(t)\\\dot{x}_2(t)\\\dot{x}_3(t)
\end{bmatrix}&= \begin{bmatrix} a_1(x_2(t)-x_1(t)) +u_1(t)\\
    x_1(t)(a_2-x_3(t))-x_2(t)+u_2(t)\\
    x_1(t)x_2(t)-a_3 x_3(t)
\end{bmatrix}+ \bs{\upsilon}(t),\\
\begin{bmatrix}y_1(t)\\y_2(t)\\y_3(t)
\end{bmatrix}&= \begin{bmatrix}  1.0& 0.0&0.0\\
                                0.0& 1.0&0.0\\
                                0.0& 0.0&1.0\\\end{bmatrix}
                                \begin{bmatrix}x_1(t)\\x_2(t)\\x_3(t)\end{bmatrix}  + \bs{\eta}(t),\\
\Sigma_{\eta}&=\Sigma_{\upsilon}= 10^{-4}\begin{bmatrix}  0.95& 0.0&  0.0\\
                                    0.0& 0.95& 0.0\\
                                    0.0& 0.0&1.11\end{bmatrix}.
\end{align*}

This system has the characteristic of having a chaotic oscillating free response for certain parameters of $a_1, a_2, a_3$ and initial conditions $\bs{x}_0$. This system was selected to evaluate the proposed method in scenarios where the analyzed non-linear marginal stable dynamical systems. Note that each measured  output is mapped with each of the dynamical systems states through the identity matrix, therefore, the system is completely observable. The Signal-to-noise ratio of the states and the outputs under the simulated noise is presented in Table~\ref{tab:example22}.
\begin{table}[H]
    \color{black}
    \caption{\black{Values of SNR of $\boldsymbol{x}(t)$ and $\boldsymbol{y}(t)$ for the simulated non-linear system.}}
    \centering
    \begin{tabular}{cc}
        \hline
         SNR($\boldsymbol{x}(t)$) [dB] & SNR($\boldsymbol{y}(t)$) [dB] \\
        \hline
         (59, 52, 48) & (59, 52, 48) \\
        \hline
    \end{tabular}
    \label{tab:example22}
\end{table}
\color{black}

 The autonomous attractor of the systems is shown in Figure \ref{fg:nolinear_attractor}.

\begin{figure}[H]
	\centering
	\includegraphics[width=0.85\linewidth]{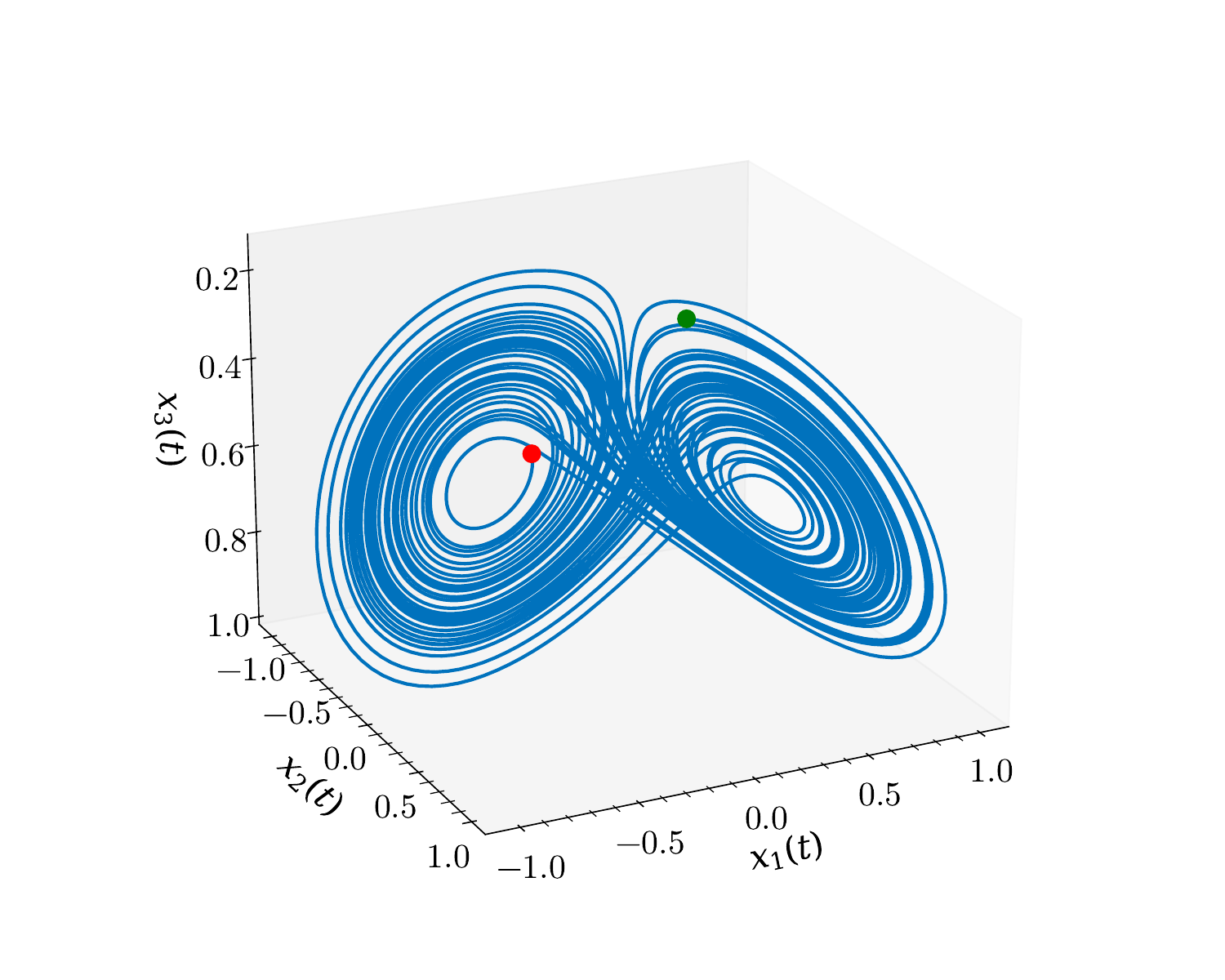}
	\caption{Free response attractor of the studied dynamical system. A double basing is presented with a chaotic behavior($a_1=50.0, a_2=140.0$ and $a_3=13.3$). The green and red dot represent the initial and final state of the dynamical system, respectively.}
	\label{fg:nolinear_attractor}
\end{figure}

\setcounter{figure}{0}
\setcounter{table}{0}
\setcounter{equation}{0}
\renewcommand{\thetable}{C.\arabic{table}}
\renewcommand{\thefigure}{C.\arabic{figure}}
\renewcommand{\theequation}{C.\arabic{equation}}

\subsection*{C. Discontinuous dynamical system model}
\label{an:c}
Discontinuous dynamical systems are characterized by having a vector field that is a discontinuous function of the states, i.e., piece wise functions driven by boundaries on the states. The double discontinuous oscillator is an example of a discontinuous dynamical system described as
\begin{align*}
\begin{bmatrix}\dot{x}_1(t)\\\dot{x}_2(t)\\\dot{x}_3(t)\\\dot{x}_4(t)
\end{bmatrix}&= A\begin{bmatrix}x_1(t)\\x_2(t)\\x_3(t)\\x_4(t)\end{bmatrix}+
\begin{bmatrix}
    0.0& 0.0\\
    0.0&  0.0\\
    0.0&0.0\\
    0.0&  1.0
\end{bmatrix}
\begin{bmatrix}u_1(t)\\u_2(t)\end{bmatrix}+\bs{\upsilon}(t),\\
    \begin{bmatrix}y_1(t)\\y_2(t)
\end{bmatrix}&= \begin{bmatrix}  1.0& 0.0&  0.01&0.0\\
                                0.0& 1.0& 0.0&0.01\\
                                \end{bmatrix}\begin{bmatrix}x_1(t)\\x_2(t)\\x_3(t)\\x_4(t)\end{bmatrix}+ \bs{\eta}(t),\\
\Sigma_{\eta}&=\Sigma_{\upsilon}= 10^{-4}\begin{bmatrix}  2.7& 0.0\\
                                    0.0& 8.5\end{bmatrix}.
\end{align*}

 On the one hand, if $\Phi_1(\bs{x})=\{(x_2(t)\geq0.0) \wedge (x_4(t)>0.0)\}$ (unconstrained). Hence,
\begin{align*}
    A= \begin{bmatrix}  0.0& 0.0&1.0&  0.0\\
        0.0&0.0&0.0&1.0\\
    -200.0&  100.0&  0.0&0.0\\
    0.0&100.0&  0.0&  0.0\end{bmatrix}.
\end{align*} 
On the other hand, if $\Phi_2(\bs{x})=\{((x_1(t)<0.0)\wedge((x_3(t)\leq0.0)\vee(((x_3(t)\leq 0.0)\wedge((x_2(t)<0)\vee$ $((x_2(t)\leq 0.0)\wedge(x_4(t)<0)))))))\}$ (constrained), then,
\begin{align*}
    A= \begin{bmatrix}  0.0&0.0& 1.0&  0.0\\
        0.0&0.0&0.0&1.0\\
    -200.0&  100.0&  0.0&0.0\\
    100.0&  -100.0&  0.0&0.0\end{bmatrix}.
\end{align*} 

This dynamical system represents an oscillator with an obstacle that limits the oscillator movement. Moreover, this system has the characteristic of being marginal stable in both constrained and unconstrained mode. This system has been selected to evaluate the proposed method in scenarios with marginal stable dynamical systems and discontinuous dynamical responses. 

\color{black}
The signal-to-noise ratio of the states and the outputs under the simulated noise is presented in Table~\ref{tab:example22}.
\begin{table}[H]
    \color{black}
    \caption{\black{Values of SNR of $\boldsymbol{x}(t)$ and $\boldsymbol{y}(t)$ for the simulated discontinuous system.}}
    \centering
    \begin{tabular}{cc}
        \hline
         SNR($\boldsymbol{x}(t)$) [dB] & SNR($\boldsymbol{y}(t)$) [dB] \\
        \hline
         (69, 65) & (69, 65) \\
        \hline
    \end{tabular}
    \label{tab:example23}
\end{table}
\color{black}

For both the constrained and unconstrained operations the rank($\mc{O}$)=4. Then, all the states of the systems can be inferred from the measured outputs. The natural frequencies and damping factors of the system are shown in Table~\ref{tb:disc_param}. The autonomous attractor of the systems is shown in Figure \ref{fg:disc_attractor}.

\begin{table}[H]
    \caption{Eigenvalues parameters of the linear dynamical system.}
    \centering
    \begin{tabular}{c c c} 
     \hline
     Parameter & Constrained model & Unconstrained model \\ [0.5ex] 
     \hline
      $\bs{f}_n$~[Hz]& [4.50, 3.18] &[5.15, 1.96]\\ 
      $\bs{\xi}_n$~[\%] & [0.00, 0.00]&[0.00, 0.00]\\
     \hline
    \end{tabular}
    \label{tb:disc_param}
    \end{table}
Since the damping factors of the system are zero, then, the analyzed dynamical system is marginal stable having the poles in the imaginary axis. The permanent oscillatory behavior of the dynamical system is shown in Figure~\ref{fg:disc_attractor}. 
 \begin{figure}[H]
	\centering
	\includegraphics[width=0.85\linewidth]{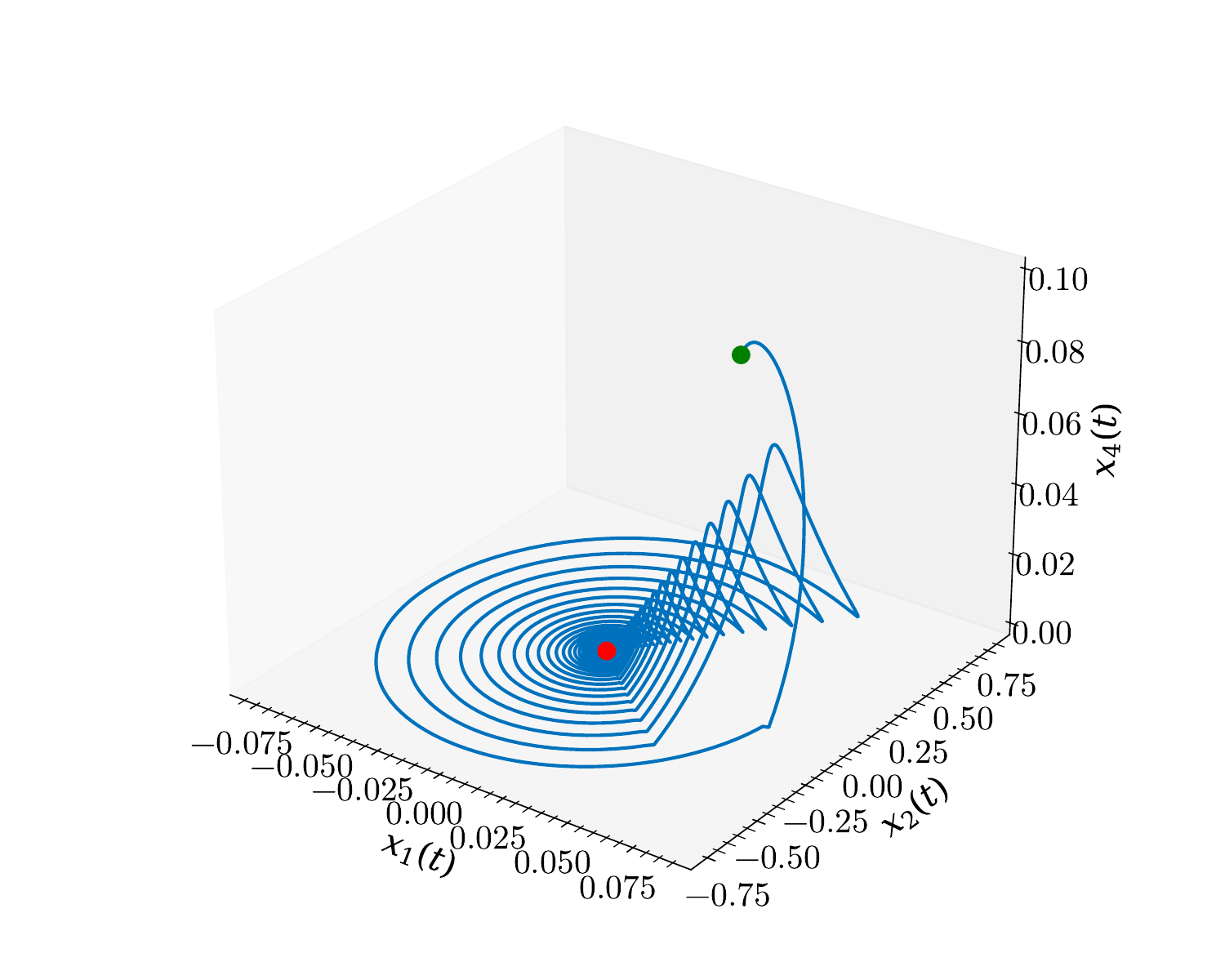}
	\caption{Free response attractor of the studied dynamical system. The dynamical system state-space is bounded in the plane $x_4(t)=0.0$. The green and red dots represent the initial and final state of the dynamical system, respectively.}
	\label{fg:disc_attractor}
\end{figure}
\setcounter{figure}{0}
\setcounter{table}{0}
\setcounter{equation}{0}
\renewcommand{\thetable}{D.\arabic{table}}
\renewcommand{\thefigure}{D.\arabic{figure}}
\renewcommand{\theequation}{D.\arabic{equation}}

\subsection*{D. Parameters used for the evaluated classifiers in the literature}

\begin{table}[H]
	\caption{Parameters used for the $h_{Rtest}$ classifier.}
	\centering
	\small
	\begin{tabular}{c c c c c c c c c c}
		\hline
		Parameter&\multicolumn{3}{c}{Linear}&\multicolumn{3}{c}{Non-linear}&\multicolumn{3}{c}{Discontinuous} \\ [0.5ex]
		\hline
        $R_{critical}$&\multicolumn{3}{c}{[3.0, 0.7]}&\multicolumn{3}{c}{[3.0, 2.0]}&\multicolumn{3}{c}{[20.0, 0.8]} \\ [0.5ex]
        $\lambda_1$&\multicolumn{3}{c}{0.6}&\multicolumn{3}{c}{0.3}&\multicolumn{3}{c}{0.4} \\ [0.5ex]
        $\lambda_2$&\multicolumn{3}{c}{0.6}&\multicolumn{3}{c}{0.3}&\multicolumn{3}{c}{0.4} \\ [0.5ex]
        $\lambda_3$&\multicolumn{3}{c}{0.6}&\multicolumn{3}{c}{0.3}&\multicolumn{3}{c}{0.4} \\ [0.5ex]
		\hline
	\end{tabular}
\end{table}
\begin{table}[H]
	\caption{Parameter used for the $h_{mfractal}$ classifier.}
	\centering
	\small
	\begin{tabular}{c c c c c c c c c c}
		\hline
		Parameter&\multicolumn{3}{c}{Linear}&\multicolumn{3}{c}{Non-linear}&\multicolumn{3}{c}{Discontinuous} \\ [0.5ex]
		\hline
        $threshold$&\multicolumn{3}{c}{0.9}&\multicolumn{3}{c}{0.5}&\multicolumn{3}{c}{0.7} \\ [0.5ex]
		\hline
	\end{tabular}
\end{table}
\begin{table}[H]
	\caption{Parameter used for the $h_{fhvdc}$ classifier.}
	\centering
	\small
	\begin{tabular}{c c c c c c c c c c}
		\hline
		Parameter&\multicolumn{3}{c}{Linear}&\multicolumn{3}{c}{Non-linear}&\multicolumn{3}{c}{Discontinuous} \\ [0.5ex]
		\hline
        $threshold$&\multicolumn{3}{c}{0.2}&\multicolumn{3}{c}{0.0002}&\multicolumn{3}{c}{4} \\ [0.5ex]
		\hline
	\end{tabular}
\end{table}
\begin{table}[H]
	\caption{Parameters used for the $h_{nb}$ classifier.}
	\centering
	\small
	\begin{tabular}{c c c c c c c c c c}
		\hline
		Parameter&\multicolumn{3}{c}{Linear}&\multicolumn{3}{c}{Non-linear}&\multicolumn{3}{c}{Discontinuous} \\ [0.5ex]
		\hline
        $P$&\multicolumn{3}{c}{1.3}&\multicolumn{3}{c}{1.1}&\multicolumn{3}{c}{1.1} \\ [0.5ex]
        $m$&\multicolumn{3}{c}{50}&\multicolumn{3}{c}{200}&\multicolumn{3}{c}{200} \\ [0.5ex]
        $k$&\multicolumn{3}{c}{10}&\multicolumn{3}{c}{250}&\multicolumn{3}{c}{250} \\ [0.5ex]
		\hline
	\end{tabular}
\end{table}

\color{black}
\setcounter{figure}{0}
\setcounter{table}{0}
\setcounter{equation}{0}
\renewcommand{\thetable}{E.\arabic{table}}
\renewcommand{\thefigure}{E.\arabic{figure}}
\renewcommand{\theequation}{E.\arabic{equation}}
\subsection*{E. Calibration procedure of the classifiers parameters}

In practical applications, the proposed algorithm parameters are typically calibrated using a training dataset. To calibrate the threshold of plane curvatures and arc length, the sample moment time representation (SMTR) curve is computed over a training dataset spanning stationary and transient time ranges. Subsequently, the difference between the standard deviation and the mean value of the transient regime geometrical properties $l_{r0}$ and $k_{c0}$ is used as thresholds for transient classification.

Moreover, the selection of the time window $t_h$ is determined through a study of the spectral content of the transient and stationary training dataset. If the analyzed system exhibits cyclo-stationary segments, $t_h$ ideally exceeds the longest period of the stationary cyclical components, ensuring consistent geometrical properties during stationary periods. It is crucial to note that $t_h$ should be minimized to detect small transitions within the system accurately, as a high $t_h$ could hinder the correct detection of consecutive transitions.

For further clarity, the calibration of the proposed method for the linear system presented in Section~A is outlined next.

Figure~\ref{fig:cali_panc} depicts the transient regime and geometrical properties of the SMTR curve during a training transient segment. Thus, the threshold for the arc length and curvature are computed as $l_{r0}=\sigma_{l_r}-\mu_{l_r}$ and $k_{c0}=\sigma_{k_c}-\mu_{k_c}$, where $\sigma_{x}$ and $\mu_{x}$ are the standard deviation and mean value of the evaluated geometrical property $x$ during the transient regime.

    \begin{figure}[H]
        \centering
        \includegraphics[width=0.45\textwidth]{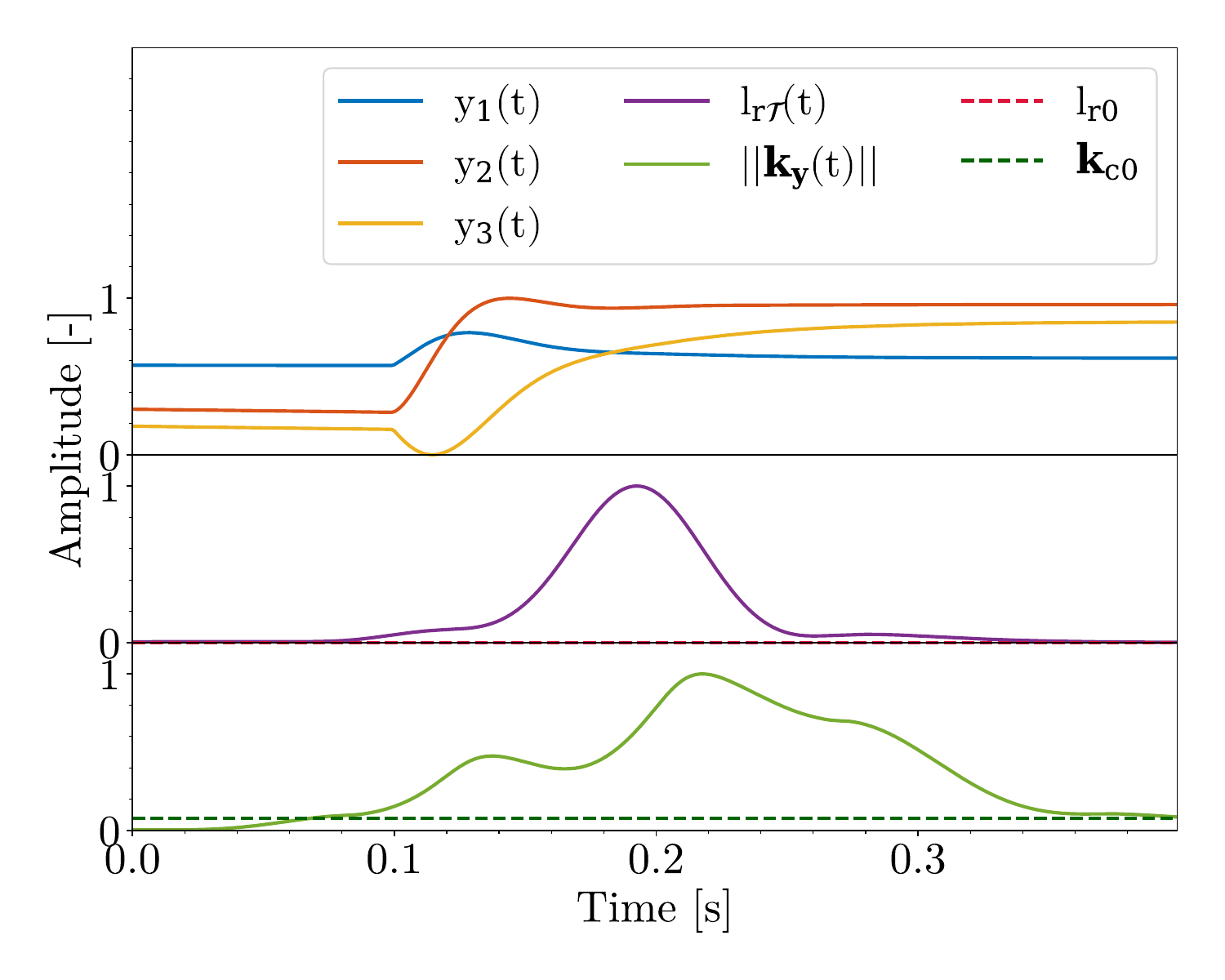}
        \caption{Calibration of the SMRT curve arc length and plane curvature threshold for a transient regime. Training transient regime for the threshold calibration of the geometrical properties.}    
        \label{fig:cali_panc}
    \end{figure}

    \begin{figure}[H]
        \centering
        \includegraphics[width=0.45\textwidth]{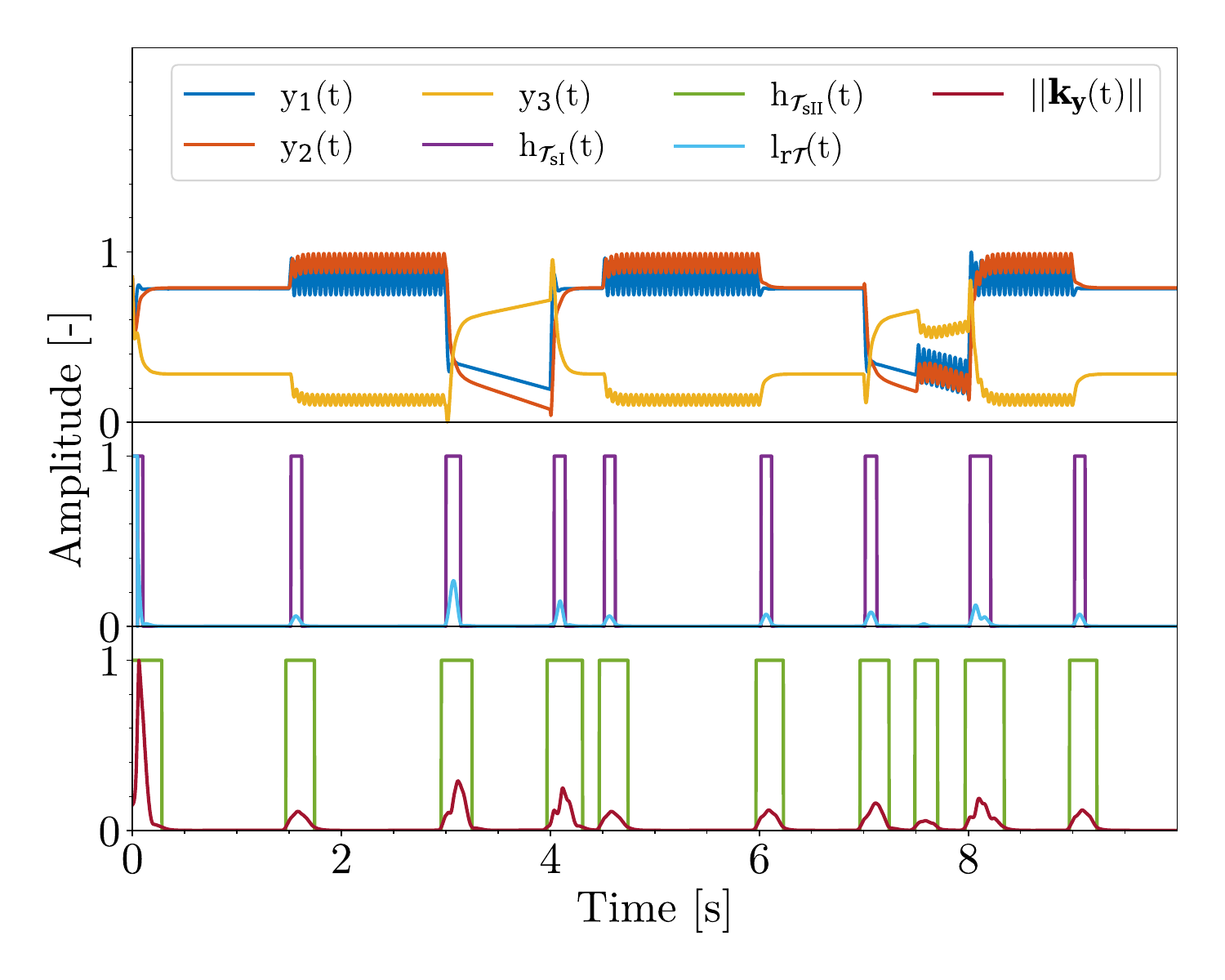}
        \caption{Calibration of the SMRT curve arc length and plane curvature threshold for a transient regime. Performance of the classifiers $h_{\mathcal{T}_{s_I}}$ and $h_{\mathcal{T}_{s_II}}$ for the calibrated thresholds.}    
        \label{fig:cali_panc2}
    \end{figure}

    To calibrate the time window $t_h$, the duration of the transient regime $t_{ss}$ is analyzed. Then, $t_{h}<t_{ss}$ based on Figure \ref{fig:cali_panc} and~\ref{fig:cali_panc2}. Nevertheless, for this application, there are cyclo-stationary periods of time. Therefore, it is necessary to use a training dataset of the cyclo-stationary regime to identify the longest period within the signal. A sample of the time and frequency domain characteristics of the cyclo-stationary regime is depicted in Figure~\ref{fig:ciclo} and~\ref{fig:ciclo2}. The longest period of the cyclo-stationary signal corresponds to the reciprocal value of the lowest frequency component $f_{lw}$ within the frequency domain information of the regime, i.e., $t_h>\frac{1}{f_{lw}}$.

    \begin{figure}[H]
        \centering
        \includegraphics[width=0.45\textwidth]{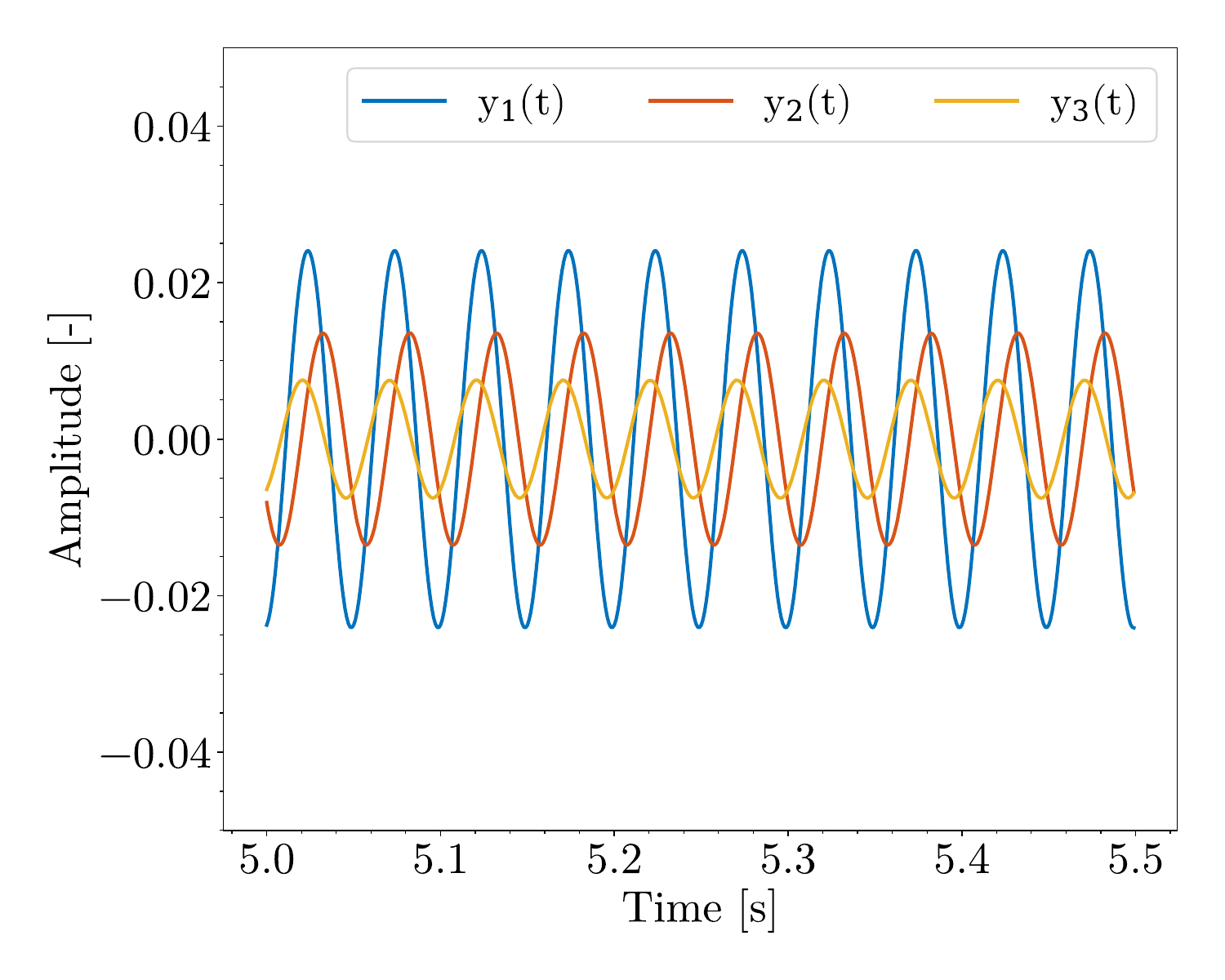}
        \caption{Time and frequency domain characteristics of the monitored outputs during their cyclo-stationary period of time. Time domain response of the outputs under cyclical input excitations.}
        \label{fig:ciclo}   
    \end{figure}
    \begin{figure}[H]
        \centering
        \includegraphics[width=0.45\textwidth]{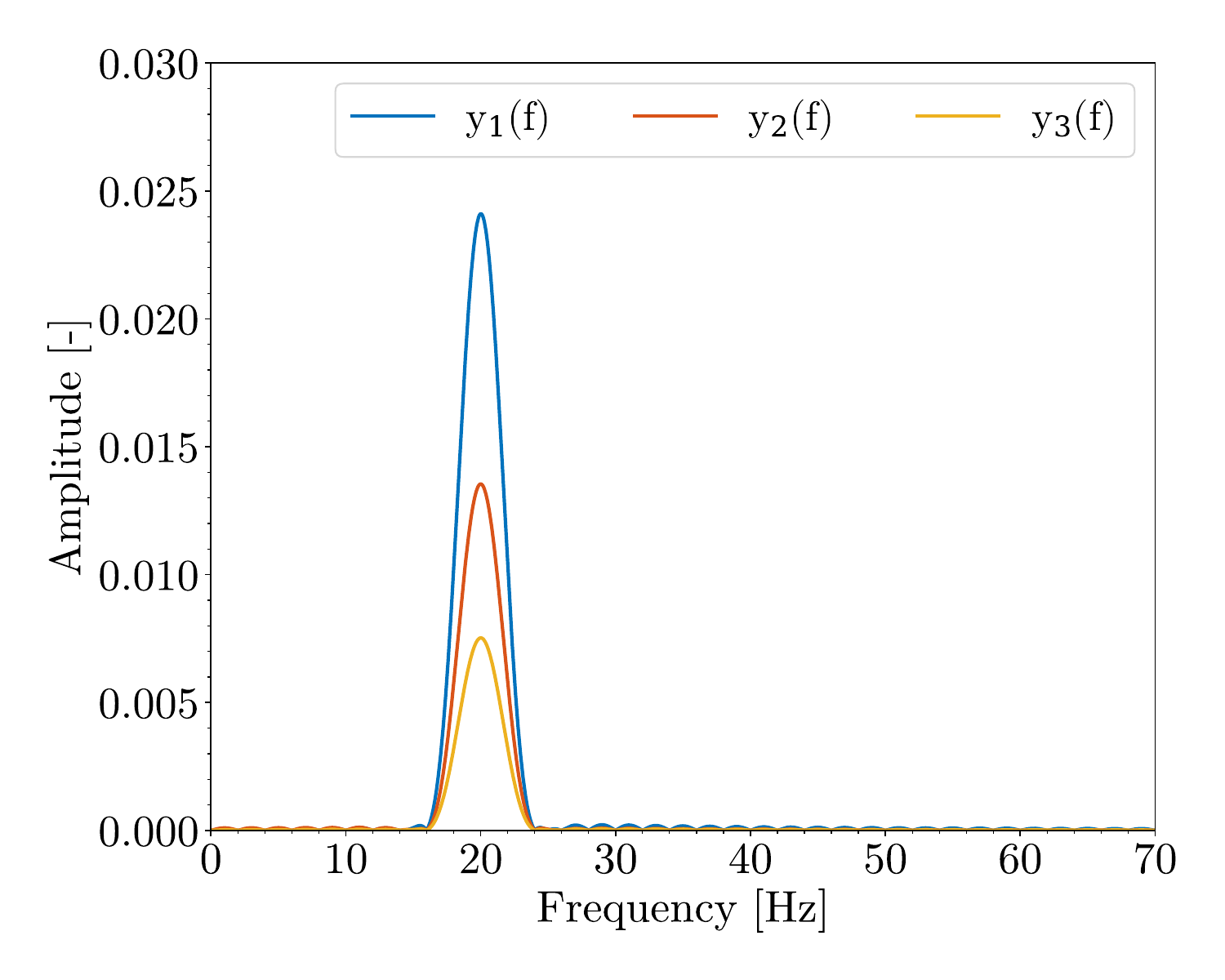}
        \caption{Cyclo-stationary time and frequency domain characteristics from the outputs of the studied dynamic system. Frequency domain content of the systems outputs under cyclical input excitations.}
        \label{fig:ciclo2}   
    \end{figure}

    In summary, it is recommended to adjust the parameters of the proposed approach by employing a training dataset that captures both the transient and stationary characteristics of the analyzed system. Subsequently, statistical analyses can be used to fine-tune the thresholds of the arc length and curvature for transient detection. Meanwhile, the choice of moving window size should be determined based on the presence of periodic patterns within the system being analyzed.

    \color{black}

\color{black}
\setcounter{figure}{0}
\setcounter{table}{0}
\setcounter{equation}{0}
\renewcommand{\thetable}{F.\arabic{table}}
\renewcommand{\thefigure}{F.\arabic{figure}}
\renewcommand{\theequation}{F.\arabic{equation}}

\subsection*{F. Derivation of numerical approximation of the arc length and curvatures of the SMTR curve}

Consider a dynamic system monitored by a group of sensors with constant sampling time $t_s$, then, the signals $\boldsymbol{y}(t) \in \mathbb{R}^{l_y}$ can be interpreted as discrete stochastic processes such that $t\approx nt_s, n\in \mathbb{Z_{+}}$. Furthermore, consider the analysis of the studied system under a moving window of size $l_h$ consisting of the current sensors measurement and the previous stored measurements as depicted in Figure~\ref{fig:sampled}.

Consider a dynamic system monitored by a group of sensors with a constant sampling time $t_s$. Then, the signals $\boldsymbol{y}(t) \in \mathbb{R}^{l_y}$ can be interpreted as discrete stochastic processes such that $t\approx nt_s$, where $n\in \mathbb{Z_{+}}$. Furthermore, consider the studied system under a moving window of size $l_h$, which comprises the current sensor measurements along with the previously stored measurements, as illustrated in Figure~\ref{fig:sampled}.
\begin{figure}[H]
    \centering
    \includegraphics[width=0.45\textwidth]{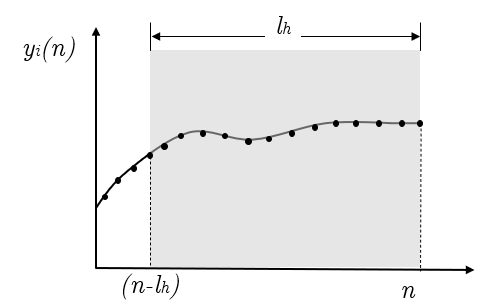}
    \caption{Example of the data sampling and transient regime classification in the sensor $y_i(n)$ for a moving window of size $l_h$. }
    \label{fig:sampled}
\end{figure}
The $p$-th sample moment of each signal in the studied window centered in $y_{i_0}$ with $i=\{0,1...,l_y-1\}$ can be written as
\begin{equation}
    \hat{\beta}_{\left(p,y_{i_0}\right)}(y_{i},n) = \frac{1}{l_h}\sum_{k=n-l_h}^{n}(y_i(k)-y_{i_0})^p,
    \label{eq:dist_moment}
\end{equation}

\noindent where $l_h \in \mathbb{Z}_{+}$ is the length of data analyzed for a moving window. Assuming that the probability distribution of the measurements within $l_h$ is uniform, \eqref{eq:dist_moment} is a discrete approximation of $\hat{\beta}_{\left(p,y_{i_0}\right)}(y_{i},t)$ for a duration time $t_h=l_ht_s$. 

The proposed SMTR $r_{\mathcal{G}}(t) \in \mathbb{R}^{l_y+1}$ written as $r_{\mathcal{G}}(t)=[t,\hat{\beta}_{\left(p,y_{0_0}\right)}(y_{0},t),\hat{\beta}_{\left(p,y_{1_0}\right)}(y_{1},t),\\
\dots,\hat{\beta}_{\left(p,y_{l_y-1}\right)}(y_{{l_y-1}_0},t)]$ is intended to be used as an indicator of transient and stationary regimes of dynamic systems. The presented classification is performed with both arc length and the norm of the plane curvatures of the SMTR. An analysis of both geometrical properties based on numerical approximation is next shown
\begin{itemize}
    \item Arc length:
    The arc length $l_r(t)$ of $r_{\mathcal{G}}(t)$ within the interval $[t-t_h,t]$ expressed as:
    
    \begin{equation}
        l_r(t) = \int_{t-t_h}^{t} \lVert r^{(1)}_{\mathcal{G}}(t)\rVert dt.
        \label{eq:lr}
    \end{equation}
    The time derivative in \eqref{eq:lr} can be expressed as

    \begin{equation}
        \resizebox{.99\linewidth}{!}{$r^{(1)}_{\mathcal{G}}(t) = [1,\hat{\beta}^{(1)}_{\left(p,y_{0_0}\right)}(y_{0},t),\hat{\beta}^{(1)}_{\left(p,y_{1_0}\right)}(y_{1},t),\dots,\hat{\beta}^{(1)}_{\left(p,y_{{l_y-1}_0}\right)}(y_{l_y-1},t)].$}
        \label{eq:rdiff}
    \end{equation}

    By means of finite difference within the studied time interval, the first time derivative of each sample moment in \eqref{eq:rdiff} can be approximated as

\resizebox{.99\linewidth}{!}{
  \begin{minipage}{\linewidth}
    \begin{subequations}
        \label{eq:lrRiemann}
        \begin{align*}
        \hat{\beta}^{(1)}_{\left(p,y_{i_0}\right)}(y_{i},t) &\approx \frac{\hat{\beta}_{\left(p,y_{i_0}\right)}(y_{i},n)-\hat{\beta}_{\left(p,y_{i_0}\right)}(y_{i},n-1)}{t_s}\\
        &\approx \frac{\sum_{k=n-l_h}^{n}(y_i(k)-y_{i_0})^p - \sum_{k=n-1-l_h}^{n-1}(y_i(k)-y_{i_0})^p}{t_s*l_h}\\
        & \approx \frac{(y_i(n)-y_{i_0})^p-(y_i(n-\Delta_h)-y_{i_0})^p}{t_h}.
    \end{align*}
    \end{subequations}
\end{minipage}
    }

    With $\Delta_h=l_h+1$ and $i=\{0,1,\dots,l_y-1\}$. Therefore, the norm of the first derivative of $r_{\mathcal{G}}(t)$ can be written as

    \resizebox{.99\linewidth}{!}{
        \begin{minipage}{\linewidth}
    \begin{subequations}
    \begin{align*}
        \lVert r^{(1)}_{\mathcal{G}}(t) \rVert &\approx \sqrt[]{1+\frac{1}{t_h^2}\sum_{i=0}^{l_y-1}\left((y_i(n)-y_{i_0})^p-(y_i((n-\Delta_h))-y_{i_0})^p\right)^2},\\
        &\approx \sqrt[]{1+\frac{1}{t_h^2} \lVert (\boldsymbol{y}(n)-\boldsymbol{y}_{0})^p-(\boldsymbol{y}((n-\Delta_h))-\boldsymbol{y}_{0})^p\rVert^2}.
        \label{eq:lrRiemann1}
    \end{align*}
    \end{subequations}
\end{minipage}
}

    Finally, the arc length described in \eqref{eq:lr} can be approximated as
    \resizebox{.9\linewidth}{!}{
        \begin{minipage}{\linewidth}
    \begin{subequations}
        \begin{align}
        l_r(t)\approx& \frac{1}{l_h}\sum_{k=n-l_h}^{n} \lVert r^{(1)}_{\mathcal{G}}(k) \rVert,\\
        \approx& \frac{1}{l_h}\sum_{k=n-l_h}^{n}\sqrt[]{t_h^2+\lVert (\boldsymbol{y}(k)-\boldsymbol{y}_{0})^p-(\boldsymbol{y}(k-\Delta_h)-\boldsymbol{y}_{0})^p\rVert^2}.
        \end{align}
        \label{eq:lrRiemann}
    \end{subequations}
    \end{minipage}
    }

    A numerical approximation of the relationship between the acquired output signals and the arc length of the proposed spatial curve is shown in (\ref{eq:lrRiemann}b). For the purpose of clarity, consider the case where $p=1$ and $\boldsymbol{y}_0=\emptyset$, for that scenario the arc length is written as
    \begin{equation}
        l_r(t) \approx \frac{1}{l_h}\sum_{k=n-l_h}^{n}\sqrt[]{t_h^2+\lVert(\boldsymbol{y}(k))-(\boldsymbol{y}(k-\Delta_h))\rVert^2},
        \label{eq:lrapp}
    \end{equation}

    It is proposed to use $l_r(t)-t_h$ to generate a transition classifier. For the case illustrated in \eqref{eq:lrapp}, $l_r(t)-t_h=0$ when $\lVert(\boldsymbol{y}(k))-(\boldsymbol{y}(k-\Delta_h))\rVert^2=0~\forall k \in\{n-l_h,n\}$. Therefore, if the norm of the difference between all the intervals of $\boldsymbol{y}(k)$ and $\boldsymbol{y}(k-\Delta_h)$ are null, then the output signals will be invariant in said intervals, the arc length will be $t_h$, and the system will be in a stationary regime.

    It is important to mention that the methodology has been proposed with $p=2$ and $\boldsymbol{y}_{0}= \emptyset$. These parameters generate the computation of the second sampled raw moment. In the field of signal processing, the use of this moment is widely used for the characterization of periodic signals as exposed in \cite{petrovic2012root, van1991periodic, poomjan2013accurate}. Increasing $p$ for the sample moment computation makes the smaller transient fluctuations negligible and highlights the more abrupt changes. An experimental and theoretical study of the effect of the order of the sampling moment will be carried out in future research.

    \item Plane curvatures:
    The curvature of each plane curve $[t,\hat{\beta}_{\left(p,y_{j_0}\right)}(y_{j},t)]$can be written as:
    
    \begin{equation}
        \resizebox{.99\linewidth}{!}{$ k_{y_j}(t) =\frac{\left| \hat{\beta}^{(2)}_{\left(p,y_{j_0}\right)}(y_{j},t) \right|}{\left( 1+\left(\hat{\beta}^{(1)}_{\left(p,y_{j_0}\right)}(y_{j},t) \right)^2 \right)^{\frac{3}{2}}},~j=0,1,\dots,l_y-1.$}
        \label{eq:kurr}
    \end{equation} 

    By means of finite differences, the plane curvature in \eqref{eq:kurr} can be approximate as
    
    \begin{subequations}
        \label{eq:lrRiemannq}
        \begin{minipage}{\linewidth}
            \resizebox{0.95\linewidth}{!}{
                \begin{minipage}{\linewidth}
        \begin{align}
        &k_{y_j}(n) \approx t_hl_h\left(\frac{\left|\Delta y^p_j(n,l_h,y_{j0}) -\Delta y^p_j(n-1,l_h,y_{j0})\right|}{\left( (t_h^2+\Delta y^p_j(n,l_h,y_{j0})^2)\right)^{3/2}}\right),\\
        &\Delta y^p_j(n,l_h,y_{j0}) =  (y_j(n)-y_{j0})^p-(y_j(n-l_h-1)-y_{j0})^p,
        \end{align}
    \end{minipage}
    }
  \end{minipage}
    \end{subequations} 

    Consider the case where $p=1$ and $\boldsymbol{y}_0=\emptyset$, for that scenario the norm of the plane curvature vector is written as

    \begin{subequations}
        \label{eq:kurrt}
        \begin{minipage}{\linewidth}
            \resizebox{0.95\linewidth}{!}{
                \begin{minipage}{\linewidth}
        \begin{align}
        &\lVert\bs{k}_{\bs{y}}(n)\rVert \approx t_hl_h\sqrt{\sum_{j=1}^{l_y}\left(\frac{\left|\Delta y_j(n,l_h) -\Delta y_j(n-1,l_h)\right|}{\left( (t_h^2+\Delta y_j(n,l_h)^2)\right)^{3/2}}\right)^2},\\
        &\Delta y_j(n,l_h) =  (y_j(n))-(y_j(n-l_h-1)),
        \end{align}
    \end{minipage}
    }
  \end{minipage}
    \end{subequations} 

It is proposed to use $\lVert\bs{k}{\bs{y}}(n)\rVert $ for designing a transition classifier. In the scenario outlined in equation \eqref{eq:kurrt}, $\lVert\bs{k}{\bs{y}}(n)\rVert =0$ when the difference between the intervals $[n-l_h-1,n]$ and $[n-l_h-2,n-1]$ converges to zero across all outputs of the monitored dynamic system. Consequently, the output signals remain unchanged within these intervals, leading to the curvature norm approaching zero and indicating a stationary regime.

It is crucial to apply the same considerations to the parameters $p$ and $\boldsymbol{y}_0$. When dealing with this geometrical property, increasing $p$ for sample moment computation renders smaller differences negligible and accentuates the more abrupt changes in curvature norms.

Conversely, this geometric property is contingent on the first and second numerical time derivatives of the sample moments. Consequently, employing smoothing filters, such as the Savitzky-Golay filter outlined in \cite{gorry1990general}, becomes imperative. These filters facilitate the computation of numerical time derivatives of signals perturbed by noise. Furthermore, their application is crucial for precisely identifying the regions where curvature becomes zero, indicating stationary regimes, while disregarding the inflection points in the sample moments.

\end{itemize}
\color{black}

\color{black}
\setcounter{figure}{0}
\setcounter{table}{0}
\setcounter{equation}{0}
\renewcommand{\thetable}{G.\arabic{table}}
\renewcommand{\thefigure}{G.\arabic{figure}}
\renewcommand{\theequation}{G.\arabic{equation}}

\subsection*{G. Study of the effect of observability in the classification performance}
Consider the nonlinear systems written as
    \begin{align*}
        \dot{\boldsymbol{x}}(t)&=\boldsymbol{\alpha}(\boldsymbol{x}(t),\boldsymbol{u}(t))\\
        \boldsymbol{y}(t)&=\boldsymbol{\gamma}(\boldsymbol{x}(t),\boldsymbol{u}(t)),        
    \end{align*}

    \noindent techniques such as the empirical observability Gramian, as discussed in \cite{powel2015empirical}, and Sedoglavic's algorithm for observability determination of rational functions, presented in \cite{sedoglavic2001probabilistic}, can be employed to estimate the observability characteristics of the system. However, it is worth to mention that these methods may incur significant computational costs in real-time applications. Hence, alternatives such as linearization around an operating point can be utilized to assess this condition during online analysis.
    
    The observability matrix of a linear or linearized system around an operating point is expressed as
    
    \begin{equation}
        \mathcal{O} = \begin{bmatrix}
            J_{\boldsymbol{\gamma}}(\boldsymbol{x}_0,\boldsymbol{u}_0)\\
            J_{\boldsymbol{\gamma}}(\boldsymbol{x}_0,\boldsymbol{u}_0)J_{\boldsymbol{\alpha}}(\boldsymbol{x}_0,\boldsymbol{u}_0)\\
            J_{\boldsymbol{\gamma}}(\boldsymbol{x}_0,\boldsymbol{u}_0){J^2_{\boldsymbol{\alpha}}(\boldsymbol{x}_0,\boldsymbol{u}_0)}\\
        \end{bmatrix},
    \label{eq:obs}
    \end{equation}

    \noindent where, $J_{\boldsymbol{\alpha}}(\boldsymbol{x}_0,\boldsymbol{u}_0)$ and $J_{\boldsymbol{h}}(\boldsymbol{x}_0,\boldsymbol{u}_0)$ are the Jacobian matrices of the nonlinear mappings $\boldsymbol{\alpha}(\boldsymbol{x}(t),\boldsymbol{u}(t))$ and $\boldsymbol{\gamma}(\boldsymbol{x}(t),\boldsymbol{u}(t))$, respectively.

    For instance, consider the studied Lorenz system and the nonlinear mapping written as
    \begin{subequations}
    \begin{align}
        \begin{bmatrix}\dot{x}_1(t)\\\dot{x}_2(t)\\\dot{x}_3(t)
    \end{bmatrix}&= \boldsymbol{\alpha}(\boldsymbol{x,\boldsymbol{u}})=\begin{bmatrix} a_1(x_2(t)-x_1(t)) +u_1(t)\\
        x_1(t)(a_2-x_3(t))-x_2(t)+u_2(t)\\
        x_1(t)x_2(t)-a_3 x_3(t)
    \end{bmatrix},\\
    \begin{bmatrix}y_1(t)\\y_2(t)\\y_3(t)
    \end{bmatrix}&= \boldsymbol{\gamma}(\boldsymbol{x})=\begin{bmatrix}x_{1} \log{\left(x_{2} \right)}\\
        \frac{x_{2}^{2}}{x_{3}^{2} + 1}\\
        x_{3} e^{\sin{\left(x_{2} \right)}}
    \end{bmatrix}.  
    \end{align}
    \label{eq:nonmap}
    \end{subequations}
    Under the mappings of \eqref{eq:nonmap}, the Jacobian matrices can be expressed as
    \begin{subequations}
        \begin{align*}
        J_{\boldsymbol{\gamma}}(\boldsymbol{x})&=
        \begin{bmatrix}
        \log{\left(x_{2} \right)} & \frac{x_{1}}{x_{2}} & 0\\
        0 & \frac{2 x_{2}}{x_{3}^{2} + 1} & - \frac{2 x_{2}^{2} x_{3}}{\left(x_{3}^{2} + 1\right)^{2}}\\
        0 & x_{3} e^{\sin{\left(x_{2} \right)}} \cos{\left(x_{2} \right)} & e^{\sin{\left(x_{2} \right)}}
        \end{bmatrix},\\
        J_{\boldsymbol{f}}(\boldsymbol{x})&=
        \begin{bmatrix}
             -a_{1} & a_{1} & 0\\
            a_{2} - x_{3} & -1 & - x_{1}\\
            x_{2} & x_{1} & - a_{3}
        \end{bmatrix}.
    \end{align*}
        \end{subequations}   

        Thus, the observability matrix $\mathcal{O}$ of the system can be computed with \eqref{eq:obs}, the parameters of the system $[a_1,a_2,a_3]=[560.0,200.0,53.3]$ and the operation point $\boldsymbol{x}_0=[102,102,199]$, $\boldsymbol{u}_0=[0,0]$. For this scenario, $\mathcal{O}$ is full rank and therefore, the system is completely observable. Furthermore, the geometrical properties of the SMTR curve are sensible to the transient regimes of the system. The results of the classification for the non-linear mapping are shown in Figure~\ref{fig:nonmapingdd}.
        \begin{figure}[H]
        \centering
        \begin{subfigure}{0.45\textwidth}
            \includegraphics[width=\textwidth]{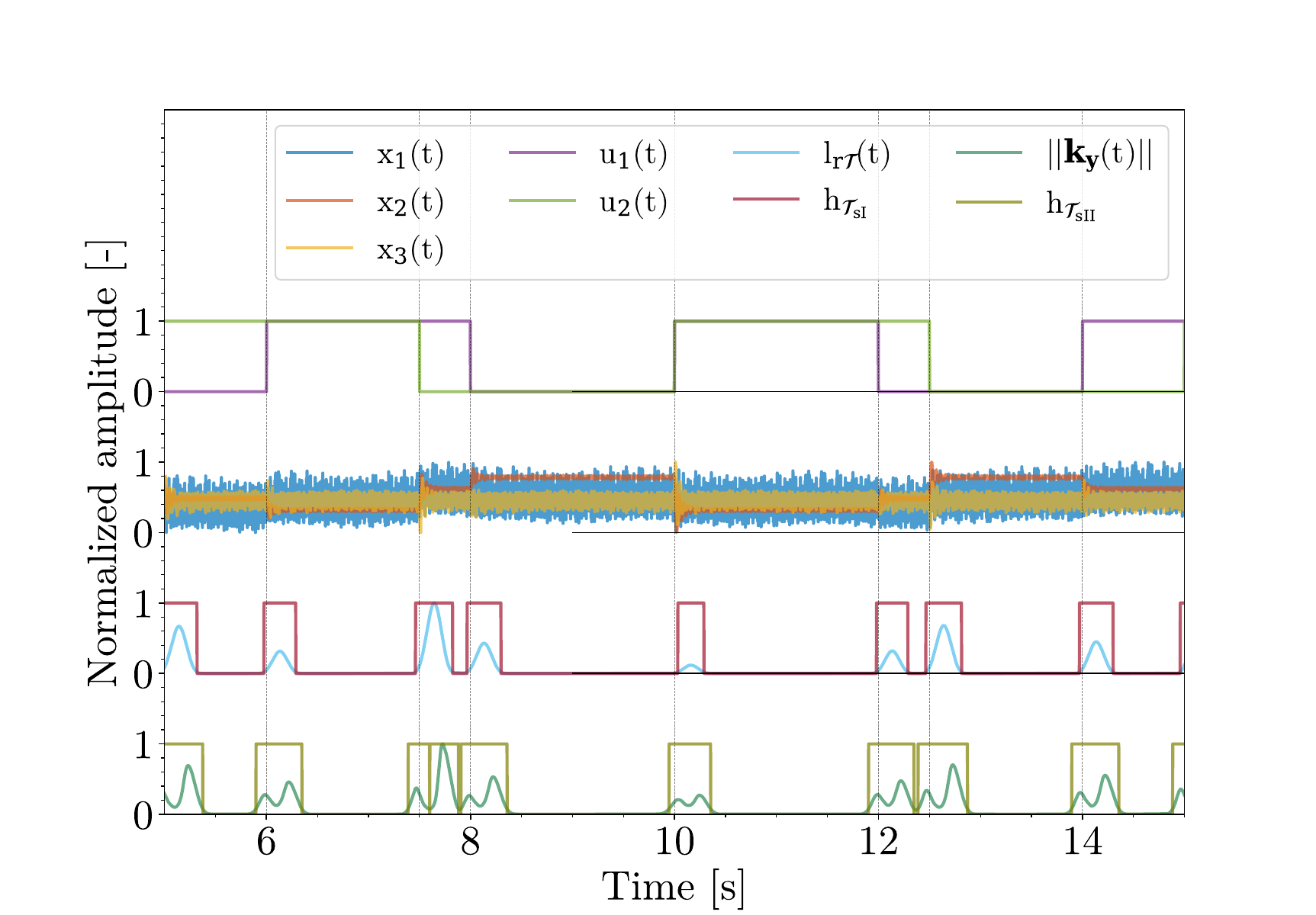}
            \caption{}
        \end{subfigure}
    \end{figure}
    \begin{figure}[H]\ContinuedFloat
        \centering
        \begin{subfigure}{0.45\textwidth}
            \includegraphics[width=\textwidth]{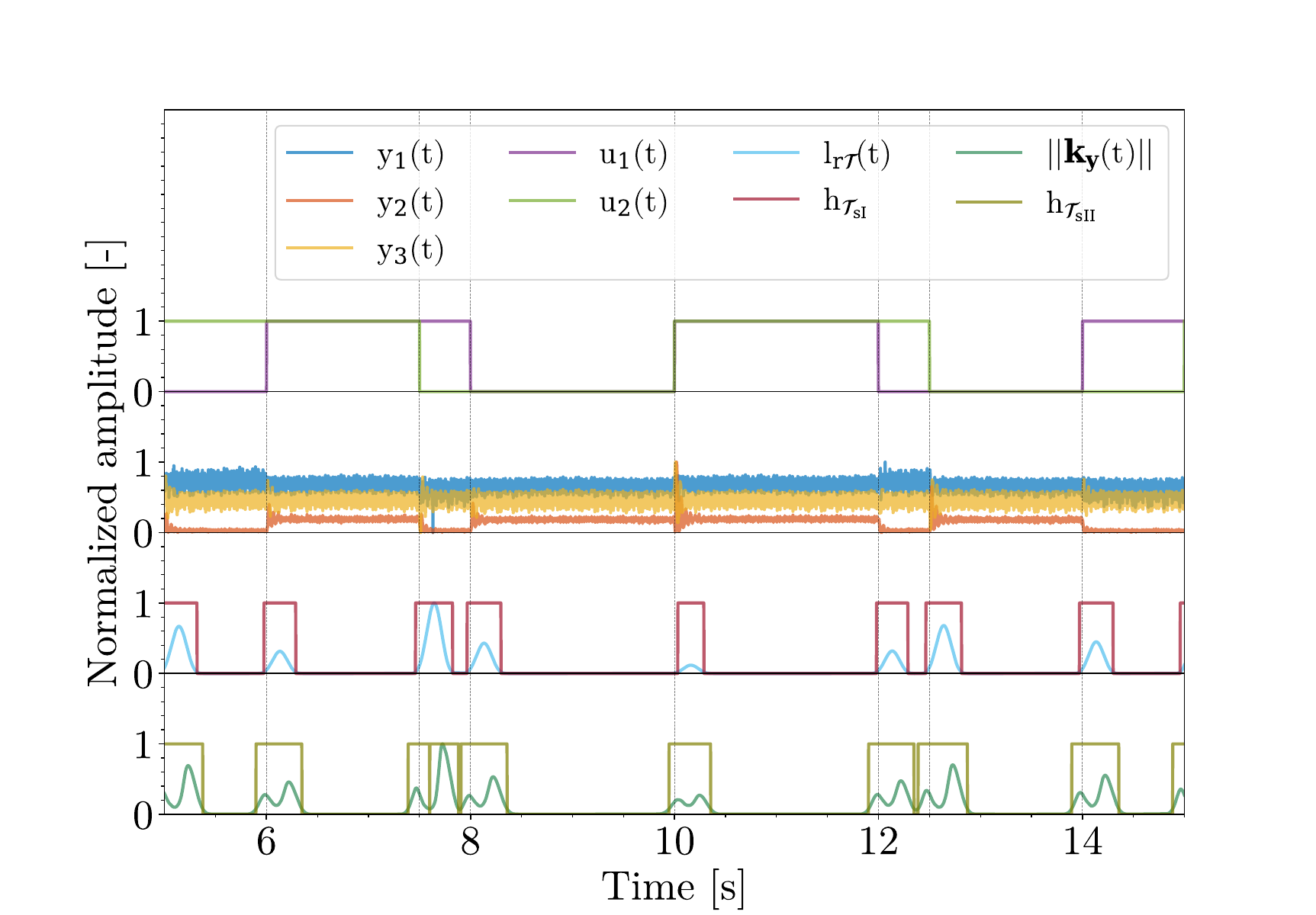}
            \caption{}
        \end{subfigure}
        \caption{Classification of the transient regimes using the proposed methodology for the observable mapping of \eqref{eq:nonmap}. (a) Time response of the inputs, states, geometrical properties of the SMTR curve and proposed classifier. (b) Time response of the inputs, outputs, geometrical properties of the SMTR curve and proposed classifier.}
        \label{fig:nonmapingdd}
        \end{figure}

    Now, consider the non-linear output-state mapping written as
        \begin{equation}
            \begin{bmatrix}y_1(t)\\y_2(t)\\y_3(t)
            \end{bmatrix}= \boldsymbol{\gamma}(\boldsymbol{x})=\begin{bmatrix}- \frac{x_{1}^{2}}{2} + x_{1} x_{2} - \frac{x_{2}^{2}}{2}\\0\\0\end{bmatrix}
            \label{eq:nonmap2}
        \end{equation}
    For the mapping showed in \eqref{eq:nonmap2}, the observability matrix is not full rank and therefore the complete information of the dynamic system cannot be inferred from the outputs. The results of applying the classifier to this scenario are shown in Figure~\ref{fig:nonmaping_nb22}.
    \begin{figure}[H]
        \centering
        \begin{subfigure}{0.45\textwidth}
            \includegraphics[width=\textwidth]{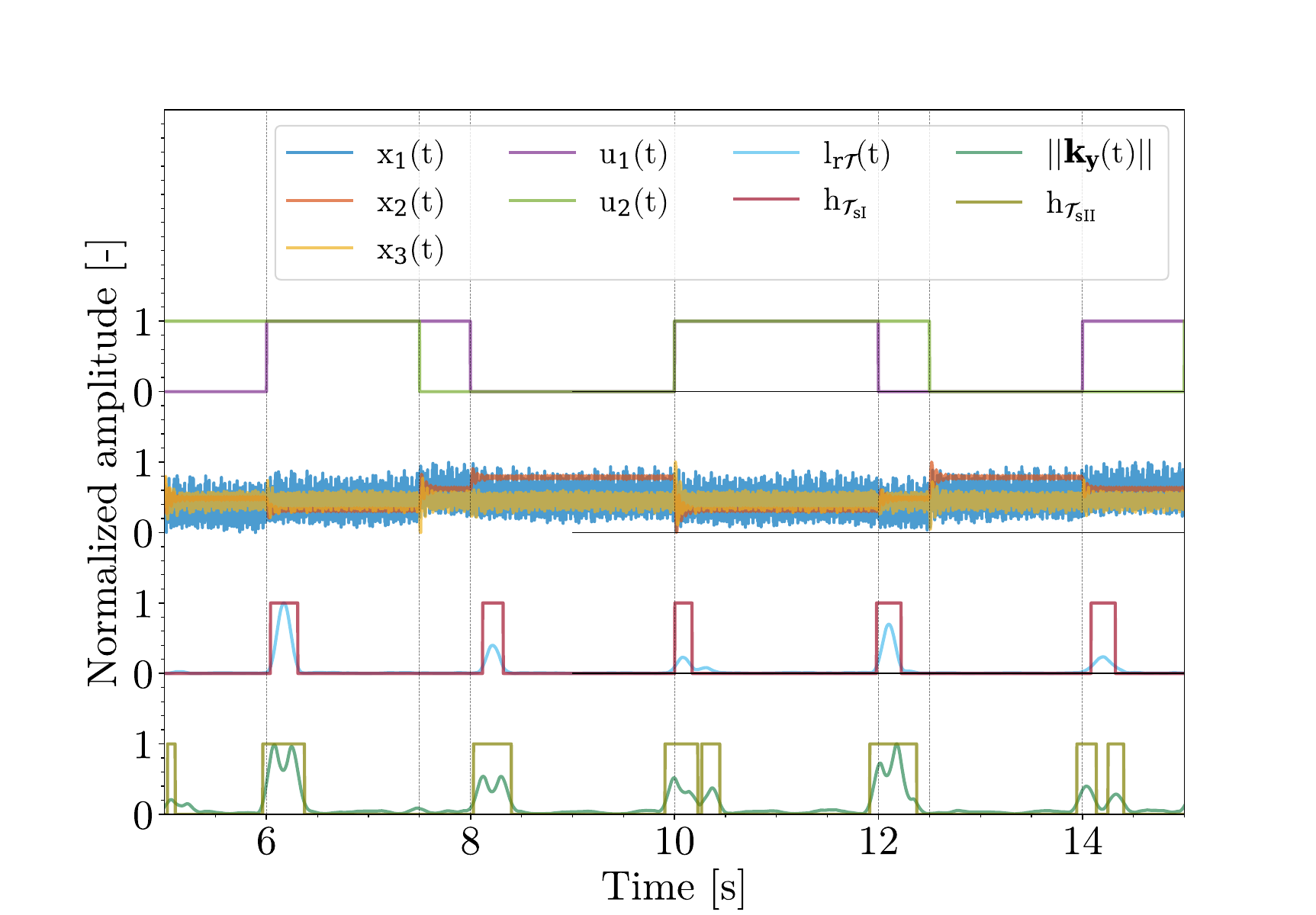}
            \caption{}
        \end{subfigure}
    \end{figure}
    \begin{figure}[H]\ContinuedFloat
        \centering
        \begin{subfigure}{0.45\textwidth}
            \includegraphics[width=\textwidth]{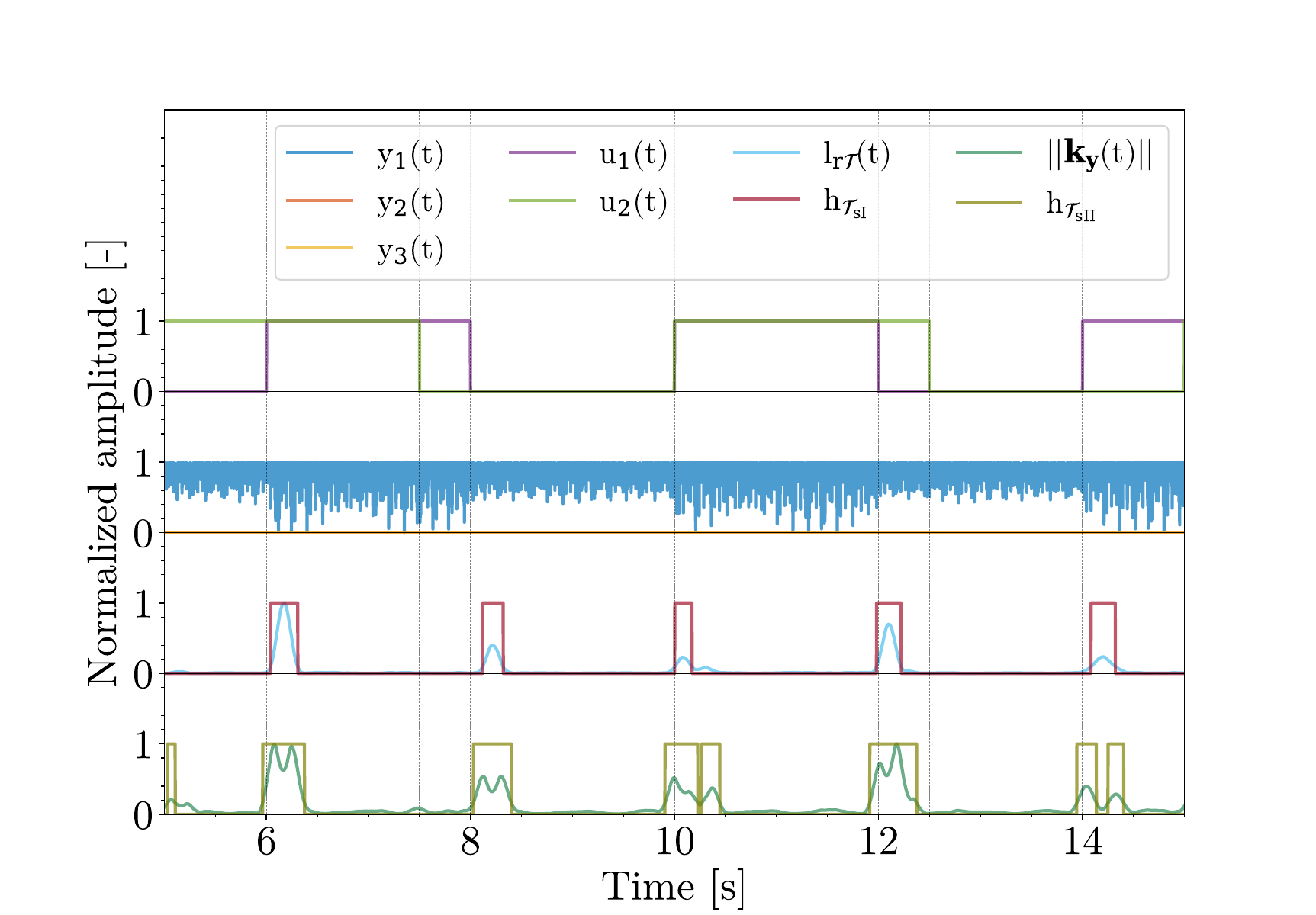}
            \caption{}
        \end{subfigure}
        \caption{Classification of the transient regimes using the proposed methodology for the non-observable mapping of \eqref{eq:nonmap}. (a) Time response of the inputs, states, geometrical properties of the SMTR curve and proposed classifier. (b) Time response of the inputs, outputs, geometrical properties of the SMTR curve and proposed classifier.}
        \label{fig:nonmaping_nb22}
        \end{figure}
    In summary, the proposed algorithm is effective for identifying transient regimes within nonlinear mappings from the states of the system to its outputs, as long as the outputs can accurately represent the dynamics of the studied system.
\color{black}

\color{black}
\setcounter{figure}{0}
\setcounter{table}{0}
\setcounter{equation}{0}
\renewcommand{\thetable}{H.\arabic{table}}
\renewcommand{\thefigure}{H.\arabic{figure}}
\renewcommand{\theequation}{H.\arabic{equation}}

\subsection*{H. Study of the effect of noise in the classification performance}

Consider the dynamic system with noise in the states and outputs mappings in state space form written as
\begin{align*}
    \dot{\boldsymbol{x}}(t)&=\boldsymbol{\alpha}(\boldsymbol{x}(t),\boldsymbol{u}(t)) + \boldsymbol{\upsilon}(t),\\
    \boldsymbol{y}(t)&=\boldsymbol{\gamma}(\boldsymbol{x}(t)) + \boldsymbol{\eta}(t),        
\end{align*}
\noindent where $\boldsymbol{\upsilon}(t) \sim \mathcal{N}(0,\Sigma_{\upsilon})$ and $\boldsymbol{\eta}(t) \sim \mathcal{N}(0,\Sigma_{\eta})$. Moreover, $\Sigma_{\upsilon} \neq O$, $\Sigma_{\eta} \neq O$  with $O$ as the null matrix.

It is worth mentioning that the covariance matrices of the noise signals were chosen to uphold a signal-to-noise ratio (SNR) sufficient for distinguishing between the system dynamics and the background noise in each analyzed signal. The signal-to-noise ratio can be expressed as
\begin{equation}
    \resizebox{.95\linewidth}{!}{$
    SNR(s(t)) = 10\log_{10}\left(\frac{p_{signal}(t)}{p_{noise}(t)}\right)=10\log_{10}\left(\frac{\int_{0}^{t_{sim}}s(t)dt}{\int_{0}^{t_{sim}}s_n(t)dt}\right),$}
    \label{eq:snr}
\end{equation}
where $s(t)$, $s_{n}(t)$ and $t_{sim}$ represent the studied signal, its noise, and the simulation period, respectively. The definitions and characteristics of SNR and its variants are elaborated in \cite{welvaert2013definition}. Generally, in signal processing, an SNR of 100 implies negligible noise in the signal, an SNR of 0 indicates comparable energy between the signal and the noise, and a negative SNR denotes predominant noise in the signal. A low SNR can lead to inaccuracies in the performance of the algorithm. For clarity, Figure~\ref{fg:noise_exp} illustrates the classification of transitions under various signal-to-noise ratios computed using \eqref{eq:snr} for the linear dynamic system presented in Section~A.
\begin{figure}[H]
    \centering
    \begin{subfigure}{0.4\textwidth}
        \includegraphics[width=\textwidth]{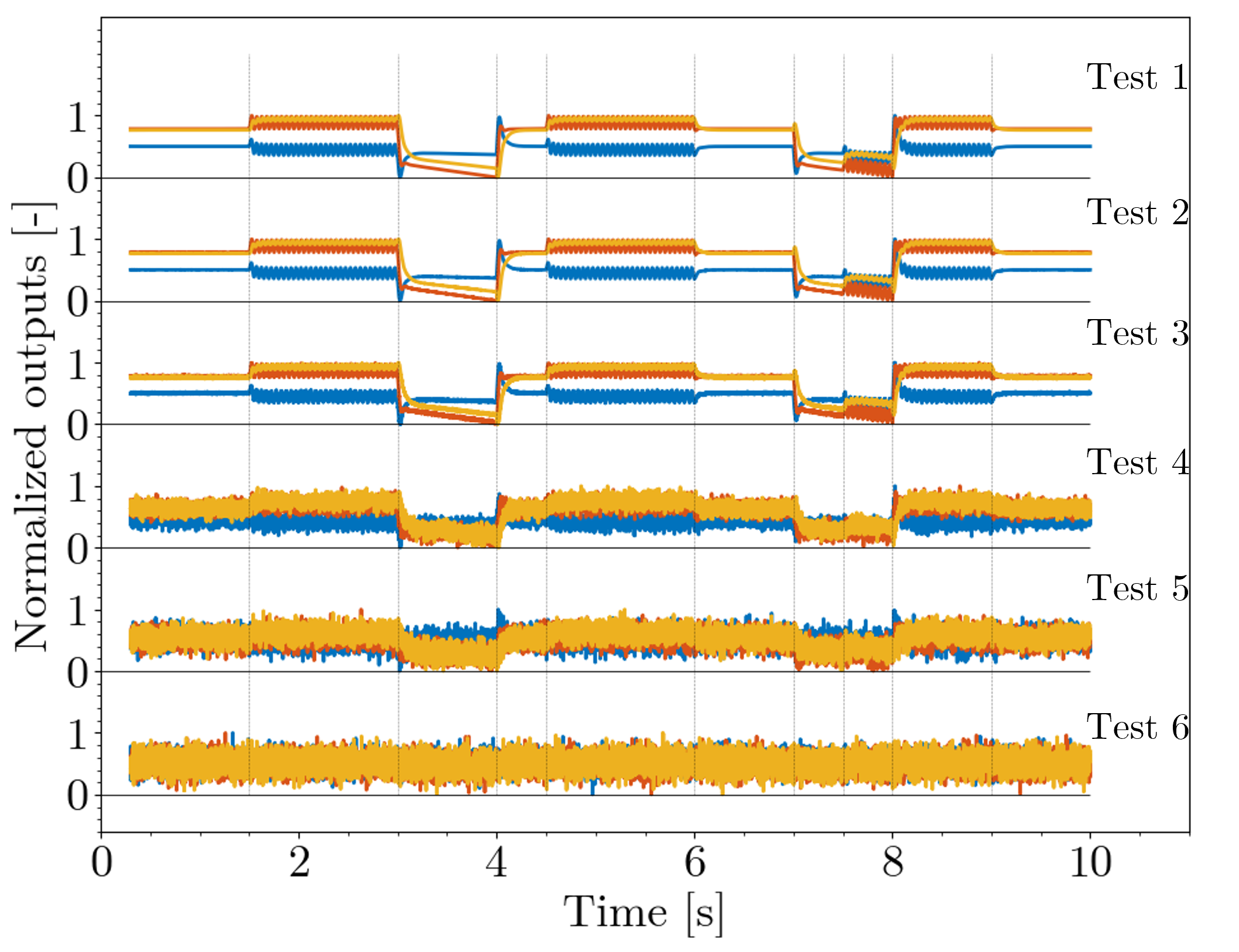}
        \caption{}
    \end{subfigure}
\end{figure}
\begin{figure}[H]\ContinuedFloat
    \centering
    \begin{subfigure}{0.4\textwidth}
        \includegraphics[width=\textwidth]{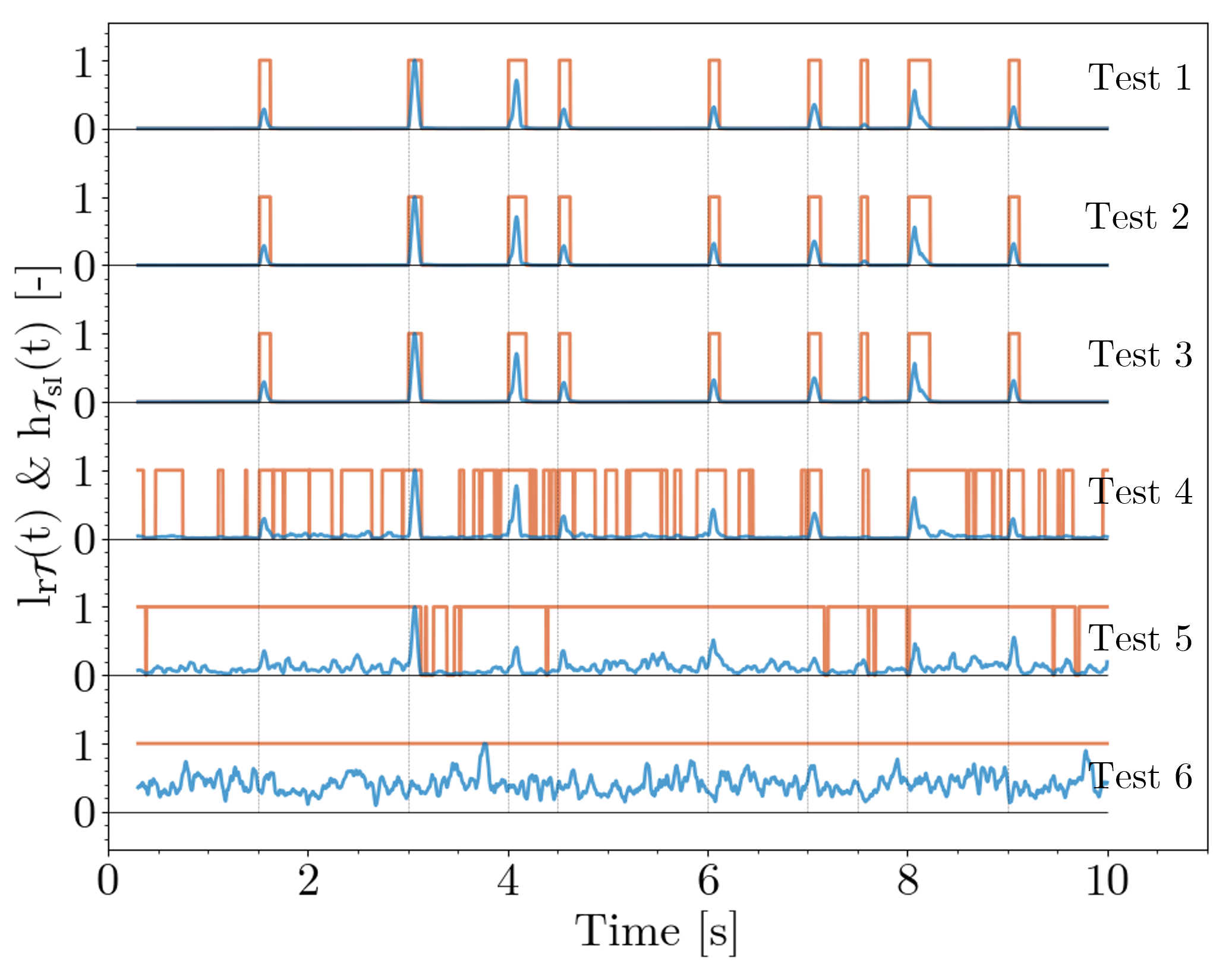}
        \caption{}
    \end{subfigure}
\end{figure}
\begin{figure}[H]\ContinuedFloat
    \centering
    \begin{subfigure}{0.4\textwidth}
        \includegraphics[width=\textwidth]{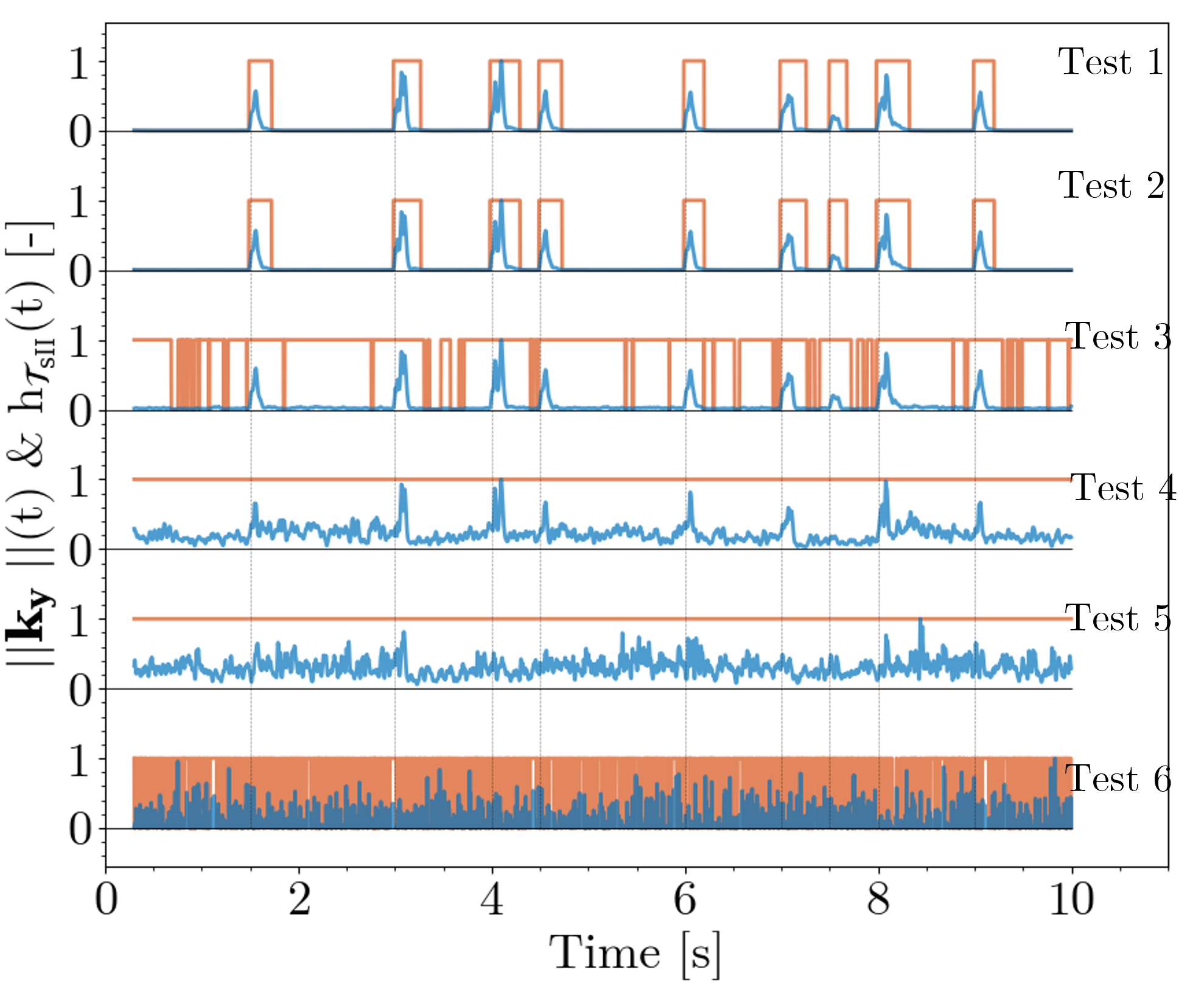}
        \caption{}
    \end{subfigure}
    \caption{Example of the performance of the proposed transient classifier with different SNRs. (a) Output signals of the simulated linear systems for the tested SNR scenarios. (b) Arc length of the SMTR curve and its classifier for each studied scenario. (c) Plane curvature vector magnitude of the SMTR curve and its classifier for each studied scenario.}
    \label{fg:noise_exp}
\end{figure}
\begin{table}[H]
    \caption{Values of SNR of $\boldsymbol{x}(t)$ and $\boldsymbol{y}(t)$ for the evaluated tests.}
    \centering
    \small
    \begin{tabular}{ccc}
        \hline
        Test & SNR($\boldsymbol{x}(t)$) [dB] & SNR($\boldsymbol{y}(t)$) [dB] \\
        \hline
        1 & (96, 95, 90) & (88, 96, 93) \\
        2 & (56, 55, 50) & (48, 56, 53) \\
        3 & (36, 35, 30) & (27, 36, 33) \\
        4 & (15, 15, 10) & (7, 14, 12) \\
        5 & (10, 9, 4) & (1, 7, 5) \\
        6 & (-3, -4, -9) & (-10, -11, -11) \\
        \hline
    \end{tabular}
    \label{tab:example}
\end{table}
Based on the conducted experiments, it is evident that the $h_{\mathcal{T}_{sI}}$ classifier demonstrates effective transition classification even when SNR values are below 50. Conversely, the $h_{\mathcal{T}_{sII}}$ classifier exhibits inadequate performance when SNR falls below 50. Ultimately, both classifiers cease to function effectively when SNR levels drop below 30, highlighting a limitation on the permissible noise levels in the signals of the dynamic systems under analysis.
\color{black}

\color{black}
\setcounter{figure}{0}
\setcounter{table}{0}
\setcounter{equation}{0}
\renewcommand{\thetable}{I.\arabic{table}}
\renewcommand{\thefigure}{I.\arabic{figure}}
\renewcommand{\theequation}{I.\arabic{equation}}

\subsection*{I. Study of the effect of variable sampling in the classification performance}

As outlined in the assumptions, the algorithm relies on a constant sampling rate in the data to ensure the reliability of geometric properties for identifying transitions. In situations where the data exhibits a variable sampling rate, it becomes necessary to preprocess the signals to align the input data of the algorithm as closely as possible with this assumption.

Resampling algorithms, such as those described in \cite{russell2006regular}, can be used to fix variable sampling rate in the sensor data. These algorithms entail the application of filters and interpolation techniques to estimate samples at desired time intervals. For the sake of clarity, a study of the effect of variable sampling time is next presented.

Consider the time vector with variable sampling time written as

\begin{equation}
    T_{data}:=\{ t_k \in \mathbb{R} | t_k=t_{k-1}+\delta \},
    \label{eq:time_sta}
\end{equation}
\noindent where $\delta \sim U(-\hat{\delta}a_1,\hat{\delta}a_2)$ is a random variable that follows a uniform distribution in the interval $-\hat{\delta}a_1$ and $\hat{\delta}a_2$ with $a_2>a_1$. Where $\delta$ is the ideal sampling time of the system.

To evaluate the impact of variable sampling time in the performance of the proposed method, $T_{data}$ was generated with \eqref{eq:time_sta} using the parameters described in Table~\ref{table:param_sampling}, the mean and standard deviation of the sampling time $t_s$ is also shown in this table. Then, the linear system described in Section~A was simulated and $\boldsymbol{y}(t)$ was sampled in the times specified by $T_{data}$. Finally, the proposed transient detection algorithm was applied. It should be noted that the selection of sampling with uniform distribution allows evaluating scenarios with the greatest sampling time uncertainty.

\begin{table}[H]
    \caption{Values of $a_1$ and $a_2$ for the evaluated tests.}
    \centering
    \small
    \begin{tabular}{ccccc}
        \hline
        Test & $a_1$ & $a_2$&Mean$t_s$ [ms] & Std. deviation $t_s$ [ms]\\
        \hline
        1 & 1.0 &  0.0 & 0.50&0.29\\
        2 & 5.0 &  1.5 & 1.81&1.45\\
        3 & 10.0 & 7.0 & 1.54&1.90\\
        4 & 15.0 & 8.0 & 3.37&3.44\\
        5 & 20.0 & 10.0& 4.72&4.74\\
        \hline
    \end{tabular}
    
    \label{table:param_sampling}
\end{table}

The transient detection algorithm was evaluated for each tested scenario using the same input parameters for the proposed transient regime detection algorithm. The Distribution of $\delta$ and an example of a signal sampled with variable sampling time are shown in Figure~\ref{fig:sampling}.
\begin{figure}[H]
        \centering
        \begin{subfigure}{0.4\textwidth}
            \includegraphics[width=\textwidth]{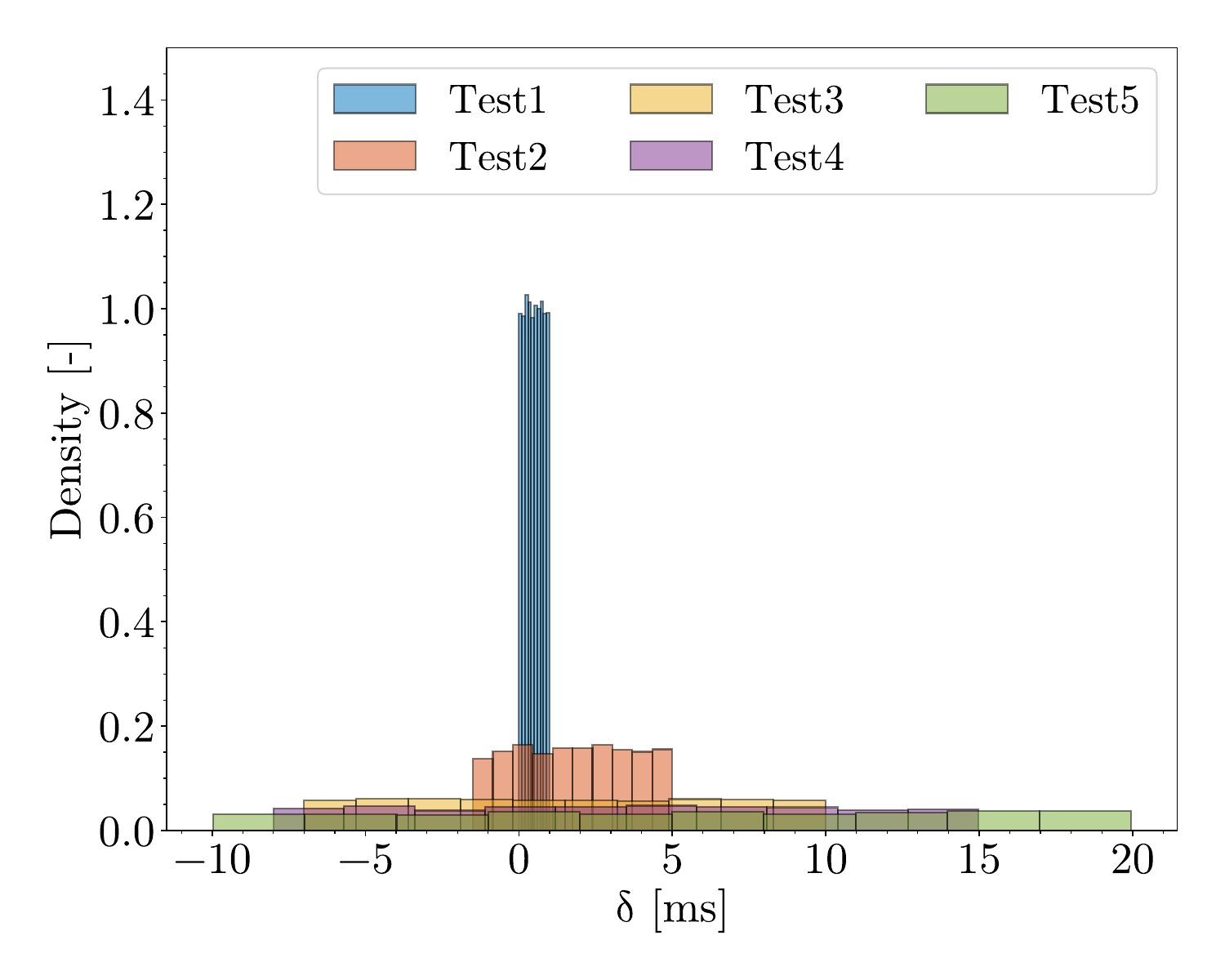}
            \caption{}
        \end{subfigure}
    \end{figure}
    \begin{figure}[H]\ContinuedFloat
        \centering
        \begin{subfigure}{0.4\textwidth}
            \includegraphics[width=\textwidth]{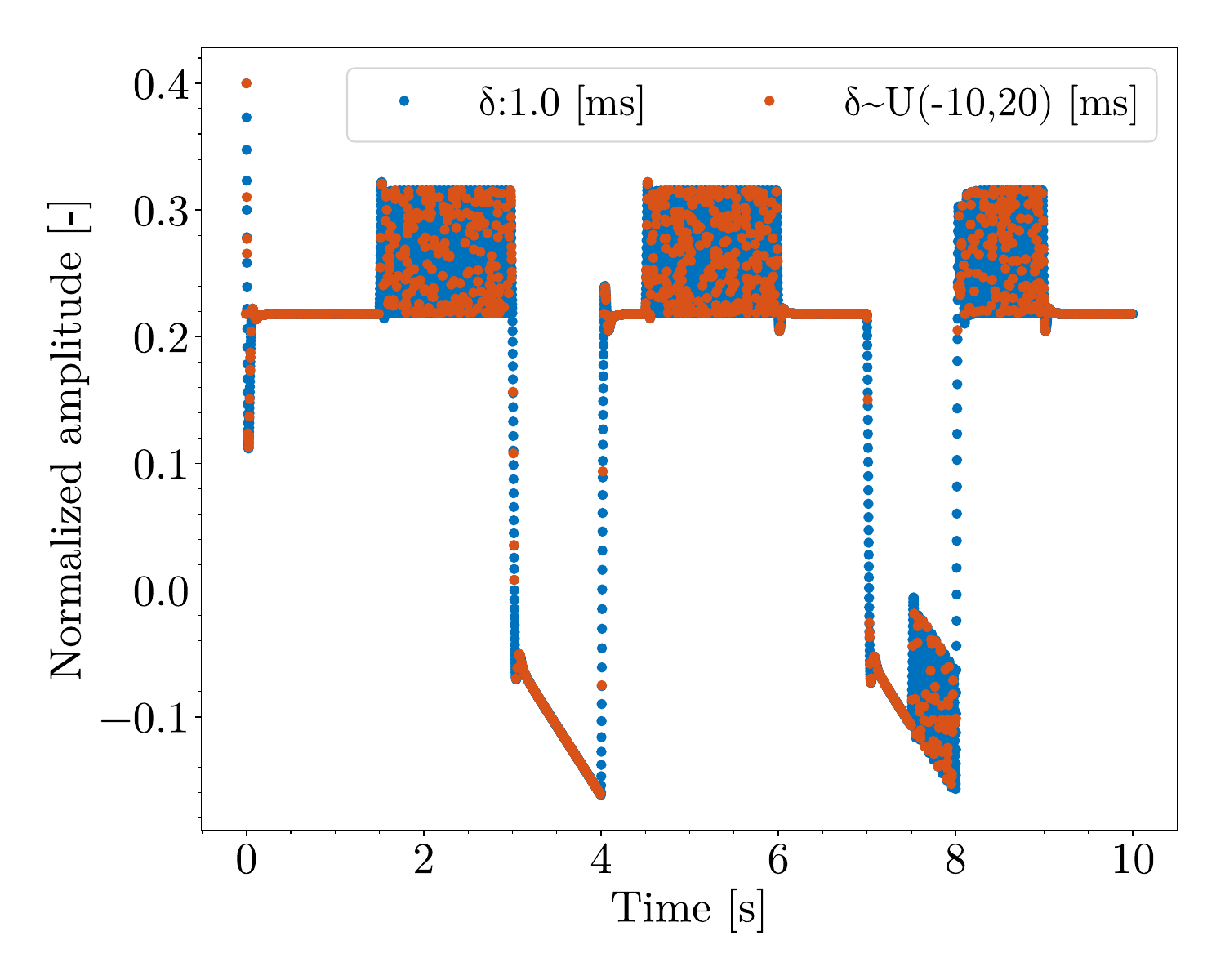}
            \caption{}
        \end{subfigure}
        \caption{Distribution of the sampling time $\delta$ for the tested sampling scenarios. (a) Histogram of the sampling time for each tested scenario. (b) Comparison of a sampled input signal using constant and non-uniform sampling time.}
        \label{fig:sampling}
\end{figure}
An illustrative example of applying the proposed classifier algorithm to non-uniformly sampled data, along with the time domain acquired data for all tested experiments are shown in Figure~\ref{fig:sampling2}.
    \begin{figure}[H]
        \centering
        \begin{subfigure}{0.4\textwidth}
            \includegraphics[width=\textwidth]{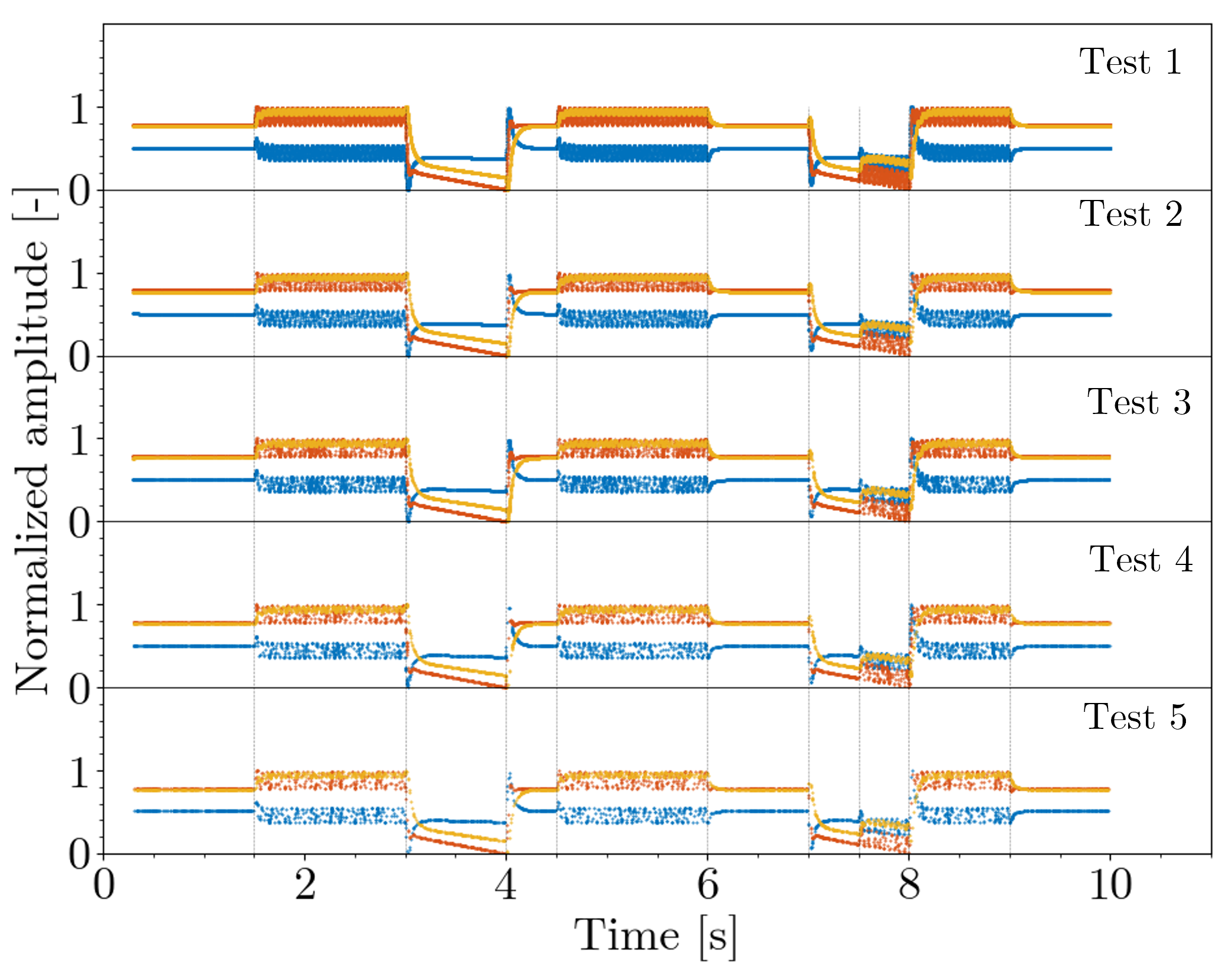}
            \caption{}
        \end{subfigure}
    \end{figure}
    \begin{figure}[H]\ContinuedFloat
        \centering
        \begin{subfigure}{0.4\textwidth}
            \includegraphics[width=\textwidth]{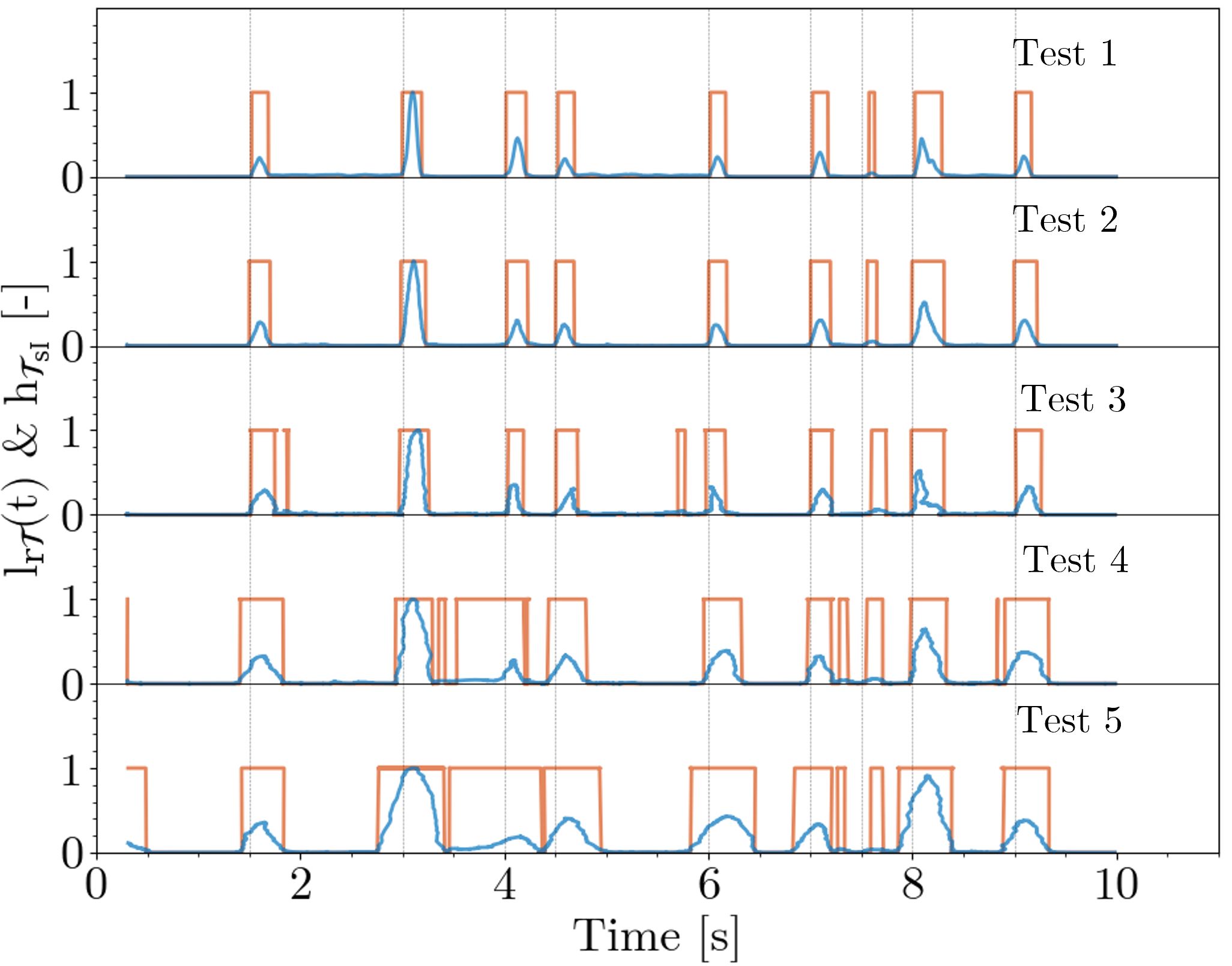}
            \caption{}
        \end{subfigure}
    \end{figure}
    \begin{figure}[H]\ContinuedFloat
        \centering
        \begin{subfigure}{0.4\textwidth}
            \includegraphics[width=\textwidth]{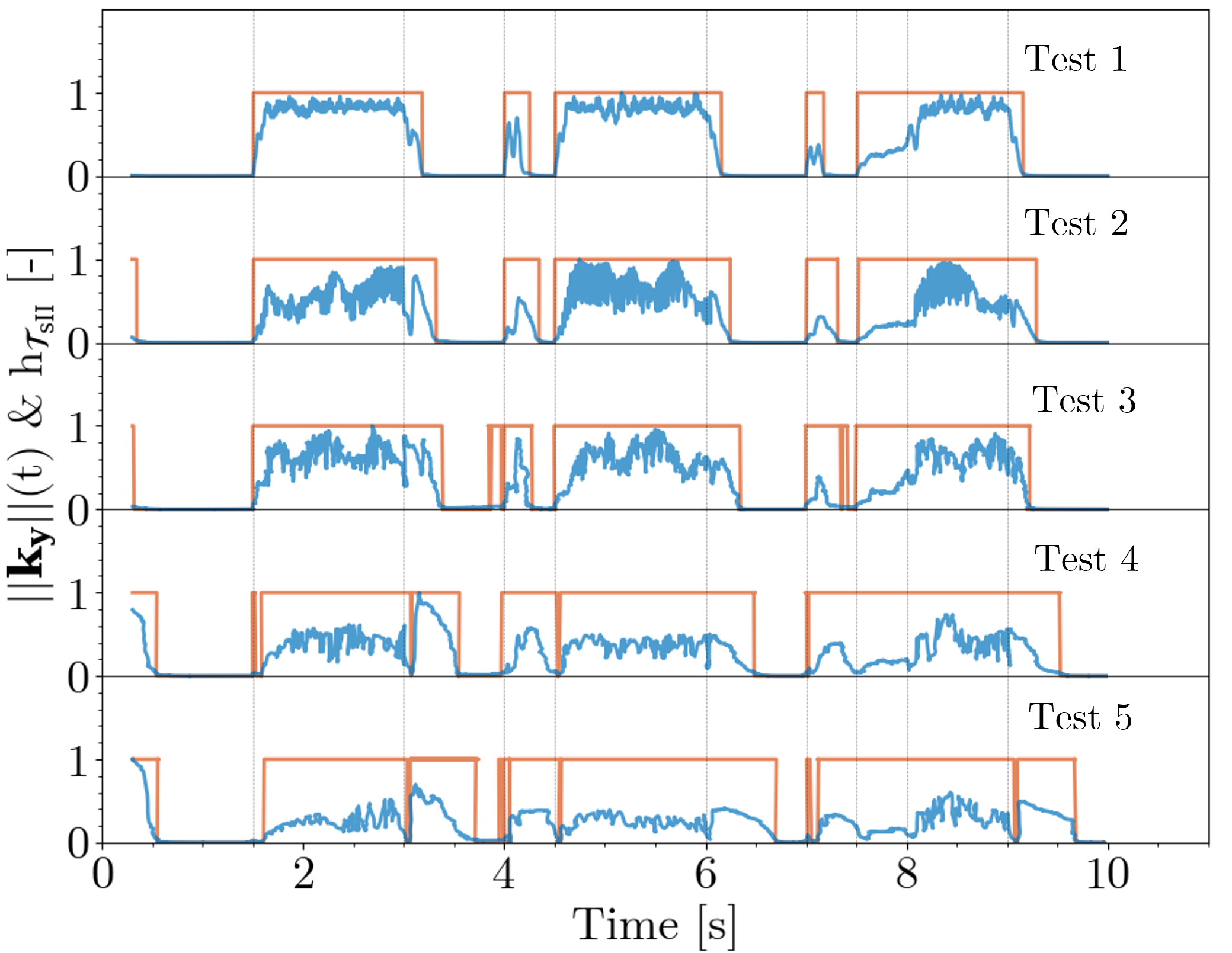}
            \caption{}
        \end{subfigure}
        \caption{Example of the performance of the proposed transient classifier with non-uniform sampling time. (a) Output signals of the simulated linear systems for the tested sampling scenarios. (b) Arc length  of the SMTR curve and its classifier for each studied scenario. (c) Plane curvature vector magnitude and its classifier for each studied scenario.}
        \label{fig:sampling2}
    \end{figure}
    Based on the conducted experiments, it is evident that the $h_{\mathcal{T}_{sI}}$ classifier can effectively classify transitions even in scenarios with significant variances in sampling time intervals. In contrast, the $h_{\mathcal{T}_{sII}}$ classifier is notably sensitive to fluctuations in sampling rate, functioning adequately only when sampling times remain constant. Ultimately, maintaining a constant sampling time emerges as a limitation for accurately classifying transient regimes using both proposed classifiers.
\color{black}
\color{black}
\setcounter{figure}{0}
\setcounter{table}{0}
\setcounter{equation}{0}
\renewcommand{\thetable}{J.\arabic{table}}
\renewcommand{\thefigure}{J.\arabic{figure}}
\renewcommand{\theequation}{J.\arabic{equation}}

\subsection*{J. Example of Control design application based on the proposed classifier}
To illustrate the applications where the proposed method can be applied. This algorithm, specifically, finds relevance in two fields: system identification and control. In practical scenarios, it becomes necessary to identify and control the dynamics of processes when subjected to specific input excitations. To illustrate this concept, consider the presented linear system of Section A with classified transitions exemplified in Figure~\ref{fig:control_1}. Furthermore, assume that the output $y_2(t)$ is required to follow the input reference given by $u_2(t)$ within the time segment enclosed by the red dashed line in Figure~\ref{fig:control_1}(a). 
\begin{figure}[H]
    \centering
    \begin{subfigure}{0.35\textwidth}
        \includegraphics[width=\textwidth]{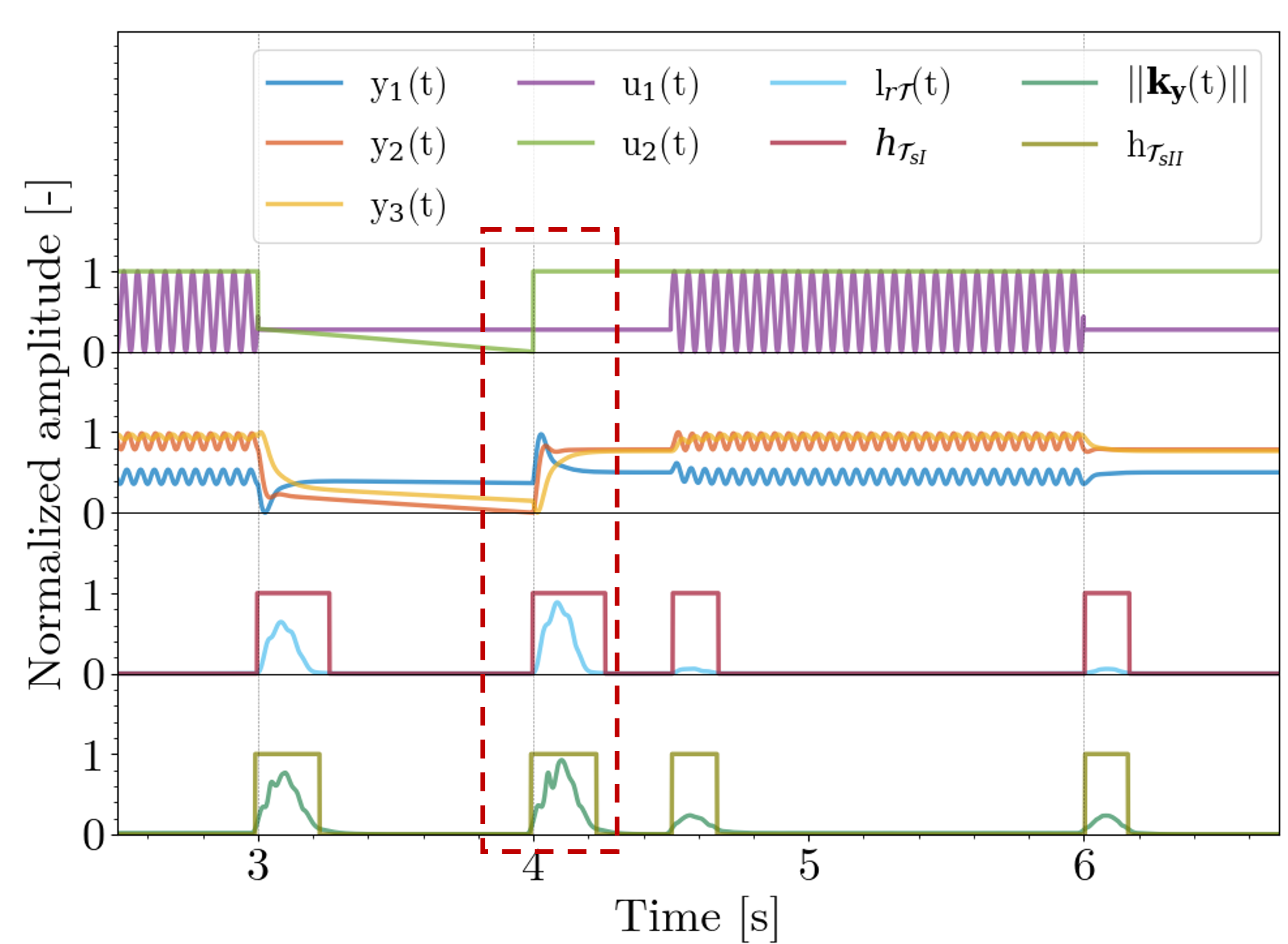}
        \caption{}
    \end{subfigure}
\end{figure}
\begin{figure}[H]\ContinuedFloat
    \centering
    \begin{subfigure}{0.36\textwidth}
        \includegraphics[width=\textwidth]{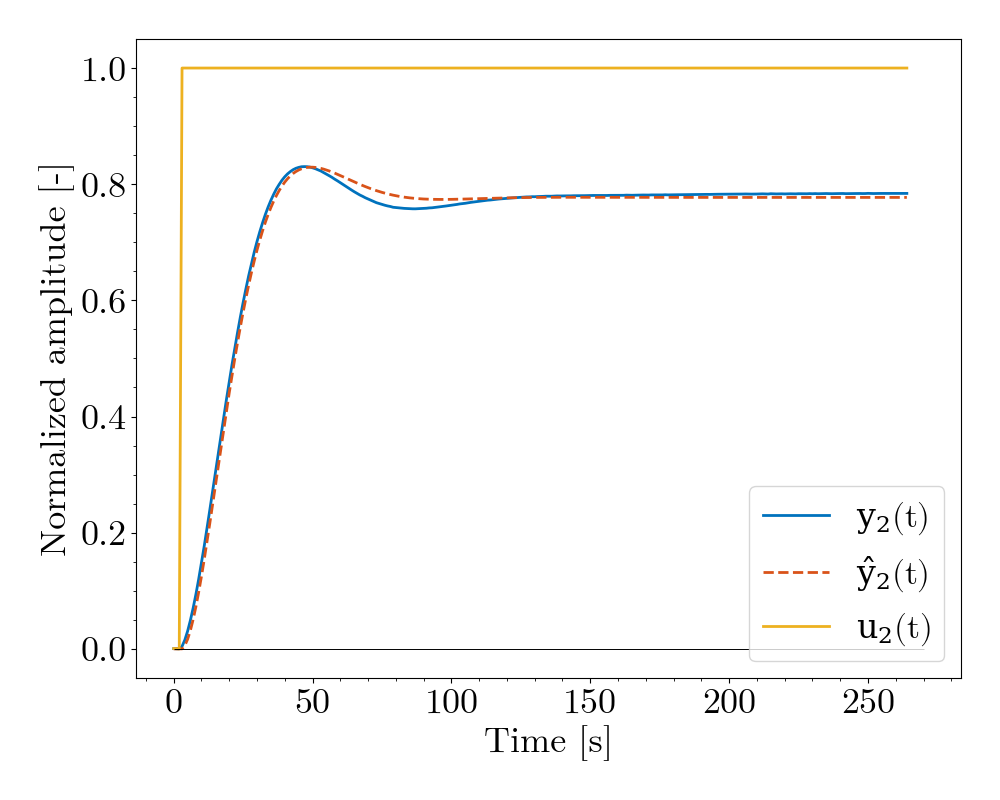}
        \caption{}
    \end{subfigure}
    \caption{Example of a control application of the proposed transient classifier method. (a) time range where the output $y_2(t)$ of the system needs to follow $u_2(t)$. (b) Time response of the input and output of the required dynamics to be controlled and estimated output $\hat{y}_2(t)$ using ARX system identification technique.}
    \label{fig:control_1}
    \end{figure}
 It is possible to use the data of the classified segment to identify the dynamics of this transition by fitting a mathematical model such as the ARX model written as
\begin{equation*}
    \resizebox{.99\linewidth}{!}{$
    G(z)=\frac{y_2(z)}{u_2(z)}\approx\frac{0.005z}{z^2-1.885z+0.892},\qquad\hat{y}_2
(z)=G(z)u_2(z).$}
\end{equation*}

Then, a controller can be designed based on predefined criteria that encompass factors like error tolerance and response time. For instance, consider the tuned PID controller in the following configuration:
\begin{equation*}
    C(z)=\frac{171.5z^2-313.3z+143.3}{z^2-z}.
\end{equation*}

Finally, the fitted ARX model and the close loop response of $G(z)C(z)$ is shown in Figure~\ref{fig:control_2}.
\begin{figure}[H]
    \centering
    \begin{subfigure}{0.42\textwidth}
        \includegraphics[width=\textwidth]{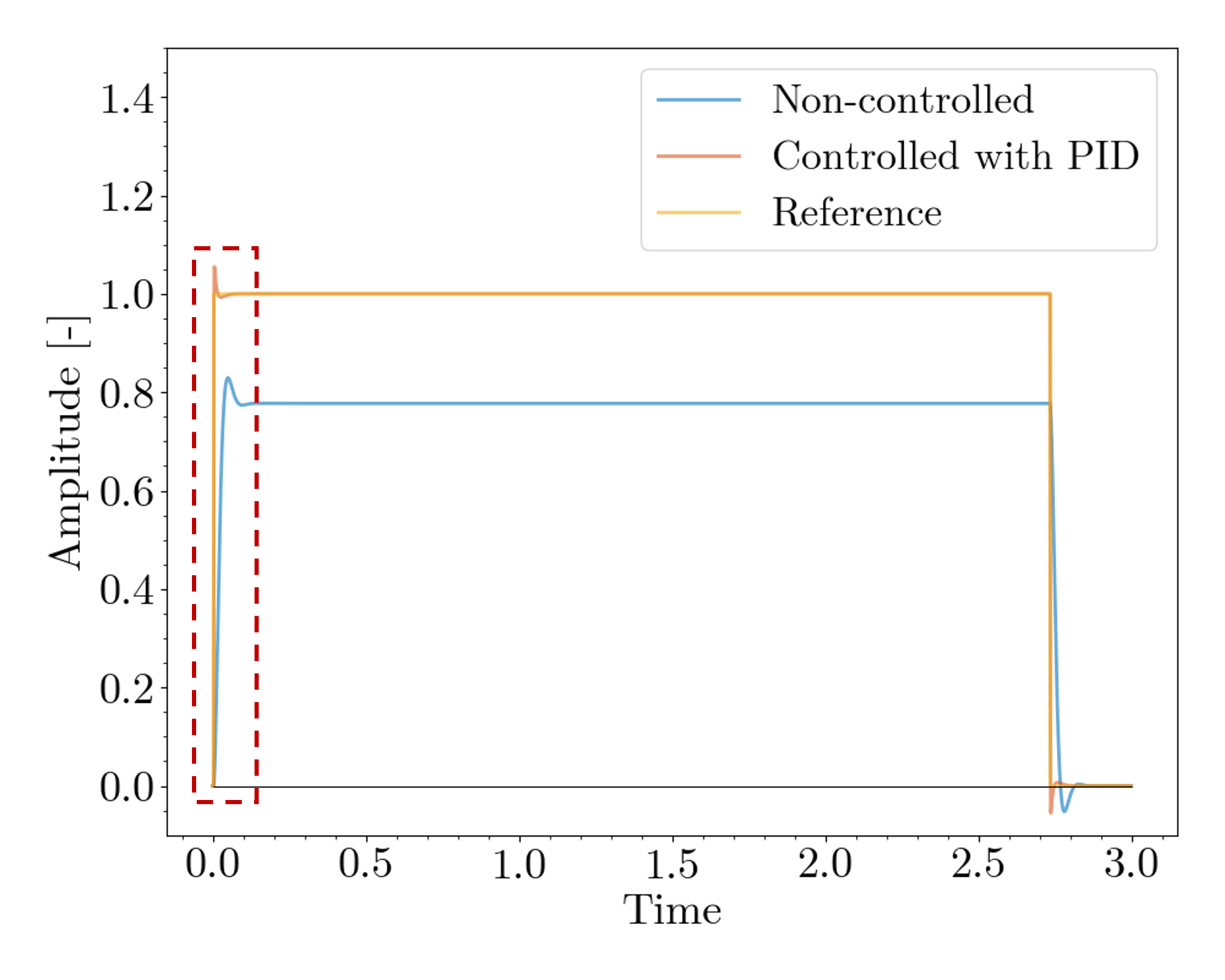}
        \caption{}
    \end{subfigure}
\end{figure}
\begin{figure}[H]\ContinuedFloat
    \centering
    \begin{subfigure}{0.43\textwidth}
        \includegraphics[width=\textwidth]{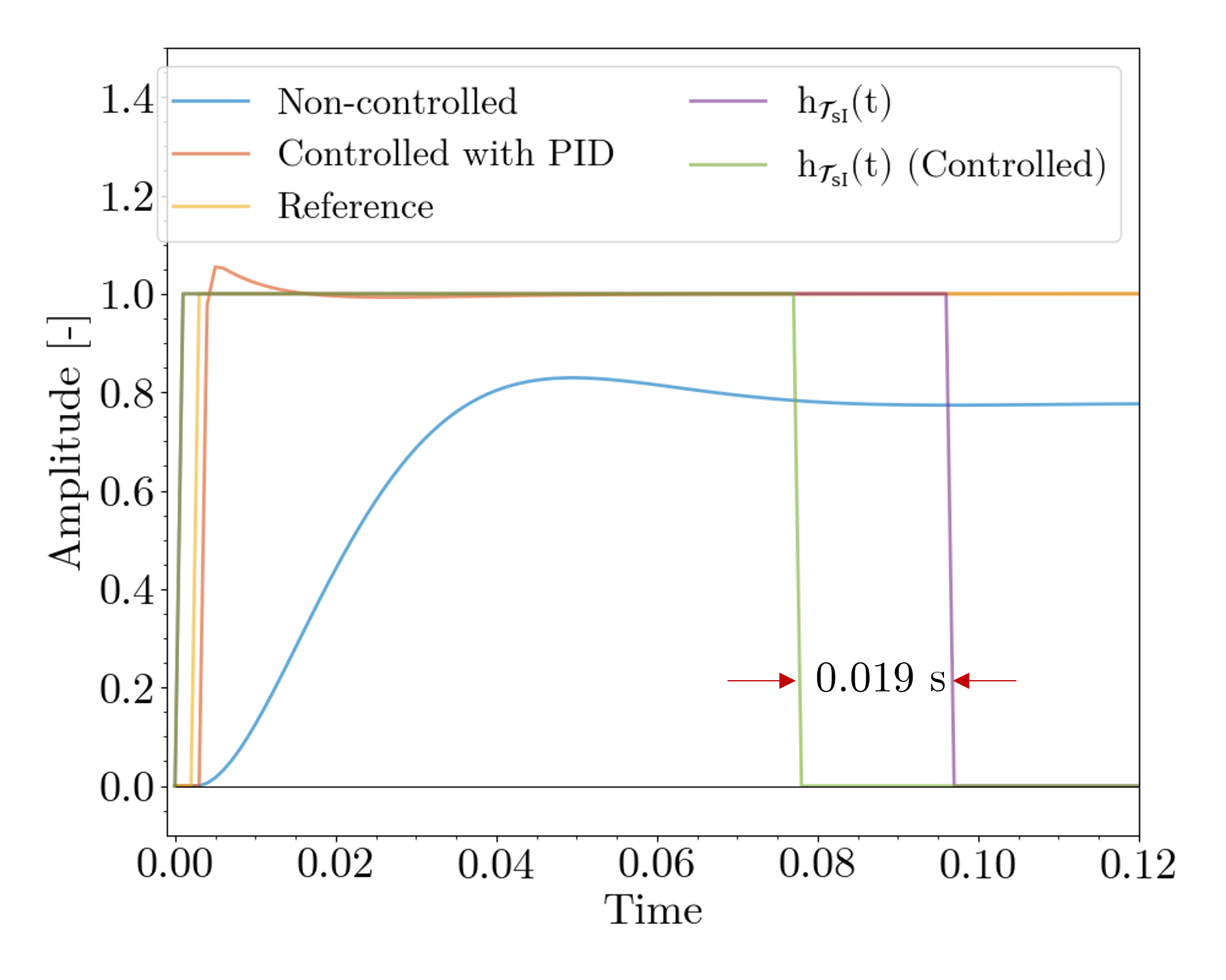}
        \caption{}
    \end{subfigure}
    \caption{Time response of the output $y_2(t)$ before and after the controller is implemented in the system. (a) Controlled, non-controlled systems response and input reference. (b) Proposed transient regimen classifier $h_{\mathcal{T}_{sI}}(t)$ for the controlled and non-controlled systems response. The transient regime of the system is reduced by 0.019 [\si[]{\second}] with the implemented PID controller.}
    \label{fig:control_2}
    \end{figure}
As shown in Figure~\ref{fig:control_2}, the effectiveness of the controller applied to the system can be validated by reapplying the classifier and assessing whether the response time aligns with the specified designed criteria. Furthermore, with the proposed classifier, the validation of the time response requirements can be performed online so that the controller can be syntonized until it meets the desired control requirements. 
\color{black}

\color{black}
\setcounter{figure}{0}
\setcounter{table}{0}
\setcounter{equation}{0}
\renewcommand{\thetable}{K.\arabic{table}}
\renewcommand{\thefigure}{K.\arabic{figure}}
\renewcommand{\theequation}{K.\arabic{equation}}

\subsection*{K. Construction of reference classifier for Hypothesis test}
To compute the probability of type I and type II error for the studied classifiers, a reference classifier $h_{ref}$ was generated. For the linear system, the reference classifier was built using the step input settling time criterion written as
\begin{equation}
    t_{stc}=\frac{4}{\sigma_f},
    \label{eq:tstc}
\end{equation}
\noindent
where $\sigma_f$ is the absolute value of the smallest real part of the complex poles of the studied dynamic system. Then, a value of one was assigned to $h_{ref}$  within the ranges $(t_u,t_u+t_{stc})$, with $t_u$ as the times when the systems input switches. Finally, the value of zero was assigned to $h_{ref}$ for the rest of the simulated time.

Regarding the nonlinear and discontinuous system, the rise time criterion was used to calculate the classifier. The rise time can be computed as
\begin{equation}
    t_{rt}=\frac{1}{2f_{nm}},
    \label{eq:trt}
\end{equation}
\noindent

where $f_{nm}$ is the smallest natural frequency of the system. Based on the previous guidelines, the reference classifier $h_{ref}$ was computed for each studied dynamic system. Therefore, $h_{ref}$ was assigned to 1 every time that the input signal switched during 0.2, 0.25 and 2.05 seconds for the linear, non-linear and discontinuous dynamic systems, respectively. Moreover, the classifier was assigned to 0 the remaining time periods.

Once the reference classifier was obtained, the false positive rate $b_{fpr}$ was computed as the probability of committing type I and II errors. For the studied classifiers, a value of one means that the system is in transient regime, while a value of zero means a stationary regime. Thus, the false positive rate can be expressed as
\begin{equation}
    b_{fpr}=\frac{b_{fp}}{b_{fp}+b_{tn}},
    \label{eq:bfpr}
\end{equation}
\noindent
where $b_{fp}$ is the number of events where the reference classifier is zero and the studied classifier is one, while $b_{tn}$ is the number of events where the reference classifier and the studied classifier are zero.

For the type II error, \eqref{eq:bfpr} was also used. In this case, $b_{fp}$ is the number of events where the reference classifier is one and the studied classifier is zero, while $b_{tn}$ is the number of events where the reference classifier and the studied one are one.
\color{black}
\setcounter{figure}{0}
\setcounter{table}{0}
\setcounter{equation}{0}
\renewcommand{\thetable}{L.\arabic{table}}
\renewcommand{\thefigure}{L.\arabic{figure}}
\renewcommand{\theequation}{L.\arabic{equation}}

\subsection*{L. Transient identification several transient prominence conditions}

 To assess the performance of the method under varying levels of transient prominence, the linear system described in Section IV-A of the manuscript was employed. For this simulation, the inputs were defined as follows:

\begin{align*}
u_1(t) &= 
     \begin{cases}
        \text{$\delta$}&\text{$0<t<\delta$, } \\
        \text{$0.0$}&\text{$\delta<t$,}\\
     \end{cases}
     \\
     u_2(t) &= 0~\forall~t.\\
\end{align*}

The input amplitude was varied using the values $\delta = [0.01, 0.03, 0.08, 0.3, 1.0, 2.0]$, in order to produce inputs with different amplitudes and durations that generated more or less prominent transients in the measured signals. The system inputs and outputs, the spatial properties of the SMTR curve, and the evaluated transient classifiers are presented in Figure~\ref{fg:prominence}.

In addition, the probabilities of Type I and Type II errors for these classifiers, as well as for those used for comparison in the literature, are reported in Table~\ref{tab:prom_hp}. The simulation was conducted using the same parameters specified for the linear system simulation described in Section IV of the submitted document.

    \begin{figure}[htbp!]
            \centering
        \includegraphics[width=0.4\textwidth]{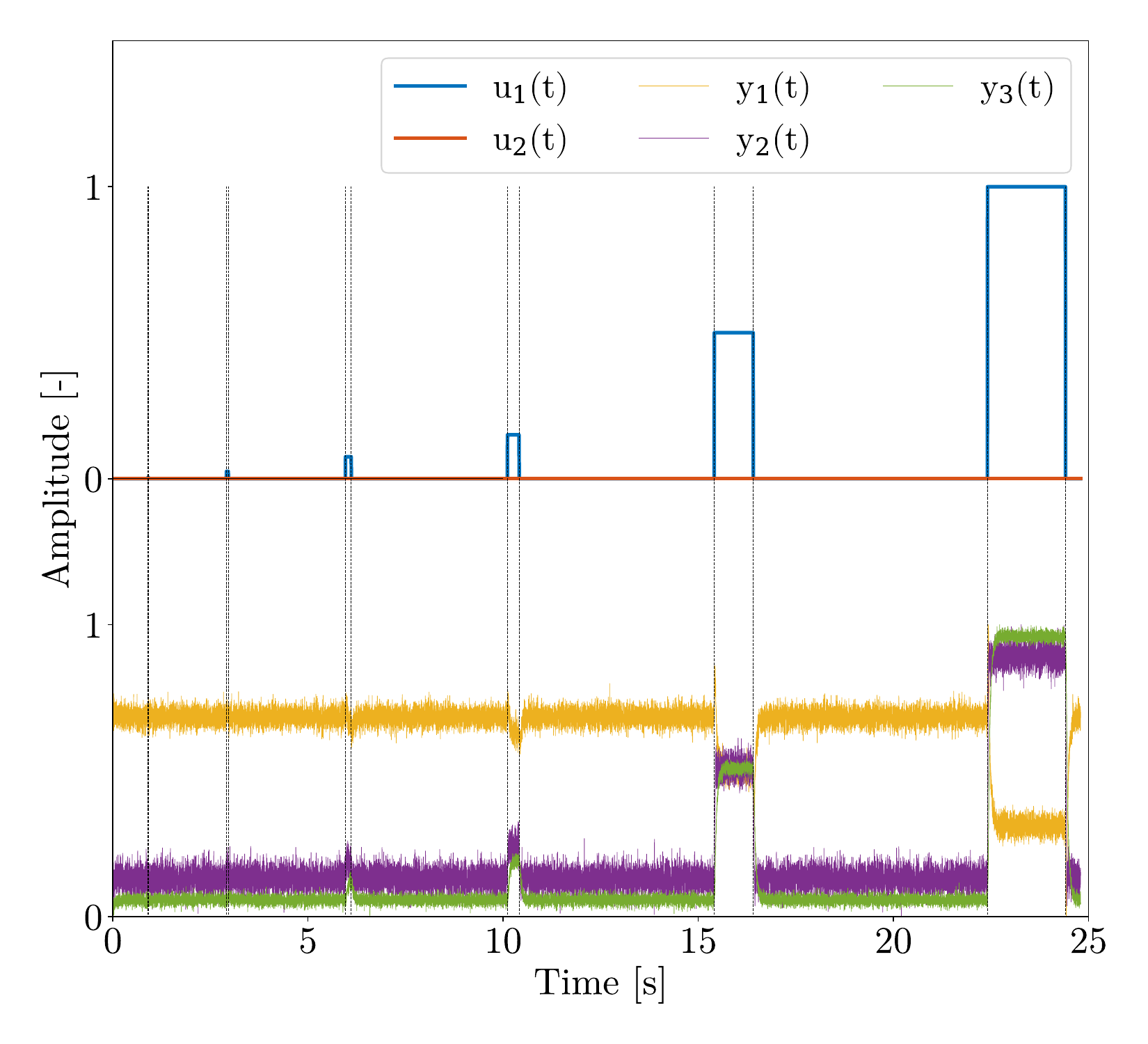}
        \caption{Results of the simulated linear dynamic system under several transient prominent conditions. Input and output signals of the simulated linear dynamical system.}
        \label{fg:p1}
    \end{figure}
    \begin{figure}[htbp!]
            \centering
        \includegraphics[width=0.4\textwidth]{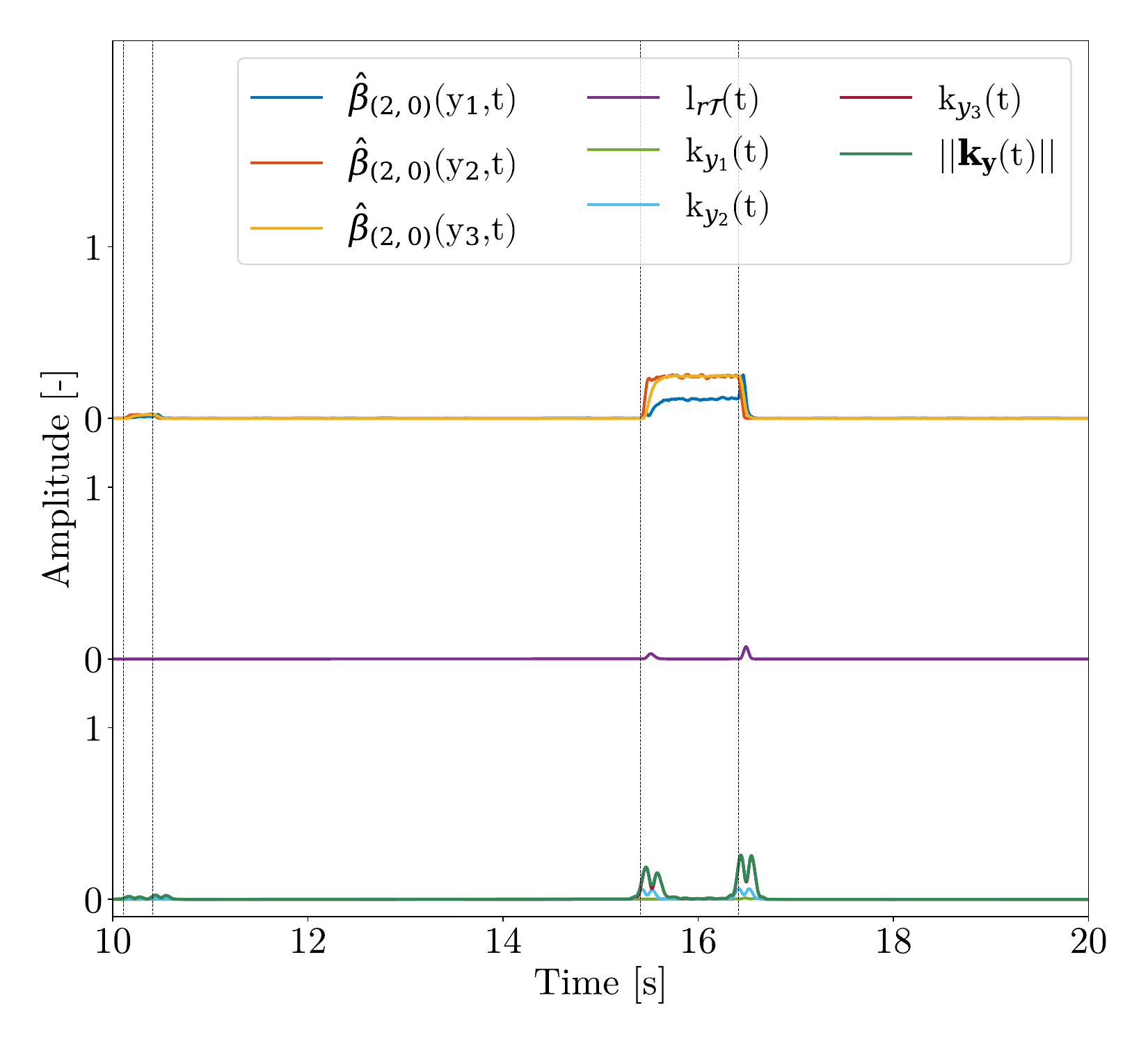}
        \caption{Results of the simulated linear dynamic system under several transient prominent conditions. SMTR curve and its geometrical properties for the considered linear-dynamical system.}
        \label{fg:p2}
    \end{figure}
    \begin{figure}[htbp!]
            \centering
        \includegraphics[width=0.4\textwidth]{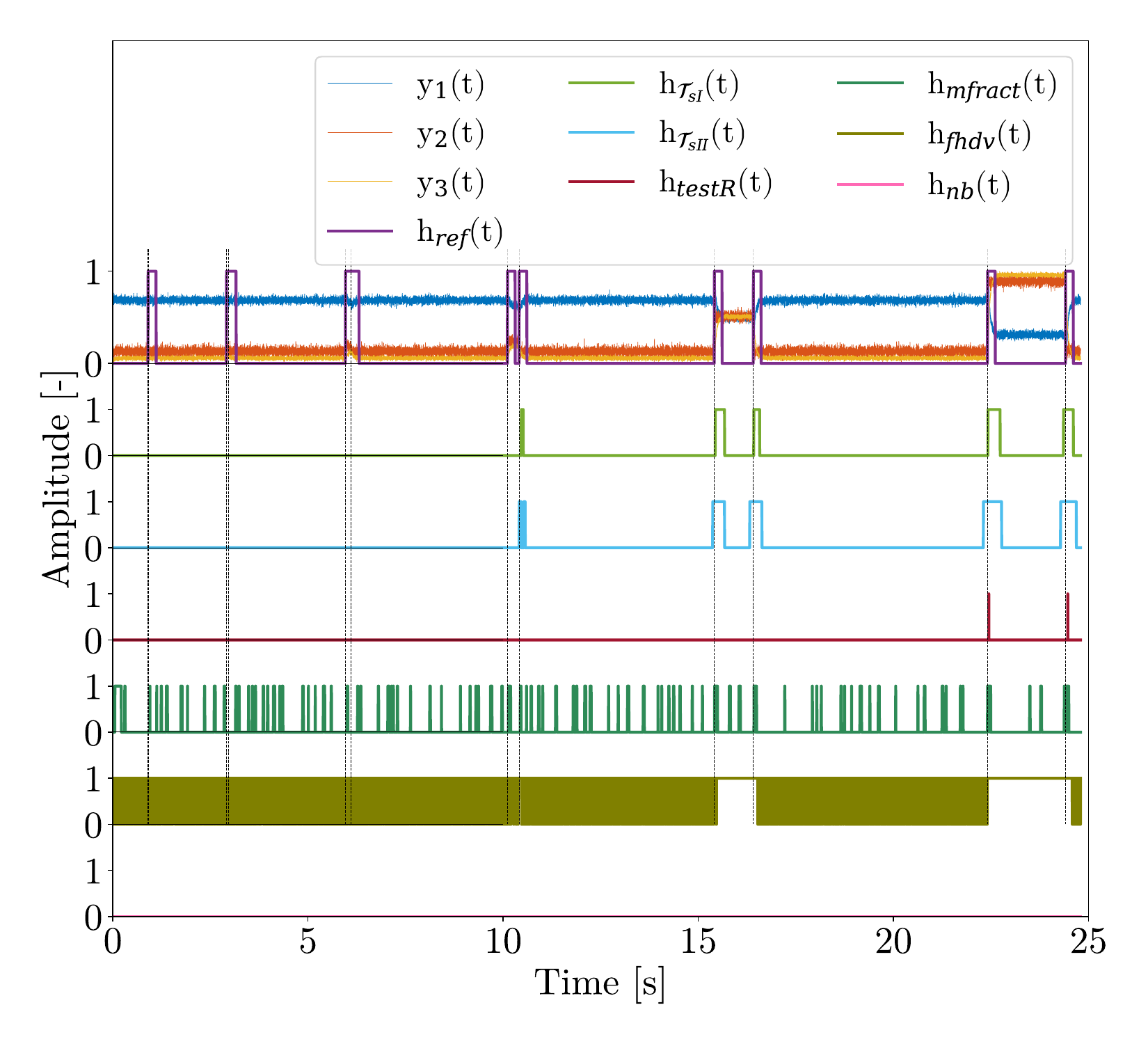}
        \caption{Results of the simulated linear dynamic system under several transient prominent conditions. Comparison between the proposed classifiers and the evaluated classifiers from the literature.}
            \label{fg:prominence}
    \end{figure}

\begin{table}[htbp!]
    \color{black}
    \centering
    \caption{Probability of Type I and type II errors for the tested transient regime classifiers based on the reference classifier $h_{ref}(t)$.}
    \begin{tabular}{ l c c c c c c}
        \hline
        \textbf{Probability}&$h_{\mathcal{T}_{sI}}$&$h_{\mathcal{T}_{sII}}$&$h_{testR}$&$h_{mfrac}$&$h_{fhdv}$&$h_{nb}$\\
        \hline
        \multirow{3}{0.05\textwidth}{Type I error [\%]}&&&&&&\\ 
            &\textbf{1.1}& 3.0 &   0.0 &    5.8 &    88.5 &    0.0\\
        &&&&&&\\
        \hline
        \multirow{3}{0.05\textwidth}{Type II error [\%]}&&&&&&\\
        &\textbf{55.9} &    62.0 &   99.7 &   81.9 &   7.4 &   100.0\\
        &&&&&&\\
        \hline
    \end{tabular} 
    \label{tab:prom_hp}
\end{table}

From the simulation results obtained by varying the prominence of the transitions, it can be observed that the proposed detectors present a limitation in identifying transitions when the rate of change and amplitude of the system inputs satisfy $\delta < 0.08$ for the linear system under study. It is important to note that the analyzed linear system has a time constant of $\tau = 0.05$ s. Therefore, applying input changes faster than this time constant prevents the system from exhibiting a measurable response, making the transition undetectable by the classifiers.

Despite this limitation, the proposed classifiers are capable of detecting transitions for input conditions with $\delta > 0.08$, as indicated in Table~\ref{tab:prom_hp}.

\setcounter{figure}{0}
\setcounter{table}{0}
\setcounter{equation}{0}
\renewcommand{\thetable}{M.\arabic{table}}
\renewcommand{\thefigure}{M.\arabic{figure}}
\renewcommand{\theequation}{M.\arabic{equation}}

\subsection*{M. Validation in real dynamic systems}

To illustrate the operation of the algorithm, a validation test was conducted on an electric motor. During this test, the supply voltage and the current consumption were monitored throughout a speed test involving programmed speed variations.

\begin{figure}[htbp!]
    \centering
        \includegraphics[width=0.3\textwidth]{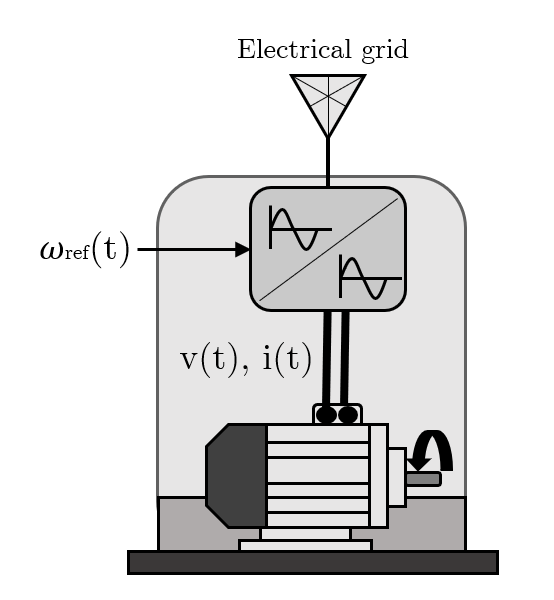}
        \caption{Controlled speed AC motor system.}
    \label{fg:motor_system}
\end{figure}

Figure~\ref{fg:motor_system} presents the controlled motor system employed for the experimental validation. In this system, the speed programmed in the variable speed drive was considered as the system excitation, represented as $\boldsymbol{u}(t) = [\omega_{ref}(t)]$, while the measured variables were the motor supply voltage and current, expressed as $\boldsymbol{y}(t) = [v(t), i(t)]$.

For the experiment, a sweep of nine discrete speeds was executed within the range of 0 to 50 Hz. The test began from rest, with the speed increased every 20 seconds and held constant between increments. It should be noted that the motor was operated without any mechanical load during this validation.

The purpose of the proposed classifiers was to automatically identify changes in speed using only with the online measurements of the system's energy consumption.

In Figure~\ref{fg:motor_validation}, the inputs and outputs of the tested AC motor system are shown, along with the components and geometric properties of the system's SMTR curve. Finally, the behavior of the two proposed classifiers is presented in this figure. It can be observed that the classification provided by the proposed indicators is consistent with the behavior of the AC motor system under programmed speed conditions.

    \begin{figure}[htbp!]
        \includegraphics[width=0.4\textwidth]{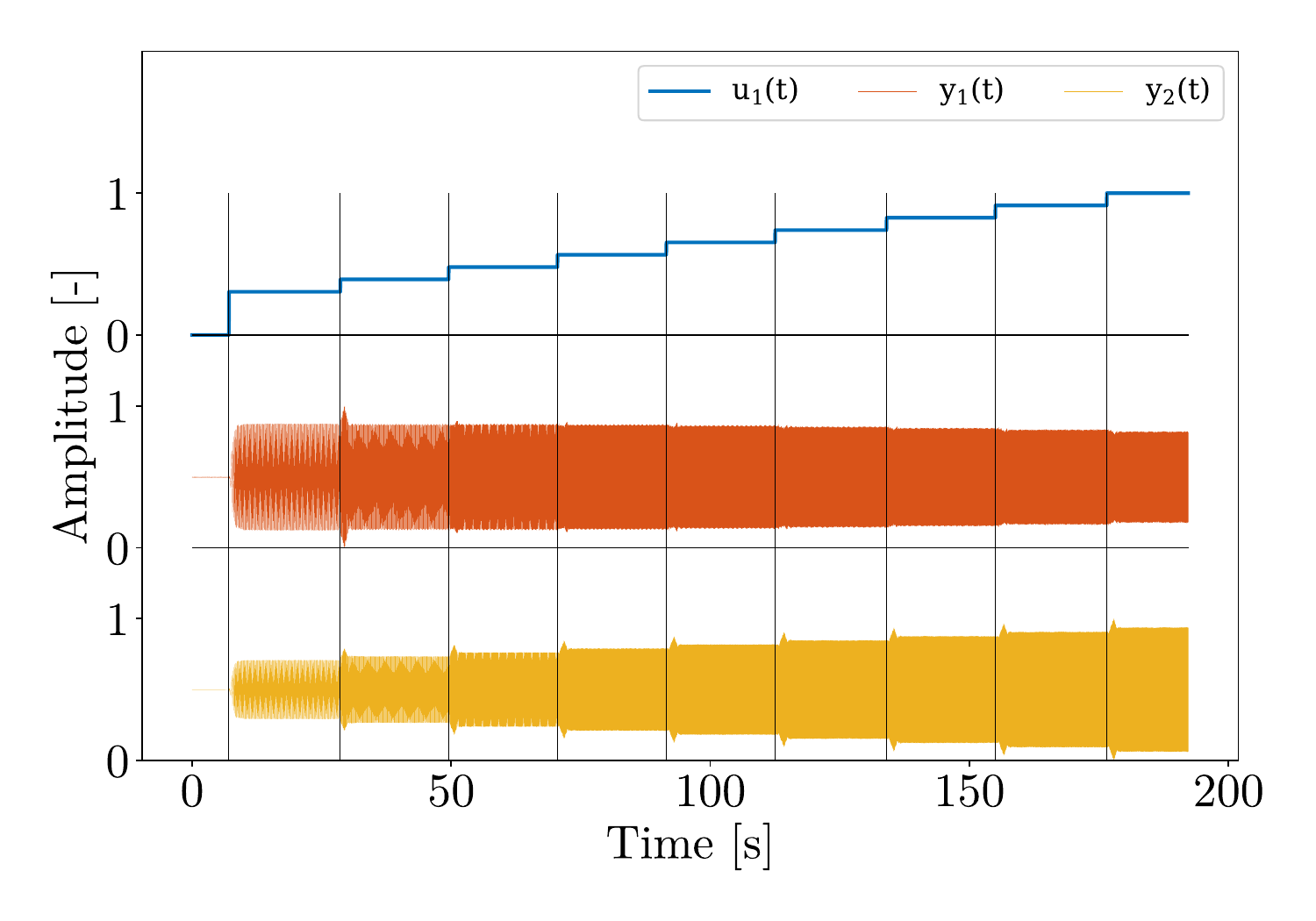}
        \caption{Results of the experimental test on the AC motor system. Input and output signals of the analyzed system.}
        \label{fg:m1}
    \end{figure}
    \begin{figure}[htbp!]
        \includegraphics[width=0.4\textwidth]{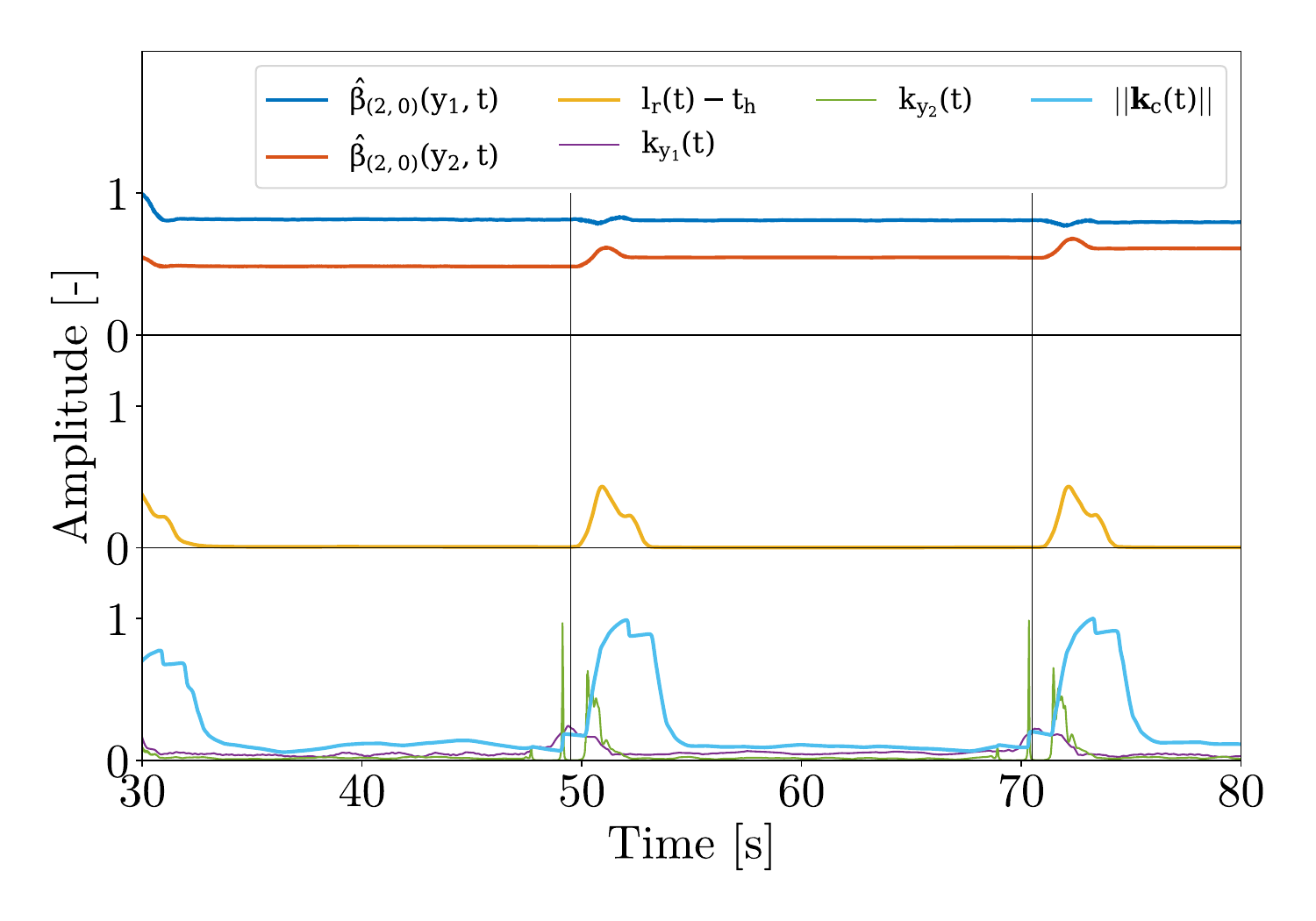}
        \caption{Results of the geometrical SMTR properties of the AC motor system. SMTR curve and its geometrical properties for the studied system.}
        \label{fg:m2}
    \end{figure}
    \begin{figure}[htbp!]
        \includegraphics[width=0.4\textwidth]{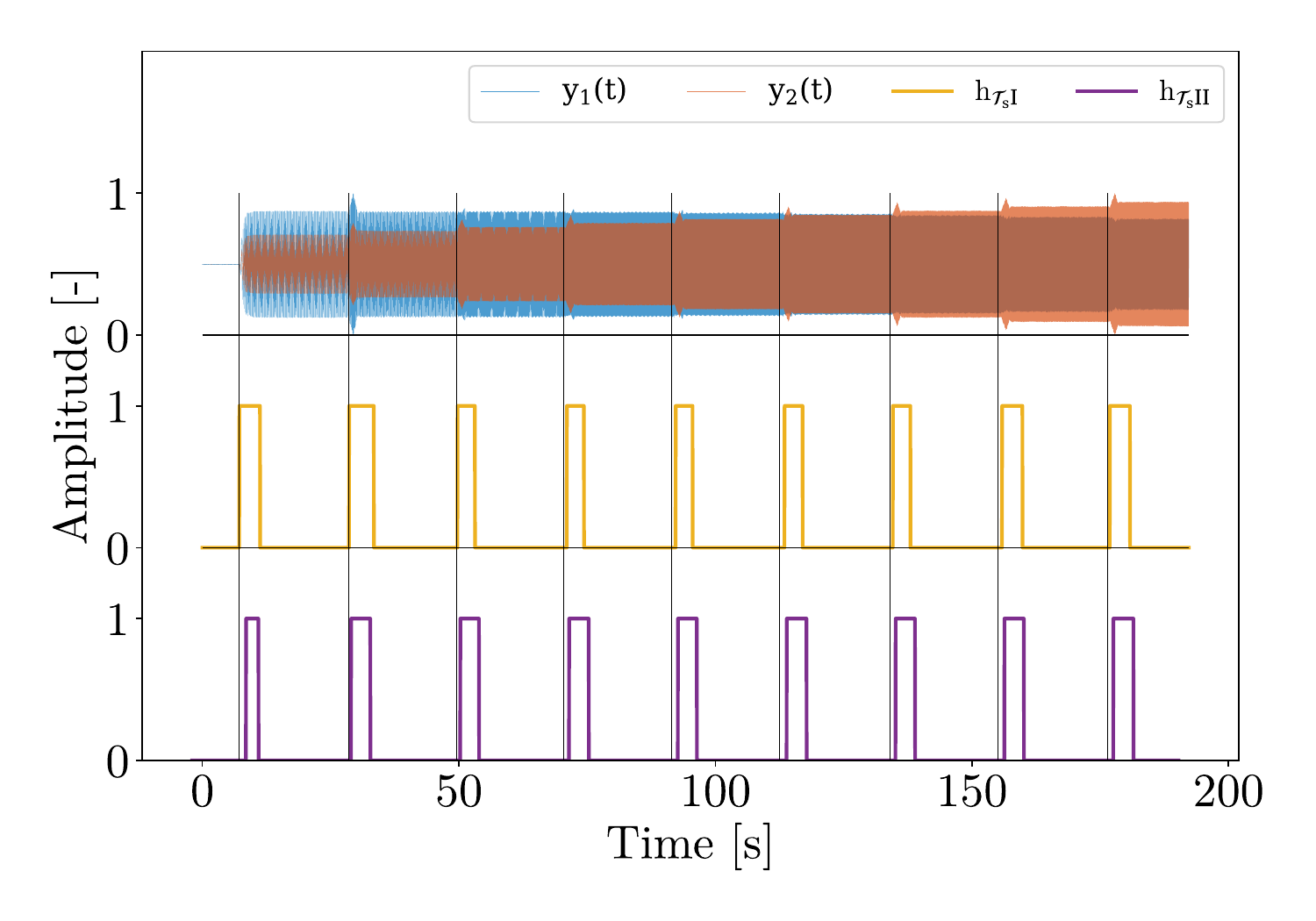}
        \caption{Results of the AC motor system. Regime classifiers using the proposed classifiers.}
    \label{fg:motor_validation}
    \end{figure}

In addition to the electric motor system with controlled speed, a validation was also carried out on a structural system. The schematic of this system is presented in Figure~\ref{fg:structural_system}. The structural system used for the validation of the proposed methodology included three accelerometer sensors that measured the gravitational acceleration of vehicles as they passed over three beams of a vehicular bridge.

The inputs to this system were the unmeasurable impulses generated by the vehicles, and the output variables were the instantaneous accelerations recorded by the sensors, expressed as $\boldsymbol{y}(t) = [a_1(t), a_2(t), a_3(t)]$. These measurements were acquired at a sampling rate of 60 Hz.

The aim of the proposed classifiers was to automatically detect and classify the transient vibrations produced by the passage of vehicles through the structural system.

\begin{figure}[htbp!]
    \centering
        \includegraphics[width=0.4\textwidth]{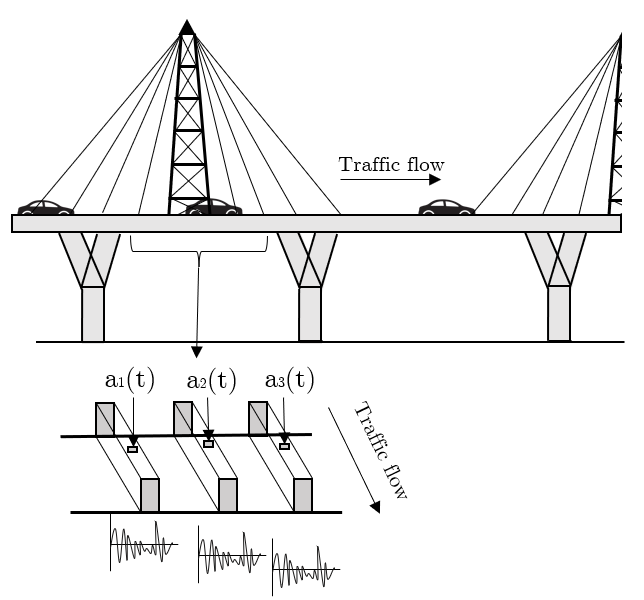}
        \caption{Analyzed structural system.}
    \label{fg:structural_system}
\end{figure}

The mechanical vibration generated by impulses from vehicle traffic over the analyzed structural system is presented in Figure~\ref{fg:s1}. The geometric properties of the SMTR curve are presented in Figure~\ref{fg:s2}. The proposed transient classifiers are presented in Figure~\ref{fg:structure_results}.

The results presented in Figure~\ref{fg:structure_results} clearly separate the impulses generated by vehicles from the random stationary noise. Therefore, this type of classifier can be applied to improve online bridge monitoring by extracting the modal properties of the bridge from automatically classified signal segments that are related to the free response of the structure.

    \begin{figure}[htbp!]
            \centering
        \includegraphics[width=0.4\textwidth]{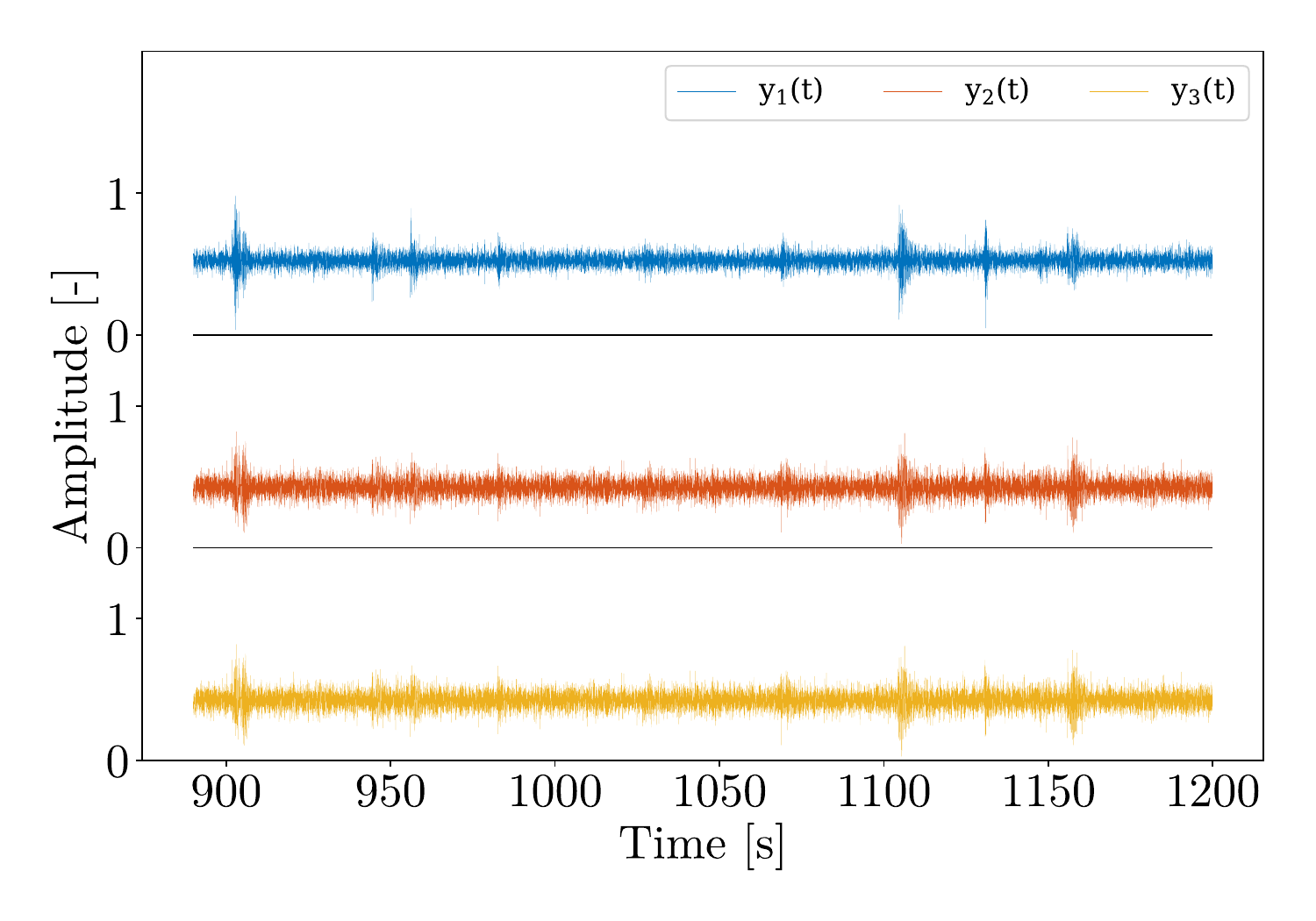}
        \caption{Results of the analyzed structural system. Input and output signals of the studied structural system.}
        \label{fg:s1}
    \end{figure}

    \begin{figure}[htbp!]
            \centering
        \includegraphics[width=0.4\textwidth]{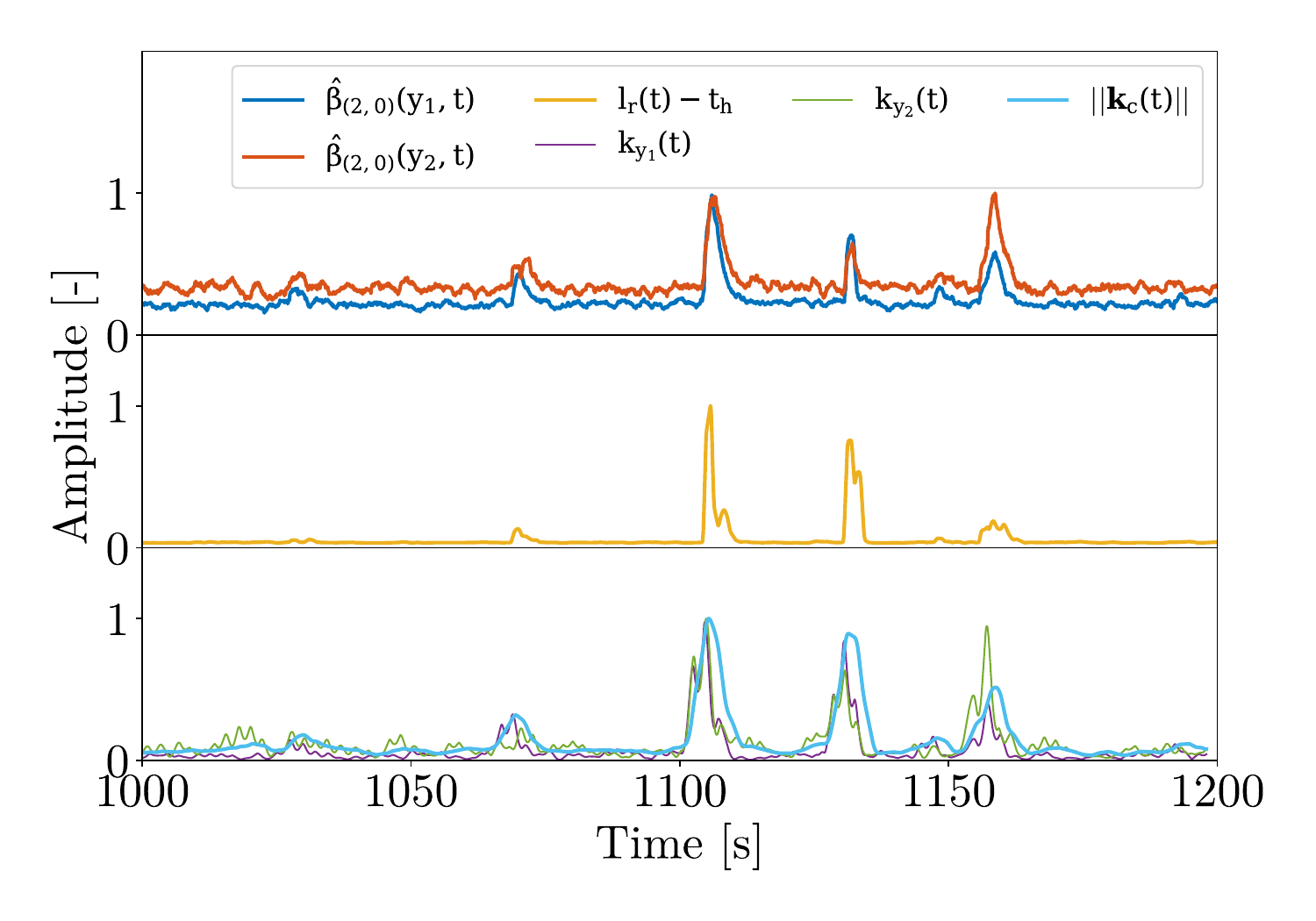}
        \caption{Results of the SMTR curve properties of the studied system. SMTR curve and its geometrical properties for the considered structural system.}
        \label{fg:s2}
    \end{figure}

    \begin{figure}[htbp!]
            \centering
        \includegraphics[width=0.4\textwidth]{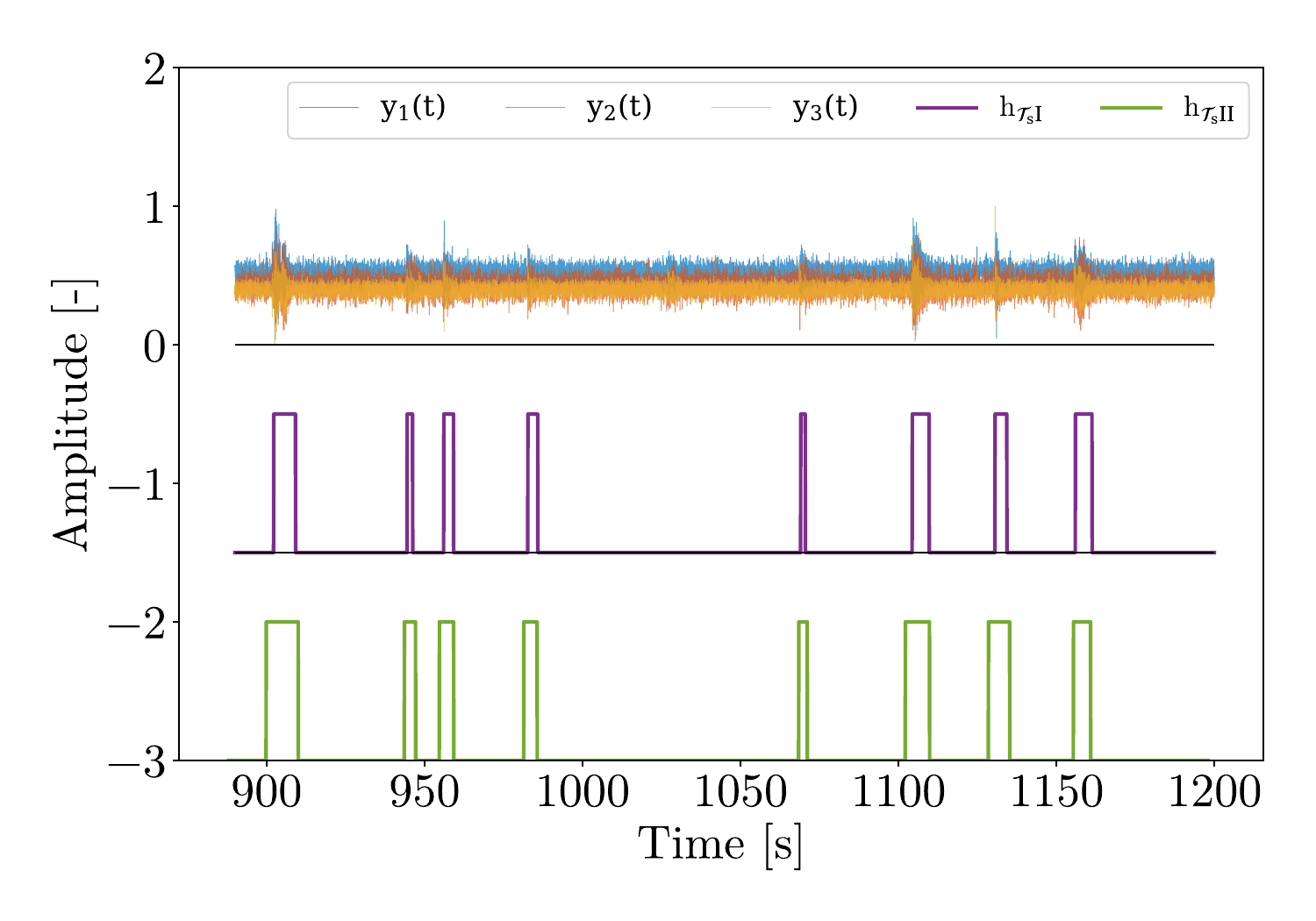}
        \caption{Results of the regime classification of the structural system. Regime classifiers using the proposed classifiers.}
        \label{fg:structure_results}
    \end{figure}

\setcounter{figure}{0}
\setcounter{table}{0}
\setcounter{equation}{0}
\renewcommand{\thetable}{N.\arabic{table}}
\renewcommand{\thefigure}{N.\arabic{figure}}
\renewcommand{\theequation}{N.\arabic{equation}}

\subsection*{\black{ N. Pseudocode}}

The pseudo-algorithm for computing the two proposed transient regime classifiers is presented below.

\begin{algorithm}[!htb]
	\SetAlgoLined
	\KwIn{ $\boldsymbol{y}(t)$, $t_s$, $t_h$, $l_{r0}$.}

	\KwResult{$h_{\mathcal{T}_{sI}}(t)$. }

    Compute the SMTR curve $r_{\mathcal{G}}(t)$ of the monitored system from the measurements of the system $\boldsymbol{y}(t) \in \mathbb{R}^{l_y}$, the acquisition sampling time $t_s$ and the moving time window $t_h$ as:\;

    $ r_{\mathcal{G}}(t) = [t,\hat{\beta}_{(2,\boldsymbol{0})}(y_1,t),\hat{\beta}_{(2,\boldsymbol{0})}(y_2,t),\dots,\hat{\beta}_{(2,\boldsymbol{0})}(y_{l_y},t)]$,\;

    with,\;

    $\hat{\beta}_{(2,\boldsymbol{0})}(y_j,t) = \frac{1}{t_h}\int_{t-t_h}^{t} y_j(\tau) d\tau, j= 1,2,\dots, l_y$.\;

    Estimate the arc length of the SMTR curve inside the analyzed time window $t_h $using  $r_{\mathcal{G}}(t)$ as: \; 

    $l_r(t-t_h,\boldsymbol{r}_{\mathcal{G}}) = \int_{t-t_h}^{t} |\boldsymbol{r}_{\mathcal{G}}^{(1)}(t)|dt = \int_{t-t_h}^{t} \left( 1+\sum_{j=1}^{l_y}(\hat{\beta}^{(1)}_{(2,\boldsymbol{0})}(y_j,t))^2 \right) dt$,\;

    $l_{r\mathcal{T}}(t) = l_r(t-t_h,\boldsymbol{r}_{\mathcal{G}})-t_h$.\;

    Finally, build the classifier $h_{\mathcal{T}_{sI}}(t), t \in (0,t_h)$ inside the analyzed the moving time window $t_h$ with $l_{r0}$ as:\; 

    \If{$l_{r\mathcal{T}}(t) > l_{r0}$}{
    $h_{\mathcal{T}_{sI}}(t)=1$.\;
    }
    \Else{
        $h_{\mathcal{T}_{sI}}(t)=0$.\;
    }
\caption{Arc length-based transient regime classifier}
\label{alg:1}
\end{algorithm}

\begin{algorithm}[!htb]
	\SetAlgoLined
	\KwIn{$\boldsymbol{y}(t)$, $t_s$, $t_h$, $\boldsymbol{k}_{0}$.}
	\KwResult{$h_{\mathcal{T}_{sII}}(t)$. }

    Compute the SMTR curve $r_{\mathcal{G}}(t)$ of the monitored system from the measurements of the system $\boldsymbol{y}(t) \in \mathbb{R}^{l_y}$, the acquisition sampling time $t_s$ and the moving time window $t_h$ as:\;

    $ r_{\mathcal{G}}(t) = [t,\hat{\beta}_{(2,\boldsymbol{0})}(y_1,t),\hat{\beta}_{(2,\boldsymbol{0})}(y_2,t),\dots,\hat{\beta}_{(2,\boldsymbol{0})}(y_{l_y},t)]$,\;

    with\;

    $\hat{\beta}_{(2,\boldsymbol{0})}(y_j,t) = \frac{1}{t_h}\int_{t-t_h}^{t} y_j(\tau) d\tau, j= 1,2,\dots, l_y$.\;

    Estimate the magnitude of the plane curvatures of the SMTR curve inside the analyzed time window $t_h $using  $r_{\mathcal{G}}(t)$ as: \; 

    $\lVert \boldsymbol{k}_{\boldsymbol{y}}(t) \rVert = \sqrt{\sum_{j=1}^{l_y}\left(\frac{\left(|\hat{\beta}_{y_j}^{(2)}(t)|\right)^2}{\left( 1+\left(\hat{\beta}_{y_j}^{(1)}(t)\right)^2\right)^3}\right)}$.\;

    Finally, build the classifier $h_{\mathcal{T}_{sII}}(t), t \in (0,t_h)$ inside the analyzed the moving time window $t_h$ with $\boldsymbol{k}_{0}$ as:\; 

    \If{$\lVert \boldsymbol{k}_{\boldsymbol{y}}(t) \rVert > \boldsymbol{k}_{0}$}{
    $h_{\mathcal{T}_{sII}}(t)=1$.\;
    }
    \Else{
        $h_{\mathcal{T}_{sII}}(t)=0$.\;
    }
\caption{Plane curvature-based transient regime classifier}
\label{alg:2}
\end{algorithm}

It is important to emphasize that the computation of the integrals and derivatives of the signals described in the proposed methodology must be performed using numerical approximations, such as the Stencil method or Savitzky-Golay filters presented in \cite{gorry1990general, kamakoti2010high}, in order to ensure precision in the evaluation of the classifiers. Through the use of finite difference approximations, both the arc length and the plane curvature can be directly derived from the raw measurements collected by the sensors monitoring the system. The pseudo-code for computing the proposed transition classifiers using finite differences is presented below.

\begin{algorithm}[h!]
	\SetAlgoLined
	\KwIn{$\boldsymbol{y}(t)$, $t_s$, $t_h$, $\boldsymbol{k}_{0}$, $l_{r0}$.}
	\KwResult{$h_{\mathcal{T}_{sI}}(t)$, $h_{\mathcal{T}_{sII}}(t)$. }

    Estimate the magnitude of the plane curvatures and arc length of the SMTR curve inside the analyzed time window $t_h $using  $\boldsymbol{y}(t)$ as: \;
      
    $l_r(n) \approx \frac{1}{l_h}\sum_{k=n-l_h}^{n}\sqrt[]{t_h^2+\lVert (\boldsymbol{y}(k))^2-(\boldsymbol{y}(k-\Delta_h))^2\rVert^2},$\;

    $\lVert\boldsymbol{k}_{\boldsymbol{y}}(n)\rVert \approx t_hl_h\sqrt{\sum_{j=1}^{l_y}\left(\frac{\left|\Delta y^2_j(n,l_h) -\Delta y^2_j(n-1,l_h)\right|}{\left( (t_h^2+\Delta y^2_j(n,l_h)^2)\right)^{3/2}}\right)^2},$\;
      
    with, \;

    $ \Delta y^2_j(n,l_h) =  (y_j(n))^2-(y_j(n-l_h-1))^2$, $l_h=\frac{t_h}{t_s}$ and $\Delta_h = l_h + 1$.\;

    Finally, build the classifiers $h_{\mathcal{T}_{sI}}(t)$ and $h_{\mathcal{T}_{sII}}(t),~t \in (0,t_h)$ inside the analyzed the moving time window $t_h$ with $\boldsymbol{k}_{0}$ and $l_{r0}$ as:\; 

    \If{$l_{r\mathcal{T}}(t) > l_{r0}$}{
    $h_{\mathcal{T}_{sI}}(t)=1$.\;
    }
    \Else{
        $h_{\mathcal{T}_{sI}}(t)=0$.\;
    }

    \If{$\lVert \boldsymbol{k}_{\boldsymbol{y}}(t) \rVert > \boldsymbol{k}_{0}$}{
    $h_{\mathcal{T}_{sII}}(t)=1$.\;
    }
    \Else{
        $h_{\mathcal{T}_{sII}}(t)=0$.\;
    }
\caption{Computation of transient regime classifiers using simple finite differences}
\label{alg:3}
\end{algorithm}

Finally, it is worth mentioning that, based on the constructed classifiers, it is possible to detect transient regime events. The value of the classifiers will be equal to 1 in the presence of transient behavior, and equal to 0 in the case of stationary or cyclo-stationary behavior, according to the assumptions established for the proposed methodology.

The implementation of the classifiers was carried out primarily in C++ and Python, utilizing the \texttt{lapack}, \texttt{blas}, and \texttt{numpy} libraries.


\end{document}